\documentclass[aps,pra,twocolumn,nofootinbib, superscriptaddress,10pt]{revtex4-2}

\newif\ifapp
\newcommand{\lsp}{\hspace{0.1em}}

\usepackage{float}
\makeatletter
\let\newfloat\newfloat@ltx
\makeatother

\usepackage{mathtools,mleftright}
\mleftright
\AtBeginDocument{\pagenumbering{arabic}}
\usepackage{romannum}

\usepackage{minitoc}
\usepackage[toc,page,header]{appendix}
\usepackage{physics}
\usepackage{xfrac}
\usepackage{graphicx}
\usepackage{subfigure}

\usepackage{amsfonts,amssymb,amsmath}
\usepackage{amsthm}
\usepackage{array}
\usepackage{mathrsfs}
\usepackage{mathdots}
\usepackage{mathtools}
\usepackage{dsfont}

\usepackage[]{graphics,graphicx,epsfig}
\usepackage{csquotes}
\MakeOuterQuote{"}
\usepackage{epstopdf}
\usepackage{tikz}
\usepackage[inline]{enumitem}
\usepackage{qcircuit}
\usepackage[ruled,linesnumbered]{algorithm2e}
\usepackage{paralist}
\usepackage{diagbox}
\usepackage{tabularx}
\usepackage{setspace}

\usepackage{makecell}
\usepackage{cellspace}
\usepackage{multirow}

\usepackage{hyperref}
\hypersetup{
    colorlinks=true,
    linkcolor=blue,
    filecolor=magenta,      
    urlcolor=cyan,
    citecolor=magenta
}

\def\identity{\leavevmode\hbox{\small1\kern-3.8pt\normalsize1}}

\newtheorem{theorem}{Theorem}

\newtheorem{lemma}{Lemma}

\newtheorem{proposition}{Proposition}

\newtheorem{corollary}{Corollary}

\def\bra#1{\langle #1|}
\def\ket#1{| #1\rangle}

\def\<{\langle}  %% overriding the original command \<
\def\>{\rangle}  %% overriding the original command \>

\def\Tr{\mathrm{Tr}}

\newcommand{\id}{\mathrm{id}}

\newcommand{\5}{\operatorname{\uppercase\expandafter{\romannumeral5}}}
\newcommand{\6}{\operatorname{\uppercase\expandafter{\romannumeral6}}}
\newcommand{\7}{\operatorname{\uppercase\expandafter{\romannumeral7}}}
\newcommand{\8}{\operatorname{\uppercase\expandafter{\romannumeral8}}}
\newcommand{\9}{\operatorname{\uppercase\expandafter{\romannumeral9}}}

\newcommand{\bbE}{\mathbb{E}}
\newcommand{\bbF}{\mathbb{F}}

\newcommand{\bbN}{\mathbb{N}}
\newcommand{\bbR}{\mathbb{R}}

\newcommand{\caD}{\mathcal{D}}
\newcommand{\caE}{\mathcal{E}}

\newcommand{\caH}{\mathcal{H}}

\newcommand{\caL}{\mathcal{L}}

\newcommand{\caN}{\mathcal{N}}
\newcommand{\caO}{\mathcal{O}}

\newcommand{\caQ}{\mathcal{Q}}
\newcommand{\caR}{\mathcal{R}}

\newcommand{\rmU}{\mathrm{U}}

\newcommand{\scrA}{\mathscr{A}}
\newcommand{\scrC}{\mathscr{C}}

\newcommand{\md}[1]{\mathds{#1}}

\newcommand{\vertiii}[1]{{\left\vert\kern-0.25ex\left\vert\kern-0.25ex\left\vert #1 
		\right\vert\kern-0.25ex\right\vert\kern-0.25ex\right\vert}}

\newcommand{\lref}[1]{Lemma~\ref{#1}}

\newcommand{\thref}[1]{Theorem~\ref{#1}}
\newcommand{\thsref}[1]{Theorems~\ref{#1}}
\newcommand{\Thref}[1]{Theorem~\ref{#1}}
\newcommand{\Thsref}[1]{Theorems~\ref{#1}}

\newcommand{\pref}[1]{Proposition~\ref{#1}}

\newcommand{\coref}[1]{Corollary~\ref{#1}}

\newcommand{\cref}[1]{Conjecture~\ref{#1}}

\newcommand{\Cref}[1]{Conjecture~\ref{#1}}

\def\eqref#1{\textup{(\ref{#1})}}
\newcommand{\eref}[1]{Eq.~\textup{(\ref{#1})}}
\newcommand{\eqsref}[2]{Eqs.~(\ref{#1}) and (\ref{#2})}

\newcommand{\sref}[1]{Sec.~\ref{#1}}

\newcommand{\fref}[1]{Fig.~\ref{#1}}

\newcommand{\aref}[1]{Appendix~\ref{#1}}

\newcommand{\subvarM}{\mathcal{V}^{\mathrm{M}}}
\newcommand{\subvarA}{\mathcal{V}^{\mathrm{A}}}

\newcommand{\subinaM}{\Lambda^{\mathrm{M}}}
\newcommand{\subinaA}{\Lambda^{\mathrm{A}}}

\newcommand{\subinaID}{\widetilde{\Lambda}}
\newcommand{\subinaMID}{\widetilde{\Lambda}^{\mathrm{M}}}
\newcommand{\subinaAID}{\widetilde{\Lambda}^{\mathrm{A}}}

\begin{document}

\title{Random approximate quantum information masking}

\author{Xiaodi Li}
\thanks{These authors contributed equally to this work.}
\author{Xinyang Shu}
\thanks{These authors contributed equally to this work.}
\author{Huangjun Zhu}

\date{\today}
	
\affiliation{State Key Laboratory of Surface Physics and Department of Physics, Fudan University, Shanghai 200433, China}
\affiliation{Institute for Nanoelectronic Devices and Quantum Computing, Fudan University, Shanghai 200433, China}
\affiliation{Center for Field Theory and Particle Physics, Fudan University, Shanghai 200433, China}
\begin{abstract}
Masking information into quantum correlations is a cornerstone of many quantum information applications. While there exist the no-hiding and no-masking theorems, approximate quantum information masking (AQIM) offers a promising means of circumventing the constraints. 
Despite its potential, AQIM still remains underexplored, and constructing explicit approximate maskers remains a challenge. 
In this work, we investigate AQIM from multiple perspectives and propose using random isometries to construct approximate maskers. 
First, different notions of AQIM are introduced and we find there are profound intrinsic connections among them. These relationships are characterized by a set of figures of merit, which are introduced to quantify the deviation of AQIM from exact QIM.
We then explore the possibility of realizing AQIM via random isometries in bipartite and multipartite systems. 
In bipartite systems, we identify a fundamental lower bound for a key figure of merit, implying that almost all random isometries fail to realize AQIM. This surprising result generalizes the original no-masking theorem to the no-random-AQIM theorem for bipartite systems.
In contrast, in multipartite systems, we show almost all random isometries can realize AQIM. Remarkably, the number of physical qubits required to randomly mask a single logical qubit scales only linearly.
We further explore the implications of these findings. In particular, we show that, under certain conditions, approximate quantum error correction is equivalent to AQIM. Consequently, AQIM naturally gives rise to approximate quantum error correction codes with constant code rates and exponentially small correction inaccuracies.
Overall, our results establish quantum information masking as a central concept in quantum information theory, bridging diverse notions across multiple domains.
\end{abstract}
\maketitle

\tableofcontents

\section{Introduction}

Hiding or encoding information into quantum correlations (entanglement)  is necessary for quantum information processing against eavesdropping or noise, no matter from the view of quantum cryptography, quantum secure communication or quantum error-correction. Such ideas have scattered in very wide research fields, such as quantum secret sharing \cite{Hillery_1999, Cleve_1999,Terhal_2001, DiVincenzo_2002}, quantum coding theorem \cite{Lloyd_1997, Barnum_1998, horodecki2008quantumcoding, Hayden_2008}, information scrambling \cite{Patrick_Hayden_2007, Yasuhiro_Sekino_2008, Liu_2018} and quantum error-correction \cite{gottesman1997stabilizer, Knill_1997, Terhal_2015_review}. However, there are two no-go theorems, the no-hiding theorem \cite{Braunstein_2007} and the no-masking theorem \cite{Modi_2018}, which exclude the possibility of hiding or masking all information of an arbitrary quantum state into the quantum correlations between two subsystems, in contrast with classical information. 

Fortunately, the no-masking theorem does not forbid the possibility of quantum information masking (QIM) completely. On the one hand, there are certain restricted sets of quantum states which can be masked, like the set of states constituting a hyperdisk \cite{Modi_2018, Liang_2019_complete, Ding_2020_masking} or the set of real quantum states \cite{Zhu_2021_hiding}. On the other hand, masking a general quantum state into a multipartite system is usually possible, even not always possible, because of the results proved in \cite{Li_2018_multipartite, Shi_2021_k_uniform} and the existence of some quantum error-correction codes (QECCs). 
It is shown in \cite{Shi_2021_k_uniform} that the QIM in the multipartite case, called the $k$-uniform QIM, is equivalent to the QECC with distance $k+1$ in some sense.
We also notice that the ideas underlining the eigenstate thermalization hypothesis (ETH) \cite{Deutsch1991_quantum, Srednicki1994_chaos, Popescu2006, Rigol2008, Deutsch_2018} and the thermodynamic code \cite{Brandao2019_quantum, Faist_2020_continuous} are very similar with QIM. 

However, due to the existence of physical noise, perfect masking is difficult to achieve even if a given set of states is maskable. A natural solution is to relax the conditions of the deterministic quantum information masking to consider the probabilistic or approximate quantum information masking (AQIM) \cite{Li_2019_deterministic, Li_2020_probabilistic}. 
Although it is impossible for the existence of a universal probabilistic masking as proved in \cite{Li_2020_probabilistic}, there is still a restricted set of states which can be probabilistic masked \cite{Li_2019_deterministic}.
As for the approximate masking, the authors in \cite{Li_2020_probabilistic} only provide a necessary condition of AQIM and don't provide a method for constructing approximate maskers. 
Hence, constructing approximate maskers still remains a challenge. 

In this paper, motivated by the fact that randomization is a powerful technique for solving various problems \cite{Harrow_2004_superdense, Hayden2004_randomizing, Hastings2009, Aubrun2011, Kong_2022_near_optimal, Brown_2013_short, Brown2015_decoupling, Faist_2020_continuous, Darmawan2024_low_depth, Nelson2025_Fault}, we will focus on realizing the AQIM by Haar random isometries. In particular, we mainly focus on the special case of AQIM aiming at masking all pure states in a Hilbert space.

According to a well-known result that a random pure state in a high-dimensional Hilbert space is very close to a maximally entangled state with high probability \cite{Harrow_2004_superdense, Hayden_2006}, it is reasonable to think that there maybe exist many random isometries which can realize the AQIM in bipartite system.
However, we find that the probability for the existence of a random isometry, which masks a Hilbert space into another bipartite Hilbert, is exponentially small. In other words, there doesn't exist random AQIM in bipartite system. 
As for the multipartite case, we obtain an opposite conclusion and prove that the probability of a random isometry realizing a AQIM approaches $1$ exponentially under certain conditions.  
Random AQIM also has many important implications. 
Considering the equivalence between QIM and QECC \cite{Shi_2021_k_uniform}, there are also deep connections between AQIM and approximate quantum error correction code (AQECC) \cite{Leung_1997_approximate, Schumacher2002_approximate, Crepeau_2005, Klesse_2007, Beny_2010_general, Ng_2010_simple, Beny_2011_approximate, Beny_2011_perturbative, Mandayam_2012, Wang_2020_quasi, Yi2024_complexity}. Combined our results with the results of \cite{Yi2024_complexity}, we can conclude that under some conditions, a random subspace of a large Hilbert space is an AQECC with high probability and such subspace has exponentially small error-correction inaccuracy and constant code rate. 
This result highlights the importance of random QECCs, just as the existing literatures \cite{Faist_2020_continuous, Kong_2022_near_optimal, Brown_2013_short, Brown2015_decoupling, Faist_2020_continuous, Gullans_2021_low_depth, Darmawan2024_low_depth, Nelson2025_Fault}.
In addition to the results about the AQECC, there are also some implications for multipartite quantum entanglement, as the $k$-uniform QIM is closely related to $k$-uniform state \cite{Scott_2004_multipartite, Shi_2021_k_uniform}.

This paper is organized as follows.
In \sref{sec:prel}, we review the basic definition of $k$-uniform QIM and generalize it to the approximate $k$-uniform QIM (or $k$-uniform AQIM) in the maximal version and average version. To characterize the violation of the exact QIM, we introduce several types of figures of merit and dicuss theirships. 
In \sref{sec:aqim_bipartite}, we first explain the notions of random isometry and random AQIM, then we turn to consider the expectation value of the trace distance between two reduced states and find that there exists a finite lower bound for the summation of trace distances, from which we naturally conclude the no random AQIM theorem.
Next, in \sref{sec:aqim_multipartite}, we first consider two concentration results for two scenarios in the bipartite systems, then we generalize the concentration results to the multipartite systems and conclude the existence of random AQIM in the multipartite systems.
In \sref{sec:implications}, we explore the implications of our results for the multipartite entanglement in Sec. \ref{subsec:entanglement} and the AQECC in Sec. \ref{subsec:aqecc}.
Finally, we end this paper with a summary and a discussion of future work in \ref{sec:discussion}.

\section{Preliminaries: exact and approximate quantum information masking}
\label{sec:prel}

Hereafter, we will consider the Hilbert space $\caH_A$ of system $A$ and a subspace $\caH_C$ of the Hilbert space $\caH_{B_1\cdots B_m} = \bigotimes_{i=1}^m \caH_{B_i}$ of a composite system consisting of subsystems $B_1\cdots B_m$ with $m\ge 2$, and their dimensions are denoted as $d_A := \dim(\caH_A)$, $d_C := \dim(\caH_C)$, and $d_{1\cdots m} := \prod_{i=1}^m d_i$, where $d_i := \dim(\caH_{B_i})$. Pure states are denoted by $\ket{\psi}$, $\ket{\phi}$, etc., and $\psi := \ket{\psi} \bra{\psi}$ denotes the corresponding rank-1 projector of the pure state $\ket{\psi}$. Given a subsystem $S$ of the $k$ parties, denote by $S^c$ the complementary subsystem and by $\psi_{S}:=\Tr_{S^c}(\psi)$ the reduced state of $|\psi\>$ for subsystem $S$. A Haar random pure state in Hilbert space $\caH$ is denoted by $\ket{\psi} \sim \mu(\caH)$. When there is no risk of confusion, a unitary $U$ distributed according to the Haar measure on $\mathrm{SU}(d)$ is also denoted by $U \sim \mu(d)$.  The trace distance between two quantum states $\rho_1,\rho_2$ on a given Hilbert space is defined as $D(\rho_1,\rho_2):=\|\rho_1-\rho_2\|_1$. In this paper, $\alpha$, $\delta$, and $\epsilon$  always denote nonnegative real numbers.

\subsection{Exact quantum information masking}

First, we consider bipartite QIM. Let $\scrA\subset \caH_A$ be a set of pure states in $\caH_A$, $V$ an isometry from $\caH_A$ to $\caH_{B_1, B_2}$, and $\scrC=\{V|\varphi\>\, |\, \varphi\in \scrA\}$. Then $V$ is an (exact) \emph{masker} for $\scrA$ if the reduced state of $|\psi\>$ for each party is independent of $|\psi\>\in \scrC$ 
\cite{Modi_2018}. In that case, the information contained in the set $\scrA$ is entirely masked in the entanglement between $B_1$ and $B_2$.  The set $\scrA$ is thus called \emph{maskable}; accordingly, the image set $\scrC$ is also  called maskable. In general,  QIM of the set $\scrA$ means the existence of such a masker.

Next, we generalize bipartite QIM to the multipartite setting \cite{Li_2018_multipartite, Shi_2021_k_uniform}. Now, let $V$ be an isometry from $\caH_A$ to $\caH_{B_1 B_2\cdots B_m}$, $\scrC=\{V|\varphi\>\, |\, \varphi\in \scrA\}$,  and $k$ a positive integer satisfying $1\le k\le \lfloor \frac{m}{2} \rfloor$. Then $V$ is a \emph{$k$-uniform masker} for $\scrA$
if the reduced state of $|\psi\>$ for each  subsystem composed of $k$ parties is independent of $|\psi\>\in \scrC$, that is,
\begin{equation}\label{eq:k_uniform_condition_exact}
    \psi_{S}=\sigma_{S} \quad \forall\, |\psi\>\in \scrC \mbox{ and } \forall  S \mbox{ with } |S|=k,  
\end{equation}
where  $\sigma_S$ for each subsystem $S$ is a fixed state.
In that case,  we cannot distinguish  different image states by measurements on any subsystem composed of  $k$ parties, and the sets $\scrA, \scrC$ are called \emph{$k$-uniformly maskable}. In analogy to the bipartite setting, $k$-uniform masking of the set $\scrA$ means the existence of  a $k$-uniform masker. When the explicit value of $k$ is not important, we just say $\scrA, \scrC$ are maskable for simplicity. Note that the above definition does not change if the condition $|S|=k$ is replaced by $|S|\leq k$.
In addition, a $k$-uniform masker is automatically a $k'$-uniform masker for any positive integer $k'$ with $k'\leq k$. Accordingly, 
$\scrA$ is \emph{$k'$-uniformly maskable} whenever it is \emph{$k$-uniformly maskable}. When  $m=2$, a bipartite masker introduced above is automatically a $1$-uniform masker.

As pointed out in \cite{Shi_2021_k_uniform}, $k$-uniform QIM is closely related to the notion of \emph{$k$-uniform states}. Recall that a multipartite state is $k$-uniform if the reduced state for each subsystem composed of at most $k$ parties is the maximally mixed state \cite{Scott_2004_multipartite}. By definition, it is easy to see that a set of $k$-uniform states is $k$-uniformly maskable.

\subsection{Approximate quantum information masking}
\label{sec:RQIMdef}

In practice, exact masking usually cannot be implemented due to various imperfections, even if a given set of states is maskable. Therefore, it is important to consider the notion of AQIM, first introduced in \cite{Li_2020_probabilistic}.
To this end, we need to introduce suitable figures of merit to quantify the degree of approximation. Here our discussion is based on a given set of image states. 
Let  $\scrC$ be a subset of pure states in $\caH_{B_1\cdots B_m}$, $k\leq m$ a positive integer, and $S$  a subsystem of $k$ parties.

The \emph{maximum variation of $\scrC$ with respect to $S$} is defined as the maximum trace distance between the reduced states of $\scrC$ on subystem $S$, that is,
\begin{align}
\label{Eq:subsystem_masking_variation}
   \subvarM_{S}(\scrC):=\max_{\ket{\psi}, \ket{\psi'}\in \scrC}D(\psi_{S}, \psi'_{S}),
\end{align}
The \emph{maximum subsystem variation} of order $k$ is defined as
\begin{align}
\label{Eq:total_masking_variation}
   \subvarM(\scrC,k) :=  \max_{S:|S|=k}\, \subvarM_{S}(\scrC),
\end{align}
where the argument $k$ can be omitted when $k=1$. 
Recall that the trace distance quantifies the distinguishability between quantum states. If $\subvarM_{S}(\scrC)$ is  small, then it is difficult to distinguish the states in $\scrC$ based on measurements on the subsystem $S$. 
If in addition $\subvarM(\scrC,k)$ is small, then it is difficult to distinguish the states in $\scrC$ based on measurements on any given subsystem composed of at most $k$ parties.

In contrast,  the \emph{average variation of $\scrC$  with respect to $S$} is defined as follows:
\begin{align}
\subvarA_S(\scrC) := \underset{\ket{\psi}, \ket{\psi'}\in \scrC}{\bbE}[ D(\psi_S, \psi'_S) ];
\end{align}
the \emph{average subsystem variation} of order $k$ reads 
\begin{align}
\subvarA(\scrC,k):=\max_{S:|S|=k}\,\subvarA_{S}(\scrC),
\end{align}
where the argument $k$ can be omitted when $k=1$ as before.

Based on $\subvarM(\scrC,k)$ and $\subvarA(\scrC,k)$ we can introduce two versions of approximate  $k$-uniform QIM.  Let $\scrA\subset \caH_A$ be a set of pure states in $\caH_A$, $V$ an isometry from $\caH_A$ to $\caH_{B_1\cdots B_m}$, $\scrC=\{V|\varphi\>\, |\, \varphi\in \scrA\}$. Then $V$ is a  $\delta$-\emph{approximate} $k$-uniform masker with respect to the maximum subsystem variation if 
\begin{align}
	\label{Eq:maximal_inaccuracy_delta}
	\subvarM(\scrC,k)\leq\delta.
\end{align}
By contrast $V$ is a  $\delta$-\emph{approximate} $k$-uniform masker with respect to the average subsystem variation if
\begin{align}
	\label{Eq:maximal_inaccuracy_delta}
	\subvarA(\scrC,k)\leq\delta.
\end{align}
By definition,  a  $\delta$-\emph{approximate} $k$-uniform masker with respect to the maximum subsystem variation is automatically  a  $\delta$-\emph{approximate} $k$-uniform masker with respect to the average subsystem variation. 
In analogy to the exact setting,
$\delta$-\emph{approximate} $k$-uniform masking of the set $\scrA$ means the existence of a $\delta$-\emph{approximate} $k$-uniform masker.
In addition,   a  $\delta$-\emph{approximate} $k$-uniform masker is automatically  a  $\delta$-\emph{approximate} $k'$-uniform masker for $k'\leq k$ because the trace distance is contractive under partial trace \cite{Nielsen_Chuang_2010}.

Next, we introduce several alternative figures of merit that are easier to deal with. Define
\begin{align}
\Pi(\scrC)=\bbE_{|\psi\>\in \scrC}\lsp \psi,\quad 
\Pi_S(\scrC)=\Tr_{S^c} [\Pi(\scrC)]. 
\end{align}
If $\scrC$ denotes the set of all pure states in $\caH_C$, then $\Pi(\scrC)$ coincides with  the normalized projector onto $\caH_C$ and can be abbreviated as $\Pi_C$.
The \emph{maximum subsystem (masking) inaccuracy of order~$k$} is defined as follows:
\begin{align}
	\label{Eq:total_masking_inaccuracy}
	\subinaM(\scrC,k)&:=\max_{S:|S|=k} \,\subinaM_S(\scrC),\\
\subinaM_S(\scrC)&:= \max_{\ket{\psi}\in\scrC}D(\psi_{S},  \Pi_S(\scrC)). \label{Eq:subsystem_masking_inaccuracy}
\end{align}
By contrast, the \emph{average subsystem (masking) inaccuracy of order~$k$} reads
\begin{align}
	\subinaA(\scrC,k)&:= \max_{S:|S|=k} \,\subinaA_S(\scrC),\\
\subinaA_S(\scrC) &:= \underset{\ket{\psi}\in \scrC}{\bbE} D(\psi_S, \Pi_S(\scrC)). \label{Eq:subsystem_masking_inaccuracyA}
\end{align}

Let $\widetilde{\md{1}}_{S}=\mathds{1}_{S} / d_S$ be the maximally mixed state of the subsystem $S$. By replacing $\Pi_S(\scrC)$ in Eqs.~\eqref{Eq:total_masking_inaccuracy}-\eqref{Eq:subsystem_masking_inaccuracyA} with  $\widetilde{\md{1}}_{S}$, we obtain the following variants:
\begin{align}
	\subinaMID(\scrC,k)&:=\max_{S:|S|=k}\,\subinaMID_S(\scrC),  \label{Eq:total_masking_inaccuracyAlt} \\
	\subinaMID_S(\scrC)&:= \max_{\ket{\psi}\in\scrC}D(\psi_{S}, \widetilde{\md{1}}_{S}),\\
		\subinaAID(\scrC,k)&:=\max_{S:|S|=k}\, \subinaAID_{S}(\scrC),\\
		\subinaAID_{S}(\scrC)&:= \underset{\ket{\psi}\in \scrC}{\bbE}D(\psi_{S}, \widetilde{\md{1}}_{S}). \label{Eq:subsystem_masking_inaccuracyAAlt}
\end{align}

The figures of merit introduced above are closely related to each other as shown in the following 
proposition and  proved in \aref{appendix:equivalence_aqim}.
\begin{proposition}\label{prop:equivalence_aqim_inaccuracy_variation_maximal}
	Suppose $\scrC$ is a subset of pure states in $\caH_{B_1\cdots B_m}$, $k\leq m$ is a positive integer, and $S$ is a subsystem of $k$ parties. Then  the maximum subsystem inaccuracy and variation satisfy
	\begin{gather}
		\subinaM_S \le \subvarM_S \le 2\subinaM_S, \quad
		\subinaM \le \subvarM \le 2\subinaM,\\	
		\subvarM_S \le 2\subinaMID_S,\quad \subvarM \le 2\subinaMID, 
	\end{gather}
where the dependence on $\scrC$ and $k$ is suppressed to simplify the notation.	
The average subsystem inaccuracy and variation satisfy	
	\begin{gather}
	\subinaA_S \le \subvarA_S \le 2\subinaA_S, \quad
	\subinaA \le \subvarA \le 2\subinaA,\\
	\subvarA_S \le 2\subinaAID_S,\quad \subvarA \le 2\subinaAID.
\end{gather}
The figures of merit based on the maximum version and the average version satisfy 
	\begin{gather}
	\subinaA_S \le \subinaM_S, \quad \subinaA \le \subinaM,\quad
	\subvarA_S \le \subvarM_S, \quad \subvarA \le \subvarM.
\end{gather}
\end{proposition}

Next, we clarify the relation between approximate QIM and approximate $k$-uniform states. Let   $\ket{\psi} \in \caH_{B_1\cdots B_m}$ and $\scrC=\{|\psi\>\}$. Then the functions $\subinaMID(\scrC,k)=\subinaAID(\scrC,k)$ and $\subinaMID_S(\scrC)=\subinaAID_S(\scrC)$ can be abbreviated as $\subinaID(\psi,k)$ and $\subinaID_S(\psi)$, respectively. The pure state $\ket{\psi}$ is an \emph{$\varepsilon$-approximate} $k$-uniform state if $\subinaID(\psi,k)\leq \epsilon$.  By definition, a set of $\varepsilon$-approximate $k$-uniform states is $\varepsilon$-maskable with respect to $\subinaMID$.

\section{Random approximate quantum information masking in bipartite systems}
\label{sec:aqim_bipartite}

\begin{figure}
    \centering
    \includegraphics[width=1.0\linewidth]{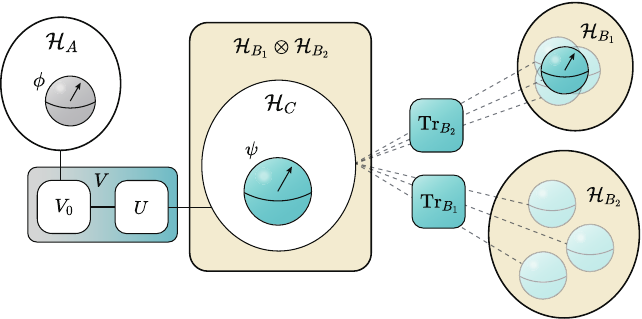}
    \caption{Illustration of AQIM in the bipartite system. The Hilbert space $\caH_A$ is mapped into a subspace $\caH_C$ of $\caH_{B_1B_2}$ by a random isometry $V = UV_0$. When $d_1\ll d_2$, if the states in $\caH_C$ are reduced on subsystem $B_1$, the reduced states maybe concentrate on $\widetilde{\mathds{1}}_{B_1}$ or $\Pi_{C}^{(B_1)}$. However, the reductions of states in $\caH_C$ on subsystem $B_2$ are not very close with each other. 
    }
    \label{fig:bipartite_AQIM}
\end{figure}

In this section we explore the power and limitation of random QIM in bipartite systems. 
Technically, a random masker corresponds to a random isometry of the form $V:\caH_A\rightarrow \caH_{B_1B_2}$, which maps $\caH$ into a subspace $\caH_C$  of dimension $d_C=d_A$ in  $\caH_{B_1,B_2}$. It can be realized by a fixed isometry followed by a random unitary transformation in $\rmU(\caH_{B_1,B_2})$. When $d_C\ll \min\{d_1, d_2\}$, it is known that with high probability all states in $\caH_C$ are nearly maximally entangled \cite{Hayden_2006, Aubrun2011}. Based on this observation, one might expect that a random isometry is a good approximate masker. In sharp contrast with this expectation, we will show that a random isometry is actually far from a good approximate masker, irrespective of $d_1$, $d_2$, and $d_C$. When $d_1\ll d_2$, \fref{fig:bipartite_AQIM} offers an intuitive picture on this result.

\subsection{No random approximate quantum information masking in bipartite systems}
\label{sec:no_random_aqim_bipartite}

In this subsection we show that a random isometry $V$ from $\caH_A$ to a subspace $\caH_C$ in $\caH_{B_1B_2}$ is far from a good approximate masker. To this end we first establish a universal lower bound for the expectation value of the average subsystem variation $\bbE_{\caH_C}\left[ \subvarA(\caH_C)\right]$ and then show that the probability of breaking this lower bound decreases exponentially with the amount of deviation. The proofs of all propositions and theorems in this subsection are relegated to \aref{appendix:no_random_aqim_bipartite}.

The following proposition establishes a simple connection between $\bbE_{\caH_C}[ \subvarA_X(\caH_C)]$ and $\subvarA_X(\caH_{B_1B_2})$, which is very helpful for understanding the limitation of random QIM. 
\begin{proposition}
    \label{prop:expectation_trace_distance_of_subspace_and_states_in_subspace}
     Suppose $\caH_C$ is a random subspace of $\caH_{B_1B_2}$ with dimension $d_C$. Then the average subsystem variation $\subvarA_X$ with respect to $X\in\{B_1,B_2\}$ satisfys
\begin{align}
    \label{eq:expectation_trace_distance_of_subspace_and_states_in_subspace}
   \underset{\caH_C}{\bbE}\,\left[ \subvarA_X(\caH_C) \right]= 
    \frac{(2d_C-2)(2d_{12}-1)}{(2d_C-1)(2d_{12}-2)} \subvarA_X(\caH_{B_1B_2}).
\end{align}
\end{proposition}

Thanks to  \pref{prop:equivalence_aqim_inaccuracy_variation_maximal},
the average subsystem  variation $\subvarA_X(\caH_{B_1B_2})$ is lower bounded by the average subsystem  inaccuracy $\subinaA_{X}(\caH_{B_1B_2})=\subinaAID_{X}(\caH_{B_1B_2})$. By virtue of \pref{prop:expectation_trace_distance_of_subspace_and_states_in_subspace} and this observation we can establish the following theorem.

\begin{theorem}
    \label{theo:expectation_subspace_sum_trace_distance_lower_bound}
    Suppose $\caH_C$ is a random subspace of $\caH_{B_1B_2}$ that has dimension $d_C$. Then the average subsystem variation $\subvarA(\caH_C)$ satisfies
\begin{align}
    \label{eq:expectation_subspace_sum_trace_distance_lower_bound}
    \underset{\caH_C}{\bbE}\,[\subvarA(\caH_C)]\geq \frac{1}{2}\,\underset{\caH_C}{\bbE}\left[\subvarA_{B_1}(\caH_C)+\subvarA_{B_2}(\caH_C)\right]  \ge w,
\end{align}
where 
    \begin{align}
    w =& \frac{1}{6}\frac{(2d_C-2)(2d_{12}-1)}{(2d_C-1)(2d_{12}-2)}> \frac{1}{9}.
    \end{align}
\end{theorem}

\begin{figure}
    \centering
    \includegraphics[width=1.0\linewidth]{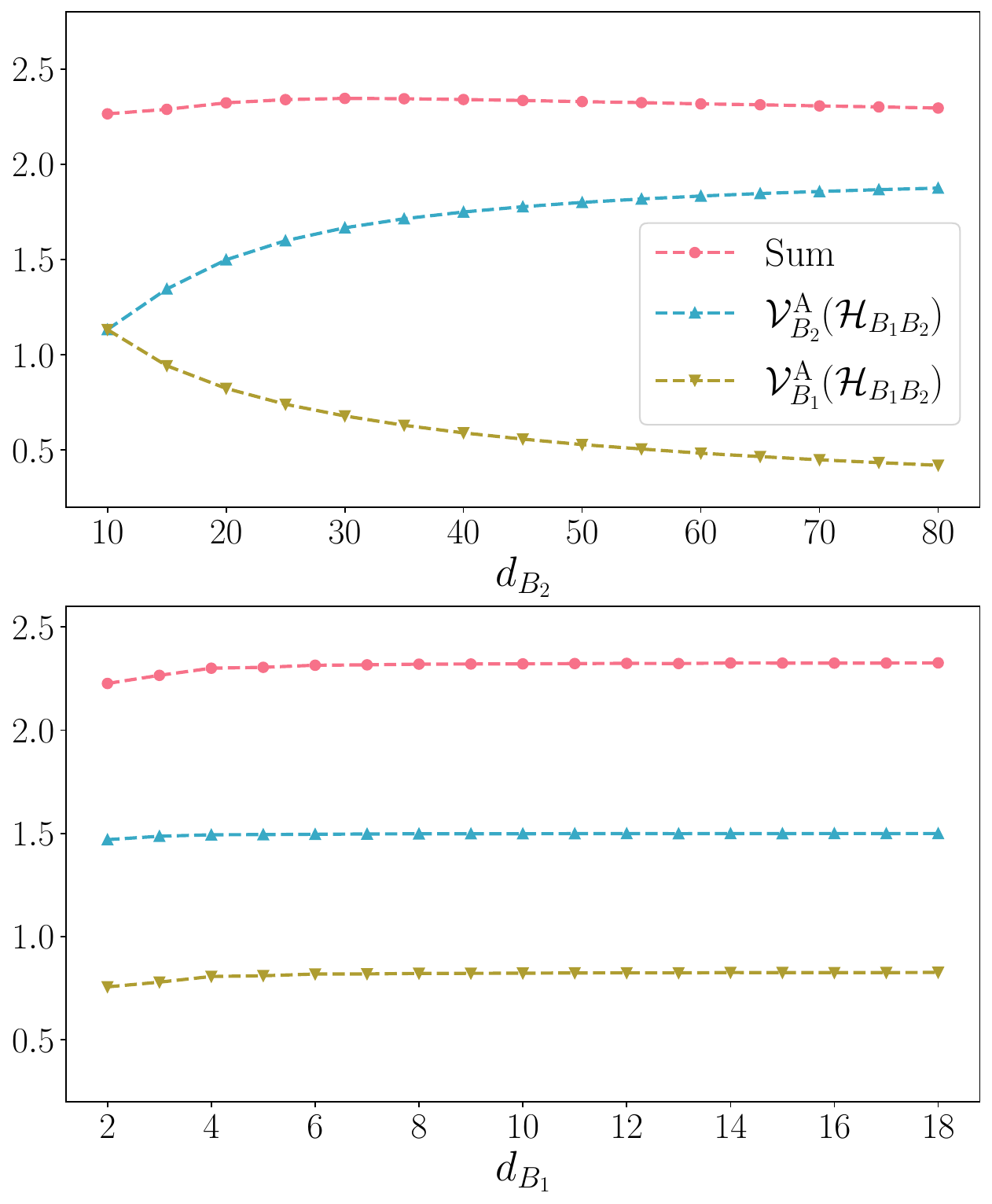}
    \caption{Numerical illustration of the behaviors of $\subvarA_{B_1}(\caH_{B_1B_2})$, $\subvarA_{B_2}(\caH_{B_1B_2})$ and their summation. In above subfigure, we let $d_{B_1}=10$ and $d_{B_2}$ increase. When $d_{B_2}=d_{B_1}=10$, $\subvarA_{B_1}(\caH_{B_1B_2})=\subvarA_{B_2}(\caH_{B_1B_2})$. When $d_{B_2}$ increases, $\subvarA_{B_1}(\caH_{B_1B_2})$ monotonically increases and approaches $2$, $\subvarA_{B_2}(\caH_{B_1B_2})$ monotonically decreases and approaches $0$, and the summation of them first increases and then decreases. In below subfigure, we set $d_{B_2}/d_{B_1}=2$ and let $d_{B_1}$ increase. Three expectation values only increase a little and then approach some constants.}
    \label{fig:lower_bound_sum_expectation}
\end{figure}
\Thref{theo:expectation_subspace_sum_trace_distance_lower_bound} shows that there is an inevitable tradeoff between $\bbE_{\caH_C}[\subvarA_{B_1}(\caH_C)]$ and $\bbE_{\caH_C}[\subvarA_{B_2}(\caH_C)]$. Both terms cannot be made small simultaneously irrespective of the choices of $d_1$, $d_2$, and $d_C$, as illustrated in \fref{fig:lower_bound_sum_expectation}. When $d_C, d_1, d_2\gg 1$, the lower bound $w$ approaches $1/6$.

Next, we further show that the probability that $\subvarA(\caH_C)$ violates the lower bond $w$ decreases exponentially with the amount of deviation. 
\begin{theorem}
    \label{theo:lower_bound_expectation_sum_subspace_bipartite}
    Suppose $\caH_C$ is a random subspace of $\caH_{B_1B_2}$ that has dimension $d_C$. Then, for $\alpha>0$, the average subsystem variation $\subvarA(\caH_C)$ satisfies
    \begin{align}
        &\Pr\left\{ \subvarA(\caH_C)< w-\alpha\right\} 
        \le \exp\left(-d_{12}\alpha^2/16 \right).
    \end{align}
 \end{theorem}
 
\Thsref{theo:expectation_subspace_sum_trace_distance_lower_bound} and \ref{theo:lower_bound_expectation_sum_subspace_bipartite}  show that a random isometry from $\caH_A$ to a subspace in $\caH_{B_1B_2}$ is not a good approximate masker even if the dimensions  $d_1$ and $d_2$ are arbitarily large.

\begin{theorem}
\label{theo:no_randomly_aqim}
Suppose $V$ is a random isometry from the Hilbert space $\caH_A$  to the bipartite Hilbert space $\caH_{B_1B_2}$. Then the probability that $V$ is a $\delta$-approximate masker with $\delta=1/9$ and with respect to the maximum or average subsystem variation is exponentially small.
\end{theorem}

\subsection{Concentration results on bipartite systems}
\label{sec:concentration_bipartite_identity_projector}

In this subsection, we prove several concentration results on a bipartite system, which show that approximate information masking from one subsystem is possible if the dimension of the other subsystem is sufficiently larger. These results complement previous results on generic entanglement as established in \cite{Hayden_2006, Aubrun2011}. Moreover, these results will provide solid foundations for exploring  AQIM in multipartite systems. The proofs of two theorems and one proposition in this subsection are relegated to \aref{appendix:concentration_bipartite_identity_projector}.

The following theorem provides a concentration result on the subsystem inaccuracy with respect to a subsystem maximally mixed state.
\begin{theorem}
\label{theo:trace_identity_bipartite}
Suppose $\mathcal{H}_C$ is a random subspace of $\caH_{B_1B_2}$ of dimension $d_C$. Then for $\alpha>0$, the maximum subsystem inaccuracy $\subinaMID_{B_1}(\caH_C)$ with respect to $B_1$ satisfy
\begin{align}
\label{Eq:bipartite_identity_concentration}
\Pr& \left\{\subinaMID_{B_1}(\caH_C)
>r+\alpha\right\} \le 2\left(\frac{10}{\alpha}\right)^{2d_C}\exp\left(-\frac{d_{12}\alpha^2}{72\pi^3\mathrm{ln}2}\right),
\end{align}
where $r=\sqrt{(d_1^2-1)/(d_{12}+1)}$.
\end{theorem} 

\begin{figure}
    \centering
    \includegraphics[width=1.0\linewidth]{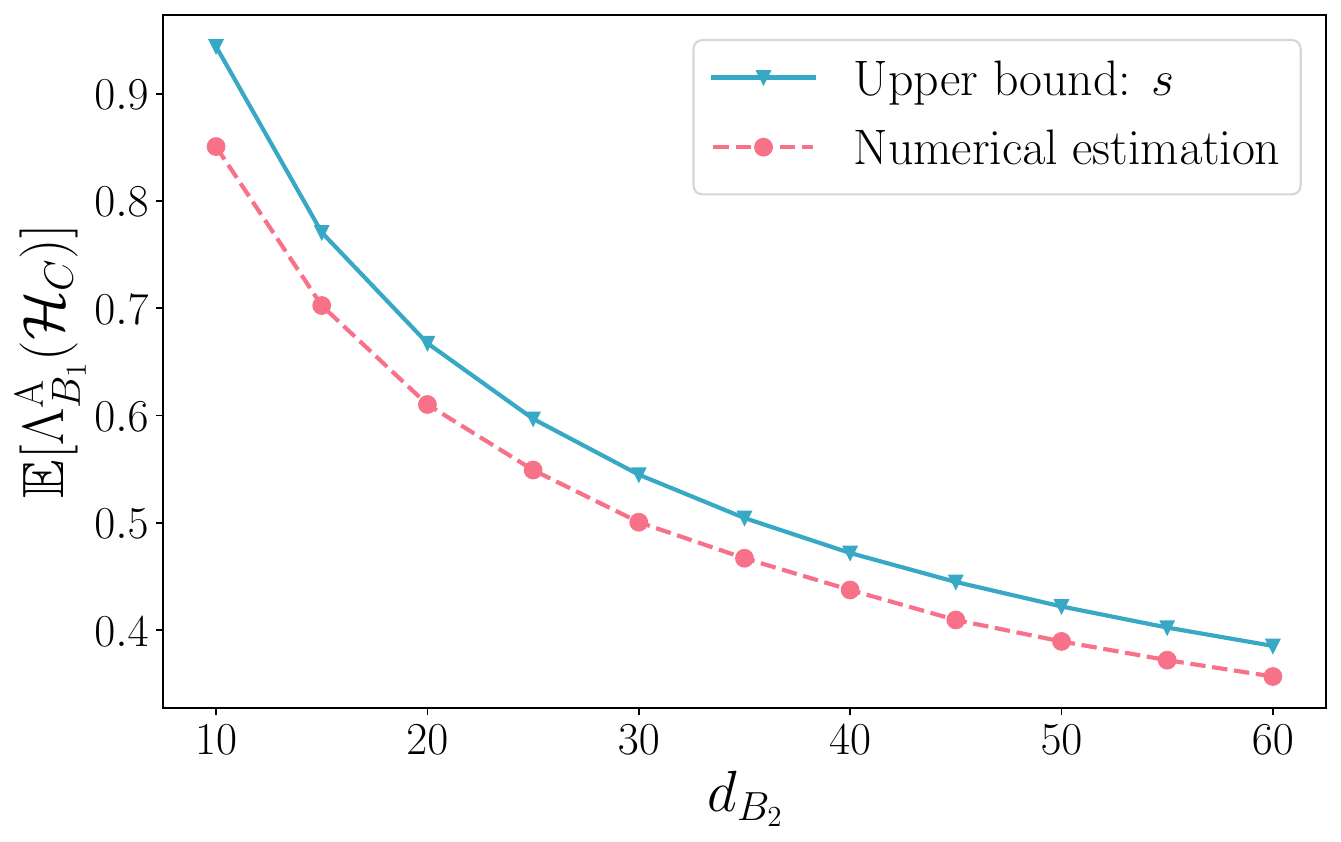}
    \caption{The numerical estimation of $\bbE[\subinaA_{B_1}(\caH_C)]$ (the red dashed line), and its analytical upper bound $s$ in \eref{Eq:projector_expectation_upper_bound_beta} (the blue solid line) for the case $d_1=10$ and $d_C=10$.  Both of them become smaller and approach $0$ when $d_2$ approaches infinity. }
    \label{fig:bipartite_projector_upper_bound}
\end{figure}

In the proof of the above theorem, we also show that $\subinaAID_{B_1}(\caH_{B_1B_2})$ is upper bounded by $r\approx \sqrt{d_1/d_2}$. 
Requiring the probability on the right-hand side of \eref{Eq:bipartite_identity_concentration} to be smaller than $1$, we get the following inequality: 
\begin{align}
\label{eq:bipartite_dc_upper_bound_mixed}
    \quad d_C < \frac{1}{2\ln(10/\alpha)} \left[\frac{d_{12}\alpha^2}{72\pi^3\ln2} -\ln2\right].
\end{align}
Hence, the physical interpretation of \thref{theo:trace_identity_bipartite} is that there exists a subspace $\caH_C$ of $\caH_{B_1B_2}$ of dimension $d_C$ satisfying \eref{eq:bipartite_dc_upper_bound_mixed}, such that the subsystem masking inaccuracy $\subinaAID_{B_1}(\caH_{C})$ is upper bounded by $r + \alpha$. Notably, the probability for $\caH_C$ not satisfying the condition $\subinaMID_{B_1}(\caH_C) < 2r$ is exponentially small.
Essentially, \thref{theo:trace_identity_bipartite} is equivalent to Theorem \Romannum{4}.1 and Corollary \Romannum{4}.2 of \cite{Hayden_2006}, which collectively establish that with high probability a random subspace contains only near-maximally entangled states as long as the dimension $d_C$ is not too large. Nevertheless, the upper bound for $d_C$ specified in \eref{eq:bipartite_dc_upper_bound_mixed} is simpler and exhibits better dependence on the subsystem dimensions.

Next, we establish the counterpart of \thref{theo:trace_identity_bipartite} for the subsystem  inaccuracy $\subinaM_{B_1}(\caH_C)$; see \eqsref{Eq:total_masking_inaccuracy}{Eq:total_masking_inaccuracyAlt}.
\begin{theorem}
\label{theo:trace_projector_bipartite}
Suppose $\mathcal{H}_C$ is a random subspace of $\caH_{B_1B_2}$ of dimension $d_C$. Then for $\alpha>0$, the maximum subsystem inaccuracy $\subinaM_{B_1}(\caH_C)$ with respect to $B_1$ satisfies 
\begin{align}
\Pr&\left\{\subinaM_{B_1}(\caH_C) > s + \alpha\right\} \leq \left(\frac{10}{\alpha}\right)^{2d_C}\exp\left(-\frac{d_{12}\alpha^2}{256}\right),
\end{align}
where
\begin{align}
\label{Eq:projector_expectation_upper_bound_beta}
    s=\sqrt{\frac{d_C-1}{d_C}\frac{d_{12}(d_1^2-1)}{d_{12}^2-1}}. 
\end{align}
\end{theorem}

The parameter $s$ is an upper bound for the expectation value $\bbE_{\caH_C} [\subinaA_{B_1}(\caH_C)]$, as illustrated in \fref{fig:bipartite_projector_upper_bound}.
If $d_1, d_2, d_C$ are very large, then $s\approx r $ and we get a similar concentration result as \thref{theo:trace_identity_bipartite}. 
If $1\ll d_1\ll d_2$ and $\alpha\leq s$, then there exists a subspace $\caH_C$ of $\caH_{B_1B_2}$ of dimension
\begin{align}
    \quad d_C <  \frac{d_{12}\alpha^2}{512\ln(10/\alpha)}
\end{align}
such that $\subinaM_{B_1}(\caH_C)$ is upper bounded by $2s$, and the probability for $\caH_C$ not satisfying this property is exponentially small.

To complement \thref{theo:trace_projector_bipartite}, next,  we analyze the properties of $\subinaA_{B_1}(\caH_C)$
by virtue of  random matrix theory \cite{mehta2004random}. 
Suppose a pure state $\ket{\psi}$ is sampled uniformly from $\caH_C$; then we can view $\Delta\psi_{B_1}:=\psi_{B_1}-\Pi_{C}^{(B_1)}$ as a random matrix with expectation value $\bbE_{\ket{\psi}\sim \mu(\caH_C)}\left[ \Delta\psi_{B_1} \right]=0$. 
According to a similar argument presented in \cite{Zhu2011_quantum}, the expectation value of the trace norm of $\Delta\psi_{B_1}$ can be approximated as follows:
\begin{align}
	\subinaA_{B_1}(\caH_C)
	&=\underset{\ket{\psi}\sim \mu(\caH_C)}{\bbE}[\|\Delta\psi_{B_1}\|_1] \notag \\
	&\approx \frac{4}{3\pi}\sqrt{d_1 \Tr[(\Delta\psi_{B_1})^2]}. 
\end{align}
This equation in turn implies that 
\begin{align}
	\underset{\caH_C}{\bbE}  [\subinaA_{B_1}(\caH_C)] 
	&\lessapprox  \frac{4}{3\pi}\sqrt{ d_1}\left( \underset{\caH_C}{\bbE} \Tr[(\Delta\psi_{B_1})^2]\right)^{1/2} \notag\\
	&=\frac{4}{3\pi} s.
\end{align}
Note that this approximate upper bound has the same order of magnitude as the parameter $s$ featured in \thref{theo:trace_projector_bipartite}.

Although the mathematical meanings of \thsref{theo:trace_identity_bipartite} and \ref{theo:trace_projector_bipartite} are similar, their physical interpretations are quite different. From \thref{theo:trace_projector_bipartite} itself, we cannot conclude that a random subspace contains only near-maximally entangled states with high probability, since $\Pi_C^{(B_1)}$ is not a completely mixed state.
Nevertheless, the deviation of  $\Pi_{C}^{(B_1)}$ from $\widetilde{\md{1}}_{B_1}$ is usual very small as shown in the following proposition.

\begin{proposition}
\label{prop:trace_projector_identity_bipartite}
Suppose $\caH_C$ is a random subspace of $\caH_{B_1B_2}$ of dimension $d_C$. Then, for $\alpha>0$,
\begin{align}
\Pr\left\{ D( \Pi_{C}^{(B_1)}, \widetilde{\md{1}}_{B_1}) > t + \alpha\right\} 
\leq \exp\left(-\frac{d_{12}\alpha^2}{16}\right), \notag
\end{align}
where
\begin{align}
    t=\sqrt{\frac{d_{12}-d_C}{d_C}\frac{d_1^2-1}{d_{12}^2-1}}. 
\end{align}
\end{proposition}
Here the parameter $t$ is also an upper bound for the expectation value, that is, 
$\bbE_{\mathcal{H}_C}[D(\Pi_{C}^{(B_1)}, \widetilde{\md{1}}_{B_1})] \leq t$. When $1\ll d_1 \ll d_2$, we have the approximation
\begin{align}
    t\approx
    \sqrt{\frac{d_1}{d_2}\left(\frac{1}{d_C}-\frac{1}{d_{12}}\right)} \le \sqrt{\frac{d_1}{d_2d_C}},
\end{align}
which is very small. So the average trace distance between 
$\Pi_{C}^{(B_1)}$ and $\widetilde{\md{1}}_{B_1}$ is very small, and the probability of large deviation is exponentially small.

\section{Random approximate quantum information masking in multipartite systems}
\label{sec:aqim_multipartite}

Although a random isometry is in general not a good approximate masker in the bipartite setting,  the situation 
is completely different in the multipartite setting, as we shall see shortly. In this section, we first generalize the concentration results in \sref{sec:concentration_bipartite_identity_projector} from the bipartite scenario to the multipartite scenario and then discuss some significant implications.

First, it is straightforward to generalize \thref{theo:trace_identity_bipartite} to the multipartite setting, in which the composite system consists of $m$ parties with the Hilbert space $\caH_{B_1 \cdots B_m}$. 
To avoid cumbersome expressions, without loss of generality, here we assume that all subsystems have the same dimension, i.e., $d_i = d\ge 2$ for $1 \leq i \leq m$.

\begin{theorem}
\label{theo:subpsace_identity_multipartite}Suppose $\caH_C$ is a random subspace of the Hilbert space $\caH_{B_1\cdots B_m}$ of dimension $d_C$. Then, for any $\alpha>0$, the maximum subsystem inaccuracy $\subinaMID(\caH_C,k)$ satisfies
   \begin{align}
   \label{Eq:prob_multipartite_idenity}
       \Pr&\left\{\subinaMID(\caH_C,k)
     \le d^{k-m/2}+ \alpha \right\} \notag\\
       &\ge 1- 2\binom{m}{k}\left(\frac{10}{\alpha}\right)^{2d_C}\exp\left(-\frac{d^m\alpha^2}{72\pi^3\ln2} \right).
   \end{align}
\end{theorem}

\begin{proof}
    This result follows immediately from \thref{theo:trace_identity_bipartite} and the union bound on all $\binom{m}{k}$ bipartite cuts, which split the whole system into two subsystems of $k$ parties and  $m-k$ parties, respectively.
\end{proof}

By virtue of the union bound employed above we can also 
generalize \thref{theo:trace_projector_bipartite} to the multipartite setting.
\begin{theorem}
\label{theo:multipartite_trace_projector}Suppose $\caH_C$ is a random subspace of the Hilbert space $\caH_{B_1 \cdots B_m}$ of dimension $d_C$. Then, for any $\alpha > 0$, the maximum subsystem inaccuracy $\subinaM(\caH_C,k)$ satisfies
\begin{align}
\Pr&\left\{ \subinaM(\caH_C,k) \le u + \alpha \right\} \notag\\
&\ge 1- \binom{m}{k} \left( \frac{10}{\alpha} \right)^{2d_C} \exp\left( -\frac{d^m \alpha^2}{256} \right),
\label{eq:multipartite_trace_bound}
\end{align}
where
\begin{align}
u = \sqrt{\frac{d_C - 1}{d_C} \cdot \frac{d^m (d^{2k} - 1)}{d^{2m} - 1}}.
\end{align}
\end{theorem}

Next, we consider  the trace distance between $\Pi_{C}^{(S)}$ and $\widetilde{\md{1}}_{S}$ and generalize \pref{prop:trace_projector_identity_bipartite} to the multipartite setting following a similar proof as presented above.
\begin{proposition}
\label{prop:trace_projector_identity_multipartite}
Suppose $\caH_C$ is a random subspace of the Hilbert space $\caH_{B_1 \cdots B_m}$ of dimension $d_C$. Then, for any $\alpha > 0$, 
\begin{align}
\Pr&\left\{\max_{S:|S|=k} D\bigl( \Pi_{C}^{(S)}, \widetilde{\md{1}}_{S}\bigr) \le t + \alpha\right\} \notag\\
&\ge 1- \binom{m}{k} \exp\left(-\frac{d^m\alpha^2}{16}\right), 
\end{align}
where
\begin{align}
    t=\sqrt{\frac{d^m-d_C}{d_C}\frac{d^{2k}-1}{d^{2m}-1}}. 
\end{align}
\end{proposition}

If  $d^{m/2-k}\gg 1$, then 
$t\le d^{k-m/2}/\sqrt{d_C}$ is very small, so  the trace distance $D\bigl(\Pi_{C}^{(S)}, \widetilde{\md{1}}_{S}\bigr)$ for a random subspace $\caH_C$ is very small except for an exponentially small probability. Accordingly, the predictions of \thsref{theo:subpsace_identity_multipartite} and \ref{theo:multipartite_trace_projector} are quite similar. Here we shall discuss the consequences of \thref{theo:subpsace_identity_multipartite}.

First, we analyze the asymptotic consequences of \thref{theo:subpsace_identity_multipartite}.
Requiring the probability in \eref{Eq:prob_multipartite_idenity} being positive and using the inequality $\binom{m}{k}\le 2^{m H(k/m)}\le 2^m$, where $H(p)=-p\log _2p-(1-p)\log_2(1-p)$ is the binary entropy function, we get an upper bound for $d_C$ as follows,
\begin{align}
    d_C < \frac{1}{2\ln(10/\alpha)} \bigg\{\frac{d^m\alpha^2}{72\pi^3\ln2} -(\ln2) (m +1)\bigg\} .  
    \label{eq:identity_subspace_dimension_inequality}
\end{align}
Further, if $1\ll k\ll m$ and $\alpha=\caO(d^{k-m/2})$, then the maximum subsystem inaccuracy $\subinaMID(\caH_C,k)$ satisfies  $\subinaMID(\caH_C,k)\le 2d^{k-m/2}$ except for an exponentially small probability. 
If instead $k<m/2$ and $m\gg 1$, then $\subinaMID(\caH_C,k)\le 2d^{k/2-m/4}$ except for an exponentially small probability.

In conjunction with the definitions in \sref{sec:RQIMdef},
we can derive the following main results on  random AQIM in multipartite systems. 

\begin{corollary}
\label{cor:aqim_identity_multipartite}
Suppose $V:\caH_A\rightarrow \caH_{B_1\cdots B_m}$ is a random isometry with the image subspace $\caH_C$.  Then $V$ is a good approximate $k$-uniform masker except for an exponentially small probability if one of the following two conditions holds: \\
(1)  $d_C=\caO(d^{2k})$ and $1\ll k\ll m$; in this case the maximum subsystem inaccuracy is $2d^{k-m/2}$. \\
(2)  $d_C=\caO(d^{k+m/2})$, $m\gg 1$, and $k<m/2$; in this case the maximum subsystem inaccuracy is $2d^{k/2-m/4}$.
\end{corollary}

Second, we analyze the non-asymptotic behaviors of \thref{theo:subpsace_identity_multipartite}, as detailed in \aref{app:consequences_multipartite}, by addressing the following question: Given a fixed subspace dimension $d_C$ and a fixed masking inaccuracy $\subinaMID=d^{k-m/2}+\alpha$, what conditions must $k$ and $m$ satisfy? 
The following corollary summarizes the main results in three cases depending on the relative magnitudes of $\alpha$ and $d^{k-m/2}$. 

\begin{corollary}
\label{cor:aqim_multipartite_mask_rate}
Suppose $V:\caH_A\rightarrow \caH_{B_1\cdots B_m}$ is a random isometry with the image subspace $\caH_C$. Then $V$ is an approximate  $k$-uniform masker with maximum subsystem inaccuracy $\subinaMID(\caH_C,k) = \alpha+d^{k-m/2}$ if $m$ is larger than the threshold $m^*$ specified below in three cases separately. \\
(1) If $\alpha$ is a fixed small constant but larger than $d^{k-m/2}$, then 
\begin{align}
m^* \propto \frac{\ln d_C}{\ln d}.
\end{align}
(2) If $\alpha=d^{k/2-m/4}$ is larger than $d^{k-m/2}$ (depending on $k,m$), then 
\begin{align}
m^* \propto \frac{\ln d_C}{(\zeta +1/2)\ln d},
\end{align}
where $\zeta = k/m$.\\
(3) If $\alpha = d^{k-m/2}$ is smaller than $d^{k-m/2}$, then
\begin{align}
m^* \propto \frac{\ln d_C}{2\zeta \ln d} .
\end{align}
\end{corollary}
More details can be found in \aref{app:consequences_multipartite}. The results above indicate that the minimum number of parties, $m^*$, is proportional to the number of qubits $l$ when $d_C=2^l$, and the proportion coefficient depends on the parameter $\alpha$ (which is related to the masking inaccuracy) and the ratio $\zeta$. Notably, the resources required to mask $l$ qubits scale linearly with $l$ rather than exponentially.

\section{Implications}
\label{sec:implications}

\subsection{Multipartite entanglement and approximate $k$-uniform states}
\label{subsec:entanglement}

A $k$-uniform masker can mask the quantum information contained in a set of states into the quantum correlations of a multipartite system completely and thus has a close connection with multipartite entanglement. Not surprisingly, our results in the previous section have significant implications for multipartite entanglement, including  $k$-uniform states in particular \cite{Scott_2004_multipartite}. As the multipartite analogue of the maximally entangled states in a bipartite system,  $k$-uniform states encompass locally maximally entangled (LME) states \cite{Bryan2019locallymaximally} and  the absolutely maximally entangled (AME) states \cite{helwig2013absolutelymaximallyentangledstates}. For simplicity, here we  assume again that all subsystems have the same dimension $d$.

Our result on approximate $k$-uniform states is summarized in the following corollary.
\begin{corollary}
\label{cor:appr_k_uniform_state}
Suppose $\ket{\psi}$ is a Haar random pure state in $\caH_{B_1\cdots B_m}$ and $k\leq m$ is a positive integer. Then, for $\alpha\ge0$, 
   \begin{align}
   \label{Eq:appr_k_uniform_state}
       \Pr&\left\{\subinaID(\psi,k) \le \alpha+d^{k-m/2} \right\} \notag\\
       &\ge 1-2\binom{m}{k} \exp\left(-\frac{d^m\alpha^2}{18\pi^3\ln2} \right).
   \end{align}
\end{corollary}
\begin{proof}
    This result follows immediately from \lref{lem:concentrate_of_identity_function} and the union bound on all $\binom{m}{k}$ bipartite cuts, which split the whole system into two subsystems of $k$ parties and $m-k$ parties, respectively.
\end{proof}
Corollary \ref{cor:appr_k_uniform_state} has a similar interpretation as \thref{theo:subpsace_identity_multipartite}: if $m>2k\gg 1$ and $\alpha\le d^{k-m/2}$, then a Haar random pure state is a $2d^{k-m/2}$-approximate $k$-uniform state with high probability. However, as $k$ approaches $\lfloor m/2\rfloor$, the deviation of $|\psi\>$ from a $k$-uniform state gets larger and larger.

Next, we turn to the generalized Meyer-Wallach (GMW) entanglement measure \cite{Scott_2004_multipartite}, which is defined as follows:
\begin{align}
\caQ_k(\ket{\psi}):= \frac{d^{k}}{d^{k}-1} \left[ 1 - \frac{k!(m-k)!}{m!} \sum_{S}\Tr(\psi_S^2) \right].
\end{align}
According to \coref{cor:appr_k_uniform_state} and the following fact:
\begin{align}
   \Tr(\psi_S^2)-\frac{1}{d_S}= \left\|\psi_S-\widetilde{\md{1}}_{S} \right\|_2^2  \le \left\|\psi_S-\widetilde{\md{1}}_{S} \right\|_1^2 , \notag
\end{align}
the GMW entanglement measure of a Haar random pure state is bounded from below by
\begin{align}
    \caQ_k(\ket{\psi})\ge 1-(4d^{2k-m}+d^{-k})\approx 1-4d^{2k-m},
\end{align} 
with probability approaching $1$.
If $k\ll m/2$, then $\caQ_m(\ket{\psi})\approx 1$ with probability approaching $1$. This result corroborates the fact that  a Haar random pure state of a multipartite system is very likely to be a highly entangled state \cite{Hayden_2006, aubrun2017alice}.

\subsection{Approximate quantum error-correction codes}
\label{subsec:aqecc}

According to  \cite{Shi_2021_k_uniform}, a subspace of a multipartite Hilbert space is $k$-uniformly maskable if and only if it is a QECC with code distance $k+1$. Not surprisingly,  AQIM is also  related to  AQECC \cite{Yi2024_complexity}. In this subsection, we shall first review the relation between AQECC and AQIM and  then show that  random AQIM can induce AQECC under certain conditions.

Intuitively speaking, an AQECC is a QECC in which the encoded logical information cannot be recovered  perfectly after the system undergoes a noise channel. 
To define an AQECC rigorously, we need to find a metric to characterize its deviation from an exact QECC. A common choice is the \emph{completely bounded purified distance} \cite{Tomamichel2010_duality, Yi2024_complexity}. 
Given the definition of fidelity between two states $\rho$ and $\sigma$ as $F(\rho,\sigma):=\|\sqrt{\rho}\sqrt{\sigma}\|_1$, the worst-case entanglement fidelity between two channels $\mathcal{A}$ and $\mathcal{B}$ is defined as follows:
\begin{align}
    \mathcal{F}(\mathcal{A},\mathcal{B}):=\min_\rho F((\mathcal{A}\otimes\mathrm{id})(\rho),(\mathcal{B}\otimes\mathrm{id})(\rho)),
\end{align}
where $\id$ denotes the identity channel on an extended system and the optimization runs over all input states on the joint system.  Then the completely bounded purified distance between $\mathcal{A}$ and $\mathcal{B}$ is defined as\footnote{Note that the completely bounded purified distance, following the convention taken in \cite{Yi2024_complexity}, is a bit different from the Bures distance in \cite{Beny_2010_general}.}
\begin{align}
    \mathfrak{D}(\mathcal{A},\mathcal{B})=\sqrt{1-\mathcal{F}(\mathcal{A},\mathcal{B})^2}\,.
\end{align}
When the states in a logical Hilbert space $\caH_L$ are encoded into a subspace $\caH_C$ of a physical system $\caH_{B_1\cdots B_m}$ via an encoding map $\mathcal{E}$, we say that the code is \emph{$\widetilde{\eta}$-correctable} under a given noise channel $\mathcal{N}$ if there exists a decoding channel $\mathcal{D}$ such that
\begin{align}
    \mathfrak{D}(\mathcal{D}\circ\mathcal{N}\circ\mathcal{E},\id_L)\leq \widetilde{\eta}.
\end{align}

Corresponding to an AQECC $\caH_C$ with the encoding channel $\caE$ and noise channel $\caN$, the quantum error-correction (QEC) \emph{inaccuracy} is defined as 
\begin{align}
\widetilde{\eta}(\caE, \caN):=\min_{\caD} \mathfrak{D}(\mathcal{D}\circ\mathcal{N}\circ\mathcal{E},\id_L).
\end{align} 
Then, in order to characterize the intrinsic properties of an AQECC $\caH_C$, which are independent of the noise, the authors of \cite{Yi2024_complexity} defined the \emph{overall subsystem variance} as 
\begin{align}
\eta(\caH_C, k):= \max_{S: |S|\le k} \max_{|\psi\>\in \caH_C} \left\| \psi_S-\Pi_{C}^{(S)} \right\|_1,
\end{align}
When the noise channel $\caN$ is a replacement channel $\mathcal{R}_S$ of the form $\mathcal{R}_S(\psi)=\mathrm{Tr}_S(\psi)\otimes \gamma_{S}$ with $\gamma_{S}$ being a fixed state of the subsystem $S$, the authors of \cite{Yi2024_complexity} proved an inequality between the QEC inaccuracy $\widetilde{\eta}(\caE, \caN)$ and the subsystem variance $\eta(\caH_C, k)$ as shown in the following lemma. 
\begin{lemma}
\label{lem:two_sides_bound_of_AQEC_inaccuracy}
For the replacement channel $\caR_S$, the QEC inaccuracy $\widetilde{\eta}(\caE, \caR_S)$ and the subsystem variance $\eta(\caH_C, k)$ satisfy the following inequalities:
\begin{align}
    \frac{1}{4}\eta(\caH_C, k) \leq \widetilde{\eta}(\caE, \caR_S)
    \leq \sqrt{d_C\eta(\caH_C, k)},
\end{align}
where $d_C$ is the dimension of the code space $\caH_C$.
\end{lemma}

Note that the definition of the subsystem variance $\eta(\caH_C, k)$ is almost the same as the maximum subsystem inaccuracy in \eref{Eq:total_masking_inaccuracy}. Since the trace distance is contractive,  we can derive the following equivalent definition:
\begin{align}
    \eta(\caH_C, k)
    =\max_{S: |S|= k} \max_{\ket{\psi}\in \caH_C} \left\| \psi_S-\Pi_{C}^{(S)} \right\|_1.
\end{align}
Now, it is clear that the maximum subsystem inaccuracy $\subinaM(\caH_C,k)$ is equal to the subsystem variance $\eta(\caH_C, k)$. By virtue of \lref{lem:two_sides_bound_of_AQEC_inaccuracy} and this observation we can derive the following  relation between AQECC and AQIM.

\begin{corollary}
\label{cor:aqim_aqecc}
    Suppose $\caH_C$ is a subspace of $\caH_{B_1\cdots B_m}$. If $\caH_C$ is $k$-uniformly maskable with inaccuracy $\subinaM(\caH_C,k)$, then it is also a $[\sqrt{d_C\subinaM(\caH_C,k)}\lsp]$-correctable AQECC being capable of correcting replacement errors on any subsystem with no more than $k$ parties. Conversely, if $\caH_C$ is a $\widetilde{\eta}$-correctable AQECC for replacement errors acting on any subsystem with no more than $k$ parties, then it is also $k$-uniformly maskable with inaccuracy $4\widetilde{\eta}$.
\end{corollary}

Note that the equivalence between $k$-uniform masking and QECC with distance $k+1$ proved in \cite{Shi_2021_k_uniform} is a natural consequence of the above Corollary with both inaccuracies set to $0$.

Combing \coref{cor:aqim_aqecc} and \thref{theo:multipartite_trace_projector}, we can derive the following result on random codes.
\begin{theorem}
\label{theo:random_codes}
Suppose $\caH_C$ is a random subspace of a multipartite Hilbert space $\caH_{B_1\cdots B_m}$ with $d_i=d$ for $1\leq i\leq m$ and its dimension $d_C$ satisfies
\begin{align}
    d_C < \frac{1}{2\ln(10/\alpha)} \left[\frac{d^m\alpha^2}{256} -(\ln2)m H\left(\frac{k}{m}\right) \right].\label{eq:random_codes_ineq_of_dc}
\end{align}
Then for $\alpha\ge 0$, the probability that $\caH_C$ is an AQECC with inaccuracy $\widetilde{\eta}(\caE, \caR_S)=\sqrt{d_C(u+\alpha)}$ and distance $k+1$, being capable of correcting replacement errors on any subsystem with no more than $k$ parties, is bounded from below by
\begin{align}
\label{Eq:random_code_probability}
   1- \binom{m}{k} \left( \frac{10}{\alpha} \right)^{2d_C} \exp\left( -\frac{d^m \alpha^2}{256} \right).
\end{align}
\end{theorem}

\Thref{theo:random_codes} shows that random  AQIM can induce  AQECCs.  In \aref{appendix:Consequences_of_random_codes} we consider two cases in more detail: in the first case, the QEC inaccuracy $\widetilde{\eta}(\caE, \caR_S)$ is a fixed small constant, while in the second case, the QEC inaccuracy $\widetilde{\eta}(\caE, \caR_S)$ scales exponentially as $d^{a(k-m/2)}$ with $0<a<1/2$.

First, suppose $\gamma=k/m<1/2$ and the QEC inaccuracy $\widetilde{\eta}(\caE, \caR_S)$ is a fixed small number, when $d_C=d^l\gg m \gg 1 $, the minimum of $m$ satisfying \eref{eq:random_codes_ineq_of_dc} scales as 
\begin{align}
m^*\propto\max\left\{3l, \frac{2l}{1-2\gamma} \right\}.
\end{align} 
This result implies that, to encode a logical Hilbert space of $l$ qubits, the minimum number $m^*$ of physical qubits required scales linearly with $l$, and the code rate $l/m^*$ is a constant, assuming that $\gamma$ is fixed. More specifically: we have $l/m^*\propto (1-2\gamma)/2$ when $1/6<\gamma<1/2$ and $l/m^*\propto 1/3$ when $\gamma\leq1/6$.

Second, we require that the QEC inaccuracy decreases exponentially as $\widetilde{\eta}(\caE, \caR_S)=d^{am(\gamma-1/2)}$, since $u\approx d^{k-m/2}$ and $k/m=\gamma$. In this case, the minimum number of physical qubits required satisfies
\begin{align}
m^*\propto\max\left\{3l, \frac{2l}{(1-2a)(1-2\gamma)} \right\}.
\end{align} 
If $1/2> a\ge 1/6$, then $2/((1-2\gamma)(1-2a))$ is always larger than $3$, and the code rate is always smaller than $1/3$. If instead  $a<1/6$, then the code rate becomes a piecewise function of $\gamma$ with the maximum being $1/3$. 

In summary, the above result suggests that a random subspace  from a multipartite Hilbert space with sufficiently many parties can serve as an AQECC, with both the failure probability and inaccuracy of the code vanishing exponentially.

\section{Discussions}
\label{sec:discussion}

In this work, we generalize the notion of $k$-uniform QIM to $k$-uniform AQIM. In particular, we introduce two versions of AQIM: the maximal version and the average version, and some figures of merit to characterize the deviation of AQIM from QIM. 
Next, we employs random isometries to construct approximate quantum maskers, and evaluate its feasibility and efficiency for both bipartite and multipartite systems.
In the bipartite case, we find that the average subsystem variation admits a finite lower bound and the probability for a random subspace whose average subsystem variation being smaller than such a lower bound is exponential small, thus we conclude there does not exist AQIM in bipartite systems.
While in the multipartite case, we prove that a random isometry can approximately and $k$-uniformly mask a space of states into a multipartite space with high probability. Both results are based some concentration results about random subspaces of a Hilbert space. 
Since QIM has deep relationships with multipartite entanglement and QECC, our results naturally have important implications for them.
For the multipartite entanglement, our results imply that a random pure state of a multipartite system is an approximate $k$-uniform state with probability exponentially approaching $1$ and a random subspace is a highly entangled subspace with high probability. 
Honestly speaking, these implications are not surprising, since there have been many similar results \cite{Hayden_2006,aubrun2017alice}.
However, an astonishing implication is that about the AQECC, i.e., a randomly masked subspace can be a AQECC with high probability under certain conditions, the inaccuracy of this subspace can be exponentially small and its code rate is constant at the same time. 

Our work also reveals some new problems. 
First, the most direct question is how to realize random AQIM in experiments. Since it is hard to realize the Haar random unitary experimentally,  then we can ask if the unitary $k$-designs can serve as the random masker.
Second, although it is obvious that the states of a maskable are entangled states, we don't know if there exists some necessary conditions about the entanglement of states for this set to be maskable. 
Thirdly, our results about AQIM and AQECC highlight the importance of studying AQECC and the random code. There have been some explicit AQECCs, like the covariant code \cite{Faist_2020_continuous, Kong_2022_near_optimal, Wang_2020_quasi}. But as a quasi-code, the covariant code's QEC inaccuracy doesn't approach $0$ quickly enough \cite{wang2024_universal}, then it may be necessary to construct other AQECCs. And a promising candidate may be the random code, as the existing results \cite{Kong_2022_near_optimal, Gullans_2021_low_depth, liu2025approximatequantumerrorcorrection}. 
Fourthly, although we noticed the relation between the AQIM and the ETH and the thermodynamic code, the connections between AQIM and the ETH and thermalization deserve to be explored further.

\section*{Acknowledgments}
We thank Zhiyao Yang, Datong Chen for helpful discussions. This work is supported by Shanghai Science and Technology Innovation Action Plan (Grant No.~24LZ1400200), Shanghai Municipal Science and Technology Major Project (Grant No.~2019SHZDZX01), and the National Key Research and Development Program of China (Grant No.~2022YFA1404204).

\emph{Note added.} Recently, we became aware of a related work by Guo and Shi \emph{et al.} Their result shows that  approximate $k$-uniform states can be constructed with
high probability from  Haar random pure states and shallow random quantum circuits with respect to the  purities of reduce states. In contrast, our work focuses on AQIM implemented by random isometries with respect 
to the trace distance, which has a clear optional interpretation.

\bibliography{ref}

\appendix
\clearpage
\onecolumngrid

\section{Contents of Appendix and list of expectation values}
For convenience, we outline the contents of the sections in Appendix and provide a table summarizing the important results about the expectation values over random subspaces or random states.
In \aref{appendix:no_random_aqim_bipartite}, we provide the detailed proofs for the main results mentioned in \sref{sec:no_random_aqim_bipartite}.
In \aref{appendix:concentration_bipartite_identity_projector}, we prove the concentration results in the bipartite system shown in \sref{sec:concentration_bipartite_identity_projector}.
In \aref{app:consequences_multipartite}, we analyze the implications of \thref{theo:subpsace_identity_multipartite} in details, resulting in \coref{cor:aqim_identity_multipartite} and \coref{cor:aqim_multipartite_mask_rate}.
In \aref{appendix:Consequences_of_random_codes}, we explore the consequences of \thref{theo:random_codes}, which reveal important results about the code rate and QEC inaccuracy of random QECCs.

Finally, for readability, we list the values or bounds of these expectation values appearing in this paper.
In the following, we denote a random subspace $\caH_C$ sampled from the uniform distribution as $\caH_C \sim \mathrm{Gr}(\caH_{B_1B_2},d_C)$, where the Grassmannian $\mathrm{Gr}(\caH_{B_1B_2},d_C)$ is defined as the set of all subspaces of $\caH_{B_1B_2}$ with dimension $d_C$ and endowed with an unitarily invariant measure.
Other notations are the same as mentioned in the main body. 
\\

\setlength{\cellspacetoplimit}{5pt}   % 设置单元格顶部间距
\setlength{\cellspacebottomlimit}{5pt} % 设置单元格底部间距

% 定义新列类型 C（居中对齐）和 L（左对齐），并指定宽度
\newcolumntype{C}[1]{>{\centering\arraybackslash}p{#1}}
\newcolumntype{T}[1]{>{\centering\arraybackslash}m{#1}}

\begin{center}
	\begin{tabular}{|c|c|c|}
		\hline
		Distribution of variables & Expectation values & Proofs \\
		\hline     
		$\ket{\psi}\sim \mu(\caH_{B_1B_2})$
		& $\subinaAID_{B_1}(\caH_{B_1B_2}) \le \sqrt{\frac{d_1^2-1}{d_{12}+1}}$ 
		& \makecell{\aref{app:concentration_psi_identity}:  \\\eref{eq:proof_expectation_psi_identity_upper_bound}} \\
		\cline{2-3}
		&$\underset{\ket{\psi}}{\bbE}
		\left[\left\|\psi_{B_1}-\widetilde{\md{1}}_{B_1}\right\|_2^2\right]=\frac{d_1+d_2}{d_{12}+1}-\frac{1}{d_1}$ 
		& \makecell{\aref{app:concentration_psi_identity}:  \\\eref{eq:proof_expectation_psi_identity_2norm}} \\
		\hline     
		$\caH_C\sim \mathrm{Gr}(\caH_{B_1B_2},d_C)$
		& $\underset{\caH_C}{\bbE}[D( \Pi_{C}^{(B_1)}, \widetilde{\md{1}}_{B_1})] \le  \sqrt{\frac{d_{12}-d_C}{d_C}\frac{d_1^2-1}{d_{12}^2-1}}$
		&  \makecell{\aref{appendix:trace_identity_projector}:\\\eref{eq:proof_expectation_pi_identity_upper_bound}} \\  
		\cline{2-3}
		& $\underset{\caH_C}{\bbE}\left[\left\| \Pi_{C}^{(B_1)}, \widetilde{\md{1}}_{B_1}\right\|_2^2\right] =  \frac{d_{12}-d_C}{d_1d_C}\frac{d_1^2-1}{d_{12}^2-1}$
		& \makecell{\aref{appendix:trace_identity_projector}:\\ \eref{eq:proof_expectation_pi_identity_2norm}} \\  
		\hline    
		\makecell{$\caH_C\sim \mathrm{Gr}(\caH_{B_1B_2},d_C)$ \\ $\ket{\psi}\sim \mu(\caH_C)$}
		& $\underset{\caH_C}{\bbE}[\subinaA_{B_1}(\caH_C)]\le \sqrt{\frac{d_C-1}{d_C}\frac{d_{12}(d_1^2-1)}{d_{12}^2-1}}$
		& \makecell{\aref{app:projector_trace_distance}:\\ \eref{eq:proof_expectation_psi_pi_upper_bound}} \\ 
		\cline{2-3}
		& $\underset{\caH_C,\ket{\psi}}{\bbE}\left[\left\| \psi_{{B_1}}, \Pi_{C}^{(B_1)}\right\|_2^2\right]= \frac{d_C-1}{d_C}\frac{d_2(d_1^2-1)}{d_{12}^2-1}$
		& \makecell{\aref{app:projector_trace_distance}:\\ \eref{eq:proof_expectation_psi_pi_2norm}} \\ 
		\hline  
		$\ket{\psi},\ket{\phi}\sim \mu(\caH_{B_1B_2})$
		&$\subvarA_{B_2}(\caH_{B_1B_2}) \ge 2-2\sqrt{\frac{d_1}{d_2}}$
		&\makecell{\aref{app:Concentration_results_in_B1B2}:\\\eref{eq:proof_expectation_psi_phi_Hb1b2_lower_bound2}} \\
		% \cline{2-3}
		%     &$\subinaAID_{X}
		%     \le \subvarA_X(\caH_{B_1B_2})
		%     \le 2\subinaAID_{X}, \quad X\in\{B_1,B_2\}$
		%     &\makecell{\sref{sec:no_random_aqim_bipartite}:\\ \eref{eq:average_over_2states_in_whole_space_bound_by_idenity}}\\
		% \cline{2-3}
		%     &$\begin{aligned}
			%         \underset{\ket{\psi},\ket{\phi}}{\bbE} &\left[F(\psi_{B_1}, \phi_{B_1}) \right]\\[-0.2em]
			%         =&\frac{2}{[(d_{12})_{1/2}]^2} \sum_{k=1}^{d_1} \frac{(-1)^{d_1-k}(k)_{1/2}[(k+d_2-d_1)_{1/2}]^2}{\Gamma(d_1-k+1)\Gamma(k-d_1+1/2)}
			%     \end{aligned}$
		%     &\cite{PhysRevA.104.022438}\\
		% \cline{2-3}
		%     &$\underset{\ket{\psi},\ket{\phi}}{\bbE}\left[ D(\psi_{B_1},\phi_{B_1}) \right] \ge 2-2\underset{\ket{\psi},\ket{\phi} \sim \mu(\caH_{B_1B_2})}{\bbE} \left[F(\psi_{B_1}, \phi_{B_1}) \right]$
		%     & \makecell{\sref{sec:no_randomly_aqim}\\ \eref{eq:expect_D_of_two_states_reduced_B1}}\\
		\hline     
		\makecell{$\caH_C\sim \mathrm{Gr}(\caH_{B_1B_2},d_C)$ \\
			$\ket{\psi},\ket{\phi}\sim \mu(\caH_C)$} 
		% &$\underset{\caH_C,\ket{\psi},\ket{\phi}}{\bbE}\left[ D(\psi_{B_2},\phi_{B_2}) \right] \ge 2 - 2\sqrt{ \frac{d_{12}d_C(d_1^2-1)+d_1^2(d_2^2-1)}{d_C(d_{12}^2-1)} }$
		% &\makecell{\aref{app:no_randomly_aqim_part2}:\\ \eref{eq:proof_expectation_psi_phi_Hc_lower_bound}} \\
		% \cline{2-3}
		&$ \underset{\caH_C}{\bbE} \left[\subvarA_{X}(\caH_C)\right]
		= \frac{(2d_C-2)(2d_{12}-1)}{(2d_C-1)(2d_{12}-2)}\,\subvarA_{X}(\caH_{B_1B_2}), \quad X\in\{B_1,B_2\}$
		&\makecell{\sref{sec:no_random_aqim_bipartite}:\\ \eref{eq:expectation_trace_distance_of_subspace_and_states_in_subspace}}\\
		\hline
	\end{tabular}
\end{center}

\section{Relationships among figures of merit}
\label{appendix:equivalence_aqim}
The following is the proof of \pref{prop:equivalence_aqim_inaccuracy_variation_maximal}.
\begin{proof}
	(1) Suppose $\scrC=V(\scrA)$ with $V$ being the corresponding masker. Then according to the triangle inequality
	\begin{align}
		\max_{\ket{\psi},\ket{\psi^{\prime}}\in \scrC}\|\psi_{S}-\psi_{S}^{\prime}\|_1
		=\max_{\ket{\psi},\ket{\psi^{\prime}}\in \scrC}\|\psi_{S}-\sigma_{S}+\sigma_{S}-\psi_{S}^{\prime}\|_{1} 
		\leq\max_{\ket{\psi}\in \scrC}\|\psi_{S}-\sigma_{S}\|_{1}+\max_{\ket{\psi^{\prime}}\in \scrC}\|\psi_{S}^{\prime}-\sigma_{S}\|_{1}.
	\end{align}
	Here, $\sigma_S$ can be chosen as any state on $S$, including $\widetilde{\md{1}}$ and $\Pi_{\scrC}^{(S)}$. Therefore, the right-hand side of the inequalities can be proved immediately.
	
	As for the left-hand side, we have
	\begin{align}
		\subinaM_S=\max_{\ket{\psi^{\prime}}\in \scrC}\left\|\psi_{S}^{\prime}-\Pi_{\scrC}^{(S)}\right\|_{1}
		&=\max_{\ket{\psi^{\prime}}\in \scrC}\left\|\underset{\ket{\psi} \in \scrC}{\bbE}\left(\psi_{S}^{\prime}-\psi_{S}\right)\right\|_{1} 
		\leq  \max_{\ket{\psi^{\prime}}\in \scrC}\,\underset{\ket{\psi} \in \scrC}{\bbE}\|\psi_{S}^{\prime}-\psi_{S}\|_{1} \notag\\
		&\leq  \,\underset{\ket{\psi} \in \scrC}{\bbE}\,\max_{\ket{\psi^{\prime}}\in \scrC}\|\psi_{S}^{\prime}-\psi_{S}\|_{1} 
		\leq \max_{\ket{\psi},\ket{\psi^{\prime}}\in \scrC}\|\psi_{S}^{\prime}-\psi_{S}\|_{1}
		=\subvarM_S,
	\end{align}
	and therefore
	\begin{align}
		\subinaM=\max_S \subinaM_S\leq \max_S \subvarM_S =\subvarM.
	\end{align}
	% Note that when $\scrA$ is a Hilbert space, then $\scrC$ is also a Hilbert space and $\Pi_{\scrC}$ becomes the normalized projector on $\scrC$.
	(2) Similar to the previous case, the right-hand side of the inequalities follows directly from the triangle inequality. It is also obviously that
	\begin{align}
		\subinaA_S=\underset{\ket{\psi^{\prime}}\in \scrC}{\bbE}\left\|\psi_{S}^{\prime}-\Pi_{\scrC}^{(S)}\right\|_{1}
		&=\underset{\ket{\psi^{\prime}}\in \scrC}{\bbE}\left\|\underset{\ket{\psi} \in \scrC}{\bbE}\left(\psi_{S}^{\prime}-\psi_{S}\right)\right\|_{1} 
		\leq  \underset{\ket{\psi^{\prime}}\in \scrC}{\bbE}\,\underset{\ket{\psi} \in \scrC}{\bbE}\|\psi_{S}^{\prime}-\psi_{S}\|_{1} 
		=\subvarA_S,
	\end{align}
	and
	\begin{align}
		\subinaA=\max_S \subinaA_S\leq \max_S \subvarA_S =\subvarA.
	\end{align}
	(3) This part can be simply proved by the definition
	\begin{align}
		\subinaA_S
		= \underset{\ket{\psi}\in \scrC}{\bbE}\left\|\psi_{S}^{\prime}-\sigma_S\right\|_{1}\le \max_{\ket{\psi}\in\scrC}\left\|\psi_{S}-\sigma_S\right\|_{1}
		=\subinaM_S,
	\end{align}
	and the proof of the other cases follows from the same reasoning.
\end{proof}

\section{No-random-AQIM theorem in bipartite systems}
\label{appendix:no_random_aqim_bipartite}

In this section, we prove \thref{theo:lower_bound_expectation_sum_subspace_bipartite} and the no-random-AQIM theorem. As explained in the main body, the main idea is to obtain a lower bound of 
\begin{align}
	\label{eq:expectation_sum_subspace_two_reduced_pure}
	\underset{\caH_{C}}{\bbE}\, [\subvarA_{B_1}(\caH_C)+\subvarA_{B_2}(\caH_C)]=\underset{\caH_{C}}{\bbE}\,\underset{\ket{\psi},\ket{\phi}\sim\mu(\caH_C)}{\bbE}[D(\psi_{B_1},\phi_{B_1})+D(\psi_{B_2},\phi_{B_2})],
\end{align} 
and further derive a concentration inequality for the expectation value.

\subsection{Lower bound for expectation value of the trace distance}

\subsubsection{Proof of \pref{prop:expectation_trace_distance_of_subspace_and_states_in_subspace}}
\label{app:average_trace_distance_2states_B1_random subspace}

First, we introduce the following lemma, which allows us to deal with the expectation value taken over the random subspace.
The key idea in proving this proposition is to decompose the expectation—taken over pairs of distinct Haar-random pure states in random subspaces—into two parts: (1) the average trace distance between pairs of states with a fixed fidelity, and (2) the average over the distribution of fidelities. To this end, we begin by stating the following lemma.

\begin{lemma}
	\label{lem:explicit_form_for_expectation_of_subspace_and_states_in_subspace}
	Suppose $\caH_C$ is a random subspace of $\caH_B$, uniformly sampled from $\mathrm{Gr}(\caH_B,d_C)$, then the expectation value of an continuous and bounded function function $P(\psi, \phi)$ of two distinct Haar-random pure states $\ket{\psi}, \ket{\phi}\in \caH_C$ has the following decomposition
	\begin{align}
		\underset{\caH_C}{\bbE}\,\underset{\ket{\psi},\ket{\phi}\sim \mu(\caH_C)}{\bbE}P(\psi,\phi)
		=\int_0^1 E_P(a)d\nu(a)
	\end{align}
	where
	\begin{align}
		E_P(a)&=\underset{U\sim \mu(d_B)}{\bbE}P\left(U\ket{v_a}\bra{v_a}U^\dag,U\ket{0} \bra{0}U^\dag\right), \\
		d\nu(a)&=\frac{2}{\pi^{1/2}} \frac{2d_C-2}{2d_C-1} \frac{ \Gamma(d_C+1/2)}{ \Gamma(d_C)} (1-a^2)^{d_C-3/2} da,
	\end{align}
	and $\ket{v_a}=a\ket{0}+\sqrt{1-a^2}\,\ket{1}$ with $\ket{0},\ket{1}$ being two fixed orthogonal states.
\end{lemma}

\begin{proof}
	First, we assume there is a fixed subspace $\caH_{C_0}$ of $\caH_B$ with $d_{C_0}=d_C$ and a fixed state $\ket{0} \in \caH_{C_0}$, then a random subspace satisfies $\caH_C=U\caH_{C_0}$ for a random unitary $U\sim \mu(d_B)$ and random pure states $\ket{\psi}, \ket{\phi}\in \caH_C$ satisfies $\ket{\psi}=UW_1\ket{0}, \ket{\phi}=UW_2\ket{0}$ with $W_1, W_2 \sim \mu(d_C)$. 
	Then repeatedly applying the left and right invariance of the Haar measure, we can obtain
	\begin{align}
		\underset{\caH_C}{\bbE}\,\underset{\ket{\psi},\ket{\phi}\sim \mu(\caH_C)}{\bbE}P(\psi,\phi)
		&=\underset{U\sim \mu(d_B)}{\bbE}\underset{\substack{W_1, W_2\\ \sim \mu(d_C)} }{\bbE}P\left(UW_1\ket{0}\bra{0}W_1^\dag U^\dag,UW_2\ket{0}\bra{0}W_2^\dag U^\dag\right)\notag \\
		&=\underset{U\sim \mu(d_B)}{\bbE}\underset{\substack{W_1, W_2\\ \sim \mu(d_C)} }{\bbE}P\left(UW_2^\dag W_1\ket{0}\bra{0}W_1^\dag W_2U^\dag,U\ket{0}\bra{0}U^\dag\right)\notag \\
		&=\underset{U\sim \mu(d_B)}{\bbE}\,\underset{W \sim \mu(d_C)}{\bbE}P\left(U W\ket{0}\bra{0}W^\dag U^\dag,U\ket{0}\bra{0}U^\dag\right),
	\end{align}
	Since the function $P(\psi,\phi)$ is continuous and bounded, it is integrable. Then according to the Fubini's theorem, we can exchange the ordering of the integrations of $U$ and $W$ as
	\begin{align}
		\underset{U\sim \mu(d_B)}{\bbE}\,\underset{W \sim \mu(d_C)}{\bbE}P\left(U W\ket{0}\bra{0}W^\dag U^\dag,U\ket{0}\bra{0}U^\dag\right)
		&=\underset{W \sim \mu(d_C)}{\bbE}\, \underset{U\sim \mu(d_B)}{\bbE}P\left(U W\ket{0}\bra{0}W^\dag U^\dag,U\ket{0}\bra{0}U^\dag\right) 
		\notag \\
		&=\underset{\ket{\psi_0} \sim \mu(\caH_{C_0})}{\bbE}\, \underset{U\sim \mu(d_B)}{\bbE}P\left(U \ket{\psi_0}\bra{\psi_0} U^\dag,U\ket{0}\bra{0}U^\dag\right),
	\end{align}
	where in the last line, we denote $\ket{\psi_0}=W\ket{0}$.
	The state $\ket{\psi_0}$ can be parameterized as $\ket{\psi_0}= a \ket{0} +e^{i \alpha} \sqrt{1 - a^2}\,\ket{0^\perp}$
	where $\ket{0^\perp}$ is orthogonal to $\ket{0}$, $|a|=|\<0|\psi_0\>|$ is the fidelity between $\ket{0}$ and $\ket{\psi_0}$ and $0\leq\alpha<2\pi$. 
	Since $\ket{0^\perp}$ is orthogonal to $\ket{0}$, there exists a unitary $U'\sim \mu(d_B)$ satisfying $U'\ket{0}=\ket{0}$ and $e^{i\alpha}\ket{0^\perp}=U'\ket{1}$, then we get
	\begin{align}
		\underset{U\sim \mu(d_B)}{\bbE}P\left(U \ket{\psi_0}\bra{\psi_0} U^\dag,U\ket{0}\bra{0}U^\dag\right)
		&=\underset{U\sim \mu(d_B)}{\bbE}P\left(UU'\ket{v_a}\bra{v_a}U'^\dag U^\dag,U\ket{0} \bra{0}U^\dag\right) \notag \\
		&=\underset{U\sim \mu(d_B)}{\bbE}P\left(U\ket{v_a}\bra{v_a}U^\dag,U\ket{0} \bra{0}U^\dag\right)
	\end{align}
	where $\ket{v_a}=a\ket{0}+\sqrt{1-a^2}\,\ket{1}$. The above equation implies that the above expectation value only depends on the fidelity $a$, therefore
	\begin{align}
		\underset{\caH_C}{\bbE}\,\underset{\ket{\psi},\ket{\phi}\sim \mu(\caH_C)}{\bbE}P(\psi,\phi)
		=\underset{a \sim \nu(a)}{\bbE}\,\underset{U\sim \mu(d_B)}{\bbE}P\left(U\ket{v_a}\bra{v_a}U^\dag,U\ket{0} \bra{0}U^\dag\right)=2\int_{0}^1 E_P(a) \,d\nu(a),
	\end{align}
	where $\nu(a)$ is the integration measure of the fidelity $a$ and the constant $2$ appears because we transform the integration region $[-1,1]$ into $[0,1]$.
	
	Next, we determine the explicit form of the measure $\nu(a)$. Since a pure state is described by a ray in Hilbert space, then a pure state in $d_C$-dimensional Hilbert space corresponds to a point on a $(2d_C-2)$-dimensional unit sphere.
	According to the parameterization of $\ket{v_a}$, the set of states satisfying $\braket{0}{v_a}=a$ constitute a $(2d_C-3)$-dimensional sphere with radius $\sqrt{1-a^2}$. Hence, the differential volume of $d\nu(a)$ is
	\begin{align}
		d\nu(a) = \frac{ S_{2d_C-3}(\sqrt{1-a^2}\,) da }{\int_{-1}^1 S_{2d_C-3}(\sqrt{1-a^2}\,) da }
		= \frac{ S_{2d_C-3}(\sqrt{1-a^2}\,) da }{S_{2d_C-2}(1)}
		=\frac{1}{\pi^{1/2}} \frac{2d_C-2}{2d_C-1} \frac{ \Gamma(d_C+1/2)}{ \Gamma(d_C)} (1-a^2)^{d_C-3/2} da
	\end{align}
	where $S_{n-1}(x)$ denotes the area of a $(n-1)$-dimensional sphere with radius $x$,
	\begin{align}
		S_{n-1}(x) = \frac{n\pi^{n/2}x^{n-1}}{\Gamma(n/2+1)}.
	\end{align}
\end{proof}

With the help of the above proposition, we can now present the proof of \pref{prop:expectation_trace_distance_of_subspace_and_states_in_subspace}.

\begin{proof}
	The proofs for the cases of $X=B_1$ or $B_2$ are identical, then without loss of generality, we consider $X=B_1$.  According to \lref{lem:explicit_form_for_expectation_of_subspace_and_states_in_subspace}, the desired equation can be expressed as
	\begin{align}
		\underset{\caH_C}{\bbE}[\subvarA_{B_1}(\caH_C)]=
		\underset{\caH_C}{\bbE}\,\underset{\ket{\psi},\ket{\phi}\sim \mu(\caH_C)}{\bbE}[D(\psi_{B_1},\phi_{B_1})]=\int_0^1 E_D(a)d\nu(a),
	\end{align}
	where $E_D(a)$ has the form
	\begin{align}
		E_D(a)=\underset{U\sim \mu(d_{12})}{\bbE} \left\| \mathrm{Tr}_{B_2} \left[ U (\ket{0} \bra{0}-\ket{v_a} \bra{v_a}) U^\dagger\right] \right\|_1.
		\label{eq:expectation_value_for_2random_states_with_fixed_inner_product}
	\end{align}
	where $\ket{v_a}=a\ket{0}+\sqrt{1-a^2}\ket{1}$ with $\ket{0},\ket{1}$ being two fixed orthogonal states. To simplify the above equation, we note that the operator $\ket{0} \bra{0} - \ket{v_0} \bra{v_0}$ can be diagonalized by a unitary $W$ as
	\begin{align}
		\ket{0} \bra{0} - \ket{v_a} \bra{v_a}=W\left(\sqrt{1-a^2}\ket{0} \bra{0}-\sqrt{1-a^2}\ket{1} \bra{1}\right)W^\dag,
	\end{align}
	then \eref{eq:expectation_value_for_2random_states_with_fixed_inner_product} becomes
	\begin{align}
		E_D(a)&=\underset{U\sim \mu(d_{12})}{\bbE} \left\| \mathrm{Tr}_{B_2} \left[ U W\left(\sqrt{1-a^2}\ket{0} \bra{0}-\sqrt{1-a^2}\ket{1} \bra{1}\right)W^\dag U^\dagger\right] \right\|_1\notag\\
		&=\sqrt{1-a^2}\underset{U\sim \mu(d_{12})}{\bbE} \left\| \mathrm{Tr}_{B_2} \left[ U (\ket{0} \bra{0} - \ket{1} \bra{1}) U^\dagger\right] \right\|_1 \notag\\
		& =\sqrt{1-a^2}E_D(0).
	\end{align}
	where the second equality above follows the left and right invariance of the Haar measure. Thus, let $a=\cos\theta$ with $\theta\in [0,\pi/2)$, then
	\begin{align}
		\int_0^\frac{\pi}{2} E_D(\cos\theta)\,d\nu
		=\frac{2}{\sqrt{\pi}}\,\frac{2d_C-2}{2d_C-1}\,\frac{\Gamma(d_C+\frac{1}{2})}{\Gamma(d_C)}E_D(0)\int_0^\frac{\pi}{2}(\sin\theta)^{2d_C-2}\,d\theta
		=\frac{2d_C-2}{2d_C-1}E_D(0).
		\label{eq:integral_for_expectation_value}
	\end{align}
	When we set $d_C=d_{12}$, the expectation value over $\caH_{B_1B_2}$ can be expressed as
	\begin{align}
		\underset{\ket{\psi},\ket{\phi}\sim \mu(\caH_{B_1B_2})}{\bbE}\|\psi_{B_1}-\phi_{B_1}\|_1=\frac{2d_{12}-2}{2d_{12}-1}E_D(0),
		\label{eq:expectation_value_for_2random_states_B1B2_orthogonal_repersentation}
	\end{align}    
	which leads to the following expression
	\begin{align}
		\label{eq:expectation_trace_distance_of_subspace_and_states_in_subspace_appendix}
		\underset{\caH_C}{\bbE}[\subvarA_{B_1}(\caH_C)]=
		\underset{\caH_C}{\bbE}\,\underset{\ket{\psi},\ket{\phi}\sim \mu(\caH_C)}{\bbE}\|\psi_{B_1}-\phi_{B_1}\|_1=\frac{(2d_C-2)(2d_{12}-1)}{(2d_C-1)(2d_{12}-2)}\underset{\ket{\psi},\ket{\phi} \sim \mu(\caH_{B_1B_2})}{\bbE} \|\psi_{B_1}-\phi_{B_1}\|_1.
	\end{align}
\end{proof}

\subsubsection{Proof of \thref{theo:expectation_subspace_sum_trace_distance_lower_bound}}
\label{appendix:no_masking_bipartite}

Here we first provide analytical lower bounds for $\subinaAID_{B_1}(\caH_{B_1B_2})$ and $\subinaAID_{B_2}(\caH_{B_1B_2})$.
Next, we prove the summation of them has a finite lower bound $1/3$. 
Hence, we can conclude that the expectation value in \eref{eq:expectation_sum_subspace_two_reduced_pure} also has a lower bound as stated in \thref{theo:expectation_subspace_sum_trace_distance_lower_bound}.

\begin{proposition}
	\label{prop:lowerbound_B1B2}
	Suppose $\ket{\psi}$ is a Haar random pure state in the bipartite Hilbert space $\caH_{B_1B_2}$ with $d_1\le d_2$. Then
	\begin{gather}
		\subinaAID_{B_1}(\caH_{B_1B_2})\geq 2-\frac{4}{\sqrt{d_1}\left(d_{12}\right)_{1/2}}\sum_{i=1}^{d_1}\binom{1/2}{i}\binom{1/2}{i-1}\frac{(d_2)_{3/2-i}}{(d_1+1)_{-i}},\notag \\
		\subinaAID_{B_2}(\caH_{B_1B_2})\geq 2-\frac{4}{\sqrt{d_2}\left(d_{12}\right)_{1/2}}\sum_{i=1}^{d_1}\binom{1/2}{i}\binom{1/2}{i-1}\frac{(d_2)_{3/2-i}}{(d_1+1)_{-i}}.
		\label{eq:proof_expectation_psi_phi_Hb1b2_lower_bound}
	\end{gather}
	where $(\alpha)_\beta:=\Gamma(\alpha+\beta)/\Gamma(\alpha)$ is the Pochhammer symbol and $\binom{\alpha}{\beta}:=\Gamma(\alpha+1)/(\Gamma(\alpha-\beta+1)\Gamma(\beta+1))$ is the binomial coefficient.
\end{proposition}

\begin{proof}
	Combining the Fuchs-van de Graaf inequality \cite{Nielsen_Chuang_2010} and the result of the expected fidelity of $\psi_{B_1}$ and the maximally mixed state $\widetilde{\mathds{1}}_{B_1}$, which has been analytically computed in \cite{PhysRevA.104.022438} by the random matrix theory, we can easily prove the above conclusion as 
	\begin{align}
		\subinaAID_{B_1}(\caH_{B_1B_2})
		&=\underset{\ket{\psi}\sim\mu(\caH_{B_1B_2})}{\bbE}\left\|\psi_{B_1}-\widetilde{\mathds{1}}_{B_1} \right\|_1 
		\geq 2-2\underset{\ket{\psi}\sim\mu(\caH_{B_1B_2})}{\bbE} F\left(\psi_{B_1},\widetilde{\mathds{1}}_{B_1} \right)\notag \\
		&= 2-\frac{4}{\sqrt{d_1}\left(d_{12}\right)_{1/2}}\sum_{i=1}^{d_1}\binom{1/2}{i}\binom{1/2}{i-1}\frac{(d_2)_{3/2-i}}{(d_1+1)_{-i}}.
	\end{align}
	
	% Similarly,  we can transform the expectation value into the following form, 
	%     \begin{align}
		%         \underset{\ket{\psi}\sim \mu(\caH_{B_1B_2})}{\bbE}
		%         [D(\psi_{B_2},\widetilde{\mathds{1}}_{B_2})]
		%         &\ge 2- 2\underset{\ket{\psi}\sim \mu(\caH_{B_1B_2})}{\bbE} F(\psi_{B_2}, \widetilde{\mathds{1}}_{B_2}). 
		%     \end{align}
	For another inequality, according to the Schmidt decomposition, the non-zero eigenvalues of $\psi_{B_1}$ and $\psi_{B_2}$ are identical. Let the eigenvalues of $\psi_{B_1}$ be $\{\lambda_i\}$, then
	\begin{align}
		F\left(\psi_{B_2},\widetilde{\mathds{1}}_{B_2}\right)
		=\sum_{i=1}^{d_1}\sqrt{\frac{\lambda_i}{d_2}}
		=\sqrt{\frac{d_1}{d_2}}F\left(\psi_{B_1},\widetilde{\mathds{1}}_{B_1}\right).
	\end{align}
	So similarly, the expectation vlaue for trace distances can be lower-bounded by the average fidelity
	\begin{align}
		\subinaAID_{B_1}(\caH_{B_1B_2})
		&=\underset{\ket{\psi}\sim\mu(\caH_{B_1B_2})}{\bbE}  \left\|\psi_{B_2}-\widetilde{\mathds{1}}_{B_2}\right\|_1
		\geq 2-2\left[\underset{\ket{\psi}\sim\mu(\caH_{B_1B_2})}{\bbE} F\left(\psi_{B_1},\widetilde{\mathds{1}}_{B_2}\right)\right]\notag \\
		&\geq 2-\frac{4}{\sqrt{d_2}\left(d_{12}\right)_{1/2}}\sum_{i=1}^{d_1}\binom{1/2}{i}\binom{1/2}{i-1}\frac{(d_2)_{3/2-i}}{(d_1+1)_{-i}}.
	\end{align}
\end{proof}

Combining two previous propositions, the summation of these two expectation values has a lower bound as summarized in the following lemma.
\begin{lemma}
	\label{lem:average_identity_whole_space}
	Suppose $\ket{\psi}$ is a Haar random pure state in the bipartite Hilbert space $\caH_{B_1B_2}$ with $d_1\le d_2$. Then
	\begin{align}
		\subinaAID_{B_1}(\caH_{B_1B_2})+\subinaAID_{B_2}(\caH_{B_1B_2})\geq \frac{1}{3}.
	\end{align}
\end{lemma}

\begin{proof}
	Using the results of \pref{prop:lowerbound_B1B2}, we can obtain
	\begin{align}
		\subinaAID_{B_1}(\caH_{B_1B_2})+\subinaAID_{B_2}(\caH_{B_1B_2})
		% &=\underset{\ket{\psi}\sim\mu(\caH_{B_1B_2})}{\bbE} 
		% \left[\left\|\psi_{B_1}-\widetilde{\mathds{1}}_{B_1}\right\|_1+ \left\|\psi_{B_2}-\widetilde{\mathds{1}}_{B_2}\right\|_1\right] \notag \\
		\geq 4- \left(\frac{1}{\sqrt{d_1}} +\frac{1}{\sqrt{d_2}} \right)  \frac{4}{\left(d_{12}\right)_{1/2}}\sum_{i=1}^{d_1}\binom{1/2}{i}\binom{1/2}{i-1}\frac{(d_2)_{3/2-i}}{(d_1+1)_{-i}}.
	\end{align}
	To simplify the above inequality, we observe that the terms with $i\ge 2$ in the above summation satisfy
	\begin{align}
		\binom{1/2}{i}\binom{1/2}{i-1}\frac{(d_2)_{3/2-i}}{(d_1+1)_{-i}}
		&=\frac{\sqrt{\pi}}{2\Gamma(i+1) \Gamma(3/2-i)}\frac{\sqrt{\pi}}{2 \Gamma(i) \Gamma(5/2-i)}\frac{(d_2)_{3/2-i}}{(d_1+1)_{-i}} \notag \\
		&=\frac{(-1)^{2i-3}}{4^{2i-2}}\frac{(2i-2)!}{i!(i-1)!}\frac{(2i-4)!}{(i-1)!(i-2)!}\frac{(d_2)_{3/2-i}}{(d_1+1)_{-i}}\notag \\
		&\leq 0,
	\end{align}
	where the second equation follows from the fact that $\Gamma(1/2-n)=(-4)^n {n!}/{(2n)!}$ for non-negative integer $n$. Therefore, 
	\begin{align}
		\subinaAID_{B_1}(\caH_{B_1B_2})+\subinaAID_{B_2}(\caH_{B_1B_2})
		&\geq 4-\left(\frac{1}{\sqrt{d_1}} +\frac{1}{\sqrt{d_2}} \right)  \frac{4}{\left(d_{12}\right)_{1/2}} \sum_{i=1}^{2}\binom{1/2}{i}\binom{1/2}{i-1}\frac{(d_2)_{3/2-i}}{(d_1+1)_{-i}}\notag \\
		&=4-2\left(1 +\sqrt{\frac{d_1}{d_2}} \right)\frac{\sqrt{d_1}(d_2)_{1/2}}{(d_{12})_{1/2}}\left(1-\frac{1}{8}\frac{d_1-1}{d_2-1/2}\right) \notag \\
		&\ge 4-2\left(1 +\sqrt{\frac{d_1}{d_2}} \right)\left(1-\frac{1}{8}\frac{d_1-1}{d_2-1/2}\right) \notag\\
		&\ge \frac{1}{3},
	\end{align}
	where in the second and third line, we have used the results of \pref{prop:upperbound_for_parameter_in_average_fidelity} shown below.
	
\end{proof}

\begin{proposition}
	\label{prop:upperbound_for_parameter_in_average_fidelity}
	(1) Suppose $x,y\geq 1$. Then
	\begin{align}
		\sqrt{x}\frac{(y)_{1/2}}{(xy)_{1/2}}\leq 1.
	\end{align}  
	(2) Suppose $m\ge n\ge 2$. Then
	\begin{align}
		\left(1 +\sqrt{\frac{n}{m}} \right)\left(1-\frac{1}{8}\frac{n-1}{m-1/2}\right) \le \frac{11}{6}.
	\end{align}
\end{proposition}

\begin{proof}
	(1) To prove the above proposition, we need to use two properties of the digamma function $\Psi(z):=\frac{d}{dz}\Gamma(z)$, shown as
	\begin{gather}
		\Psi(z+1)=\Psi(z)+\frac{1}{z}, \quad
		\frac{d^m}{dz^m}\Psi(z)=(-1)^{m+1}m!\sum_{k=0}^{\infty}\frac{1}{(z+k)^{m+1}}.
	\end{gather}
	When $z$ is real and positive, the second derivative of $\Psi(z)$ is negative, then $\Psi(z)$ is concave, so
	\begin{align}
		\Psi\left(z+1/2\right)\ge \frac{1}{2}\left[\Psi\left(z+1\right) + \Psi\left(z\right) \right]= \Psi\left(z\right)+ \frac{1}{2z}.
	\end{align}
	Then we can find that 
	\begin{align}
		\frac{d}{d x}\left[\sqrt{x}\,\frac{(y)_{1/2}}{(xy)_{1/2}}\right]
		=\sqrt{x}\,y\frac{(y)_{1/2}}{(xy)_{1/2}}\left[\frac{1}{2xy}+\Psi(xy)-\Psi\left(xy+1/2\right)\right]\leq 0,
	\end{align}
	which means the function is monotonically decreasing with respect to $x$, and thus
	\begin{align}
		\sqrt{x}\,\frac{(y)_{1/2}}{(xy)_{1/2}}\leq \sqrt{x}\,\frac{(y)_{1/2}}{(xy)_{1/2}}\Bigg|_{x=1}=1.
	\end{align}
	(2) This conclusion can be directly derived by differentiating the function on the left hand side that
	\begin{align}
		\frac{d}{d n}\left(1 +\sqrt{\frac{n}{m}} \right)\left(1-\frac{1}{8}\frac{n-1}{m-1/2}\right)
		= \frac{8m-3n-2\sqrt{nm}-3}{8(2m-1)\sqrt{nm}} \ge \frac{3m-3}{8(2m-1)\sqrt{nm}} \ge 0,
	\end{align} 
	then the function $\left(1 +\sqrt{\frac{n}{m}} \right)\left(1-\frac{1}{8}\frac{n-1}{m-1/2}\right)$ is increasing for variable $n$. Since $m\ge n \ge 2$,
	\begin{align}
		\left(1 +\sqrt{\frac{n}{m}} \right)\left(1-\frac{1}{8}\frac{n-1}{m-1/2}\right)
		\le \left(1 +\sqrt{\frac{n}{m}} \right)\left(1-\frac{1}{8}\frac{n-1}{m-1/2}\right)\Bigg|_{n=m}
		= 2\,\frac{7m-3}{8m-4} 
		\le 2\,\frac{7m-3}{8m-4}\Bigg|_{m=2}
		=\frac{11}{6}.
	\end{align}
\end{proof}

Having proved \lref{lem:average_identity_whole_space}, for now, we can easily prove \thref{theo:expectation_subspace_sum_trace_distance_lower_bound}.
\begin{proof}
	According to \pref{prop:equivalence_aqim_inaccuracy_variation_maximal}, when $\scrC=\caH_{B_1B_2}$, $\Pi_{\scrC}^{(X)}=\Pi_{B_1B_2}^{(X)}=\widetilde{\md{1}}_X$ for $X\in\{B_1,B_2\}$ and then
	\begin{align}
		\subinaAID_X(\caH_{B_1B_2})\leq \subvarA_X(\caH_{B_1B_2}) \leq 2\subinaAID_X(\caH_{B_1B_2}).
	\end{align}
	Then combine this result with \pref{prop:expectation_trace_distance_of_subspace_and_states_in_subspace} and \lref{lem:average_identity_whole_space} and we can get
	\begin{align}
		\underset{\caH_C}{\bbE}\,\left[\subvarA_{B_1}(\caH_C)+\subvarA_{B_2}(\caH_C)\right] 
		&= \frac{(2d_C-2)(2d_{12}-1)}{(2d_C-1)(2d_{12}-2)}
		[\subvarA_{B_1}(\caH_{B_1B_2})+\subvarA_{B_2}(\caH_{B_1B_2})] \notag \\
		&\ge \frac{(2d_C-2)(2d_{12}-1)}{(2d_C-1)(2d_{12}-2)} [\subinaAID_{B_1}(\caH_{B_1B_2})+\subinaAID_{B_2}(\caH_{B_1B_2})] \notag \\
		&\geq \frac{1}{3}\frac{(2d_C-2)(2d_{12}-1)}{(2d_C-1)(2d_{12}-2)}.
	\end{align}
	Thus,
	\begin{align}
		\underset{\caH_C}{\bbE}\,[\subvarA(\caH_C)]\geq \frac{1}{2}\,\underset{\caH_C}{\bbE}\left[\subvarA_{B_1}(\caH_C)+\subvarA_{B_2}(\caH_C)\right] \geq \frac{1}{6}\frac{(2d_C-2)(2d_{12}-1)}{(2d_C-1)(2d_{12}-2)}.
	\end{align}
\end{proof}

\subsection{Proof of \thref{theo:lower_bound_expectation_sum_subspace_bipartite}}
\label{appendix:lower_bound_expectation_sum_subspace_bipartite}

To get the concentration result about $\subvarA_X(\caH_C)$ on the Grassmannian $\mathrm{Gr}(\caH_{B_1B_2},d_C)$, we also need to calculate the Lipschitz constant. Here, according to Appendix B.4. of \cite{aubrun2017alice}, the corresponding metric $M_{\mathrm{Gr}}(\caH_{C_1},\caH_{C_1})$ satisfies the following condition.

Given a fixed subspace $\caH_{C_0}$ in the Grassmannian $\mathrm{Gr}(\caH_{B_1B_2},d_C)$, any subspace $\caH_{C}\in\mathrm{Gr}(\caH_{B_1B_2},d_C)$ is equal to the fixed subspace $\caH_{C_0}$ up to a unitary $U$ as $\caH_C=U\caH_{C_0}$. Suppose $\caH_{C_1}=U_1\caH_{C_0}$ and $\caH_{C_2}=U_2\caH_{C_0}$, then the metric on $\mathrm{Gr}(\caH_{B_1B_2},d_C)$ satisfies
\begin{align}
	M_{\mathrm{Gr}}(\caH_{C_1},\caH_{C_2})\geq \sqrt{2}\, \| U_1-U_2\|_{\infty}.
\end{align}

% At the same time, $\ket{\psi_0}, \ket{\phi_0}$ are transformed into $\ket{\psi}=U\ket{\psi_0}$ and $\ket{\phi}=U\ket{\phi_0}$. 

Next, we introduce the following result about $\kappa$-Lipschitz function on a Grassmannian space, which generalizes the result about $1$-Lipschitz function in \cite{aubrun2017alice}. 
\begin{lemma}
	\label{lem:concentration_grass}
	Suppose $f$ is a $\kappa$-Lipschitz function defined on the Grassmannian $\mathrm{Gr}(\caH_{B_1B_2},d_C)$, then for $\alpha\ge 0$,
	\begin{align}
		\Pr\left\{ f(\caH_C) \gtrless \bbE[f] \pm \alpha \right\}
		\le \exp\left(-\frac{d_{12}\alpha^2}{2\kappa^2}\right),
	\end{align}
	where $\bbE[f]$ is the expectation value of $f$ with respect to the uniform distribution on $\mathrm{Gr}(\caH_{B_1B_2},d_C)$. 
\end{lemma}

Now we can introduce the following proposition.

\begin{proposition}
	\label{prop:Lipschitz_constant_for_1norm_2states_subspace}
	Suppose $\caH_C$ is a random subspace in the Grassmannian $\mathrm{Gr}(\caH_{B_1B_2},d_C)$. Then the Lipschitz constant of $\subvarA_X(\caH_C)$ is $4/\sqrt{2}$ for both $X=B_1$ and $B_2$.    
\end{proposition}

\begin{proof}
	Since the proofs for the cases of $X=B_1$ and $B_2$ are identical, thus we focus on the case of $X=B_1$. Suppose $\ket{\psi_1},\ket{\phi_1}\in \caH_{C_1}$ and $\ket{\psi_2},\ket{\phi_2}\in \caH_{C_2}$. Then the Lipschitz constant of $\subvarA_{B_1}$ can be calcuated as
	\begin{align}
		&\left| \subvarA_{B_1}(\caH_{C_1}) -\subvarA_{B_1}(\caH_{C_2}) \right| \notag\\
		&= \left|\underset{\substack{W_1, W_2\\ \sim \mu(d_C)} }{\bbE}\left[ \left\| \Tr_{B_2}( W_1\psi_1 W_1^{\dagger} ) -\Tr_{B_2}( W_2\phi_1 W_2^{\dagger} ) \right\|_1 \right] -
		\underset{\substack{W_1, W_2\\ \sim \mu(d_C)} }{\bbE}\left[ \left\| \Tr_{B_2}( W_1\psi_2 W_1^{\dagger} ) -\Tr_{B_2}( W_2\phi_2 W_2^{\dagger} ) \right\|_1 \right] \right| \notag\\
		&\le \underset{\substack{W_1, W_2\\ \sim \mu(d_C)} }{\bbE} \left| \left\| \Tr_{B_2}( W_1\psi_1 W_1^{\dagger} ) -\Tr_{B_2}( W_2\phi_1 W_2^{\dagger} ) \right\|_1 -
		\left\| \Tr_{B_2}( W_1\psi_2 W_1^{\dagger} ) -\Tr_{B_2}( W_2\phi_2 W_2^{\dagger} ) \right\|_1\right| \notag\\
		&\le \underset{\substack{W_1, W_2\\ \sim \mu(d_C)} }{\bbE}  \left\| \Tr_{B_2}( W_1\psi_1 W_1^{\dagger} ) -\Tr_{B_2}( W_2\phi_1 W_2^{\dagger} ) -
		\Tr_{B_2}( W_1\psi_2 W_1^{\dagger} ) +\Tr_{B_2}( W_2\phi_2 W_2^{\dagger} ) \right\|_1 \notag\\
		&\le  \underset{\substack{W_1, W_2\\ \sim \mu(d_C)} }{\bbE} \left( \left\| \Tr_{B_2}( W_1\psi_1 W_1^{\dagger} )  -
		\Tr_{B_2}( W_1\psi_2 W_1^{\dagger} )\right\|_1
		+   \left\| \Tr_{B_2}( W_2\phi_2 W_2^{\dagger} )  -\Tr_{B_2}( W_2\phi_1 W_2^{\dagger} ) \right\|_1 \right) \notag\\
		&\le  \underset{\substack{W_1, W_2\\ \sim \mu(d_C)} }{\bbE} \left( \left\| \psi_1 - \psi_2 \right\|_1
		+   \left\| \phi_2  -\phi_1 \right\|_1 \right).
	\end{align}
	where we have used the triangle inequality and the contractive property of trace distance $\|\mathrm{Tr}_B(\rho)\|_1\leq\|\rho\|_1$. Next, we can further assume that $\caH_{C_1}=U_1\caH_C$ and $\caH_{C_2}=U_2\caH_C$ with $U_1,U_2\in \mathrm{SU}(d_{12})$, then $\psi_1=U_1 \psi_0 U_1^{\dagger}$ and $\psi_2=U_2 \psi_0 U_2^{\dagger}$, and thus
	\begin{align}
		\left\| \psi_1 - \psi_2 \right\|_1
		&= \left\| U_1 \psi_0 U_1^{\dagger} - U_2 \psi_0 U_2^{\dagger} \right\|_1 
		= \left\| (U_1-U_2) \psi_0 U_1^{\dagger} - U_2 \psi_0 (U_1^{\dagger}-U_2^{\dagger} )\right\|_1 \notag \\
		&\le \| U_1-U_2\|_{\infty} \|\psi_0\|_1 \|U_1\|_{\infty} + \|U_2\|_{\infty} \|\psi_0\|_1 \| U_1-U_2\|_{\infty} 
		\le 2 \| U_1-U_2\|_{\infty}.
		\label{eq:clucuate_lip_constant_function_q_1}
	\end{align}
	Similarly, we can find that $\left\| \phi_1  -\phi_2 \right\|_1\leq 2 \| U_1-U_2\|_{\infty}$. Hence, we get 
	\begin{align}
		\frac{\left| \subvarA_{B_1}(\caH_{C_1}) -\subvarA_{B_1}(\caH_{C_2}) \right|}{M_{\mathrm{Gr}}(\caH_{C_1},\caH_{C_1})}
		\le\frac{\left| \subvarA_{B_1}(\caH_{C_1}) -\subvarA_{B_1}(\caH_{C_2}) \right|}{\sqrt{2}\,\| U_1-U_2\|_{\infty}}
		\le \frac{4}{\sqrt{2}}.
		\label{eq:clucuate_lip_constant_function_q_2}
	\end{align}
	So the Lipschitz constant of $\subvarA_{B_1}(\caH_C)$ is $4/\sqrt{2}$.
\end{proof}

Now, the proof of \thref{theo:lower_bound_expectation_sum_subspace_bipartite} is straightforward as follows.
\begin{proof}
	According to \pref{prop:Lipschitz_constant_for_1norm_2states_subspace}, we can easily find that the Lipschitz constant of function $\subvarA_{B_1}(\caH_{C_1})+\subvarA_{B_2}(\caH_{C_1})$ is $4\sqrt{2}$ as
	\begin{align}
		&\big| [\subvarA_{B_1}(\caH_{C_1})+\subvarA_{B_2}(\caH_{C_1})] -[\subvarA_{B_1}(\caH_{C_2})+\subvarA_{B_2}(\caH_{C_2})] \big| \notag \\
		&\quad\leq \left| \subvarA_{B_1}(\caH_{C_1}) -\subvarA_{B_1}(\caH_{C_2}) \right|+\left| \subvarA_{B_2}(\caH_{C_1}) -\subvarA_{B_2}(\caH_{C_2}) \right| 
		\leq 4\sqrt{2}.
	\end{align}
	Therefore, through combining \lref{lem:concentration_grass} and \thref{theo:expectation_subspace_sum_trace_distance_lower_bound}, we can directly obtain the following concentration result
	\begin{align}
		\Pr\left\{ \subvarA(\caH_C)< \frac{1}{6}\frac{(2d_C-2)(2d_{12}-1)}{(2d_C-1)(2d_{12}-2)}-\alpha\right\}
		&\leq\Pr\left\{ \subvarA_{B_1}(\caH_C)+\subvarA_{B_2}(\caH_C) < \frac{1}{3}\frac{(2d_C-2)(2d_{12}-1)}{(2d_C-1)(2d_{12}-2)}-2\alpha\right\} \notag \\
		&\leq \Pr\left\{ \subvarA_{B_1}(\caH_C)+\subvarA_{B_2}(\caH_C) < \underset{\caH_C}{\bbE}[\subvarA_{B_1}(\caH_C)+\subvarA_{B_2}(\caH_C)]-2\alpha\right\} \notag \\
		&\leq \exp\left( -\frac{d_{12}\alpha^2}{16} \right).
	\end{align}
	
\end{proof}

\subsection{Additional concentration results on bipartite systems}

\subsubsection{Concentration results on $\subvarA_{B_1}(\caH_C)$ and $\subvarA_{B_2}(\caH_C)$ separately }

In the previous subsections, we have obtained universal lower bounds for the expectation values of $\bbE_{\caH_C}[\subvarA_{B_1}(\caH_C)]$ and $\bbE_{\caH_C}[\subvarA_{B_2}(\caH_C)]$. Using a similar method, we can also easily obtain concentration results on $\subvarA_{B_1}(\caH_C)$ and $\subvarA_{B_2}(\caH_C)$ separately, which are similar to \thref{theo:lower_bound_expectation_sum_subspace_bipartite}.

\begin{lemma}
	\label{lem:lower_bound_B1B2_bipartite}
	Suppose $\caH_C$ is a random subspace of $\caH_{B_1B_2}$, uniformly sampled from $\mathrm{Gr}(\caH_{B_1B_2},d_C)$. Then for $\alpha>0$
	\begin{gather}
		\Pr\left\{ \subvarA_{B_1}(\caH_C)< v_1-\alpha\right\}\le \exp\left( -\frac{d_{12}\alpha^2}{16} \right), \\
		\Pr\left\{ \subvarA_{B_2}(\caH_C)< v_2-\alpha\right\} 
		\le \exp\left( -\frac{d_{12}\alpha^2}{16} \right),
	\end{gather}
	where
	\begin{gather}
		v_1=\frac{(2d_C-2)(2d_{12}-1)}{(2d_C-1)(2d_{12}-2)}\left[2-\frac{4}{\sqrt{d_1}\left(d_1d_1\right)_{1/2}}\sum_{i=1}^{d_1}\binom{1/2}{i}\binom{1/2}{i-1}\frac{(d_1)_{3/2-i}}{(d_1+1)_{-i}}\right],\\
		v_2=\frac{(2d_C-2)(2d_{12}-1)}{(2d_C-1)(2d_{12}-2)}\left[ 2-\frac{4}{\sqrt{d_2}\left(d_{12}\right)_{1/2}}\sum_{i=1}^{d_1}\binom{1/2}{i}\binom{1/2}{i-1}\frac{(d_2)_{3/2-i}}{(d_1+1)_{-i}} \right].
	\end{gather}
\end{lemma}

\begin{proof}
	Similar to the previous case, since $\Pi_{B_1B_2}^{(X)}=\widetilde{\md{1}}_{X}$, according to \pref{prop:equivalence_aqim_inaccuracy_variation_maximal}
	\begin{align}
		\subinaAID_X(\caH_{B_1B_2})=\subinaA_X(\caH_{B_1B_2})\leq \subvarA_X(\caH_{B_1B_2}) \leq 2\subinaAID_X(\caH_{B_1B_2}).
	\end{align}
	Then the lemma can be proved directly by applying Propositions \ref{prop:expectation_trace_distance_of_subspace_and_states_in_subspace}, \ref{prop:lowerbound_B1B2} and \lref{lem:concentration_grass}.
\end{proof}

\subsubsection{Concentration results on Haar random states in $\caH_{B_1B_2}$}
\label{app:Concentration_results_in_B1B2}

In this section, we demonstrate that there exists a lower bound for the expectation value of the trace distance of
two distinct Haar random pure states in $\caH_{B_1B_2}$. Here we are interested in the function $E_X:\caH_{B_1B_2}\times \caH_{B_1B_2}\rightarrow \bbR^{+}$ defined as
\begin{align}
	E_X(\ket{\psi}, \ket{\phi})
	:=D(\psi_{X},\phi_{X})
	=\left\| \psi_{X}-\phi_{X} \right\|_1.
\end{align}
for $X=\{B_1,B_2\}$. Note that the lower bound of the expectation value of $E_{X}(\ket{\psi}, \ket{\phi})$ taken over Haar random pure states in $\caH_{B_1B_2}$ can be computed by Propositions \ref{prop:equivalence_aqim_inaccuracy_variation_maximal} and \ref{prop:lowerbound_B1B2}, then we only need to consider computing the Lipschitz constants.  

As function $E_X(\ket{\psi},\ket{\phi})$ are defined on the Cartesian product $\caH_{B_1B_2}\times \caH_{B_1B_2}$, then the Euclidean $2$-norm in the product space is defined as
\begin{align}
	\|(\ket{\psi_1}, \ket{\psi_2})-(\ket{\phi_1}, \ket{\phi_2}) \|_2
	:= \sqrt{ \|\ket{\psi_1}- \ket{\phi_1}\|_2^2 + \|\ket{\psi_2}- \ket{\phi_2}\|_2^2 }.
\end{align}
We also consider another function defined as the expectation value of $E_X(\ket{\psi},\ket{\phi})$ over one variable,
\begin{align}
	P_X(\ket{\psi}):=\underset{\ket{\phi}\sim \mu(\caH_{B_1B_2})}{\bbE}[E_X(\ket{\psi},\ket{\phi})],
\end{align}
then we have the following results about the Lipschitz constants of these two functions.
\begin{proposition}
	The Lipschitz constant of $E_X(\ket{\psi},\ket{\phi})$ is $2\sqrt{2}$ and the Lipschitz constant of $P_X(\ket{\psi})$ is also $2\sqrt{2}$.
\end{proposition}
\begin{proof}
	The proofs for the cases of $X=B_1$ or $B_2$ are identical, then without loss of generality, we consider $X=B_1$. The computation is directly as
	\begin{align}
		|E_{B_1}(\ket{\psi_1},\ket{\psi_2})-E_{B_1}(\ket{\phi_1},\ket{\phi_2})|
		&\le \| \psi_1^{(B_2)} -\psi_2^{(B_2)} - \phi_1^{(B_2)} + \phi_2^{(B_2)}   \|_1 \notag \\
		&\le \| \psi_1^{(B_2)}  - \phi_1^{(B_2)} \|_1 + \|\phi_2^{(B_2)} - \psi_2^{(B_2)} \|_1 \notag\\
		&\le \| \psi_1 - \phi_1 \|_1 + \|\phi_2 - \psi_2 \|_1 \notag\\
		&\le 2 ( \|\ket{\psi_1}-\ket{\phi_1}\|_2 + \|\ket{\psi_2}-\ket{\phi_2}\|_2  ) \notag\\
		&\le 2\sqrt{2} \sqrt{\|(\ket{\psi_1}, \ket{\psi_2})-(\ket{\phi_1}, \ket{\phi_2}) \|_2} 
	\end{align}
	where in the last line we used Cauchy-Schwartz inequality, then we get the Lipschitz constant of $E_{B_1}(\ket{\psi},\ket{\phi})$ is $2\sqrt{2}$. As for $P_{B_1}(\ket{\psi})$, 
	\begin{align}
		|P_{B_1}(\ket{\psi}) -P_{B_1}(\ket{\phi}) |
		&= \left|\underset{\ket{\phi_2}\sim \mu(\caH_{B_1B_2})}{\bbE} [E_{B_1}(\ket{\psi},\ket{\phi_2})] -\underset{\ket{\phi_2}\sim \mu(\caH_{B_1B_2})}{\bbE}[E_{B_1}(\ket{\phi},\ket{\phi_2})] \right| \notag\\
		&\le \underset{\ket{\phi_2}\sim \mu(\caH_{B_1B_2})}{\bbE} |E_{B_1}(\ket{\psi},\ket{\phi_2}) -E_{B_1}(\ket{\phi},\ket{\phi_2})| \notag\\
		&\le 2\sqrt{2} \underset{\ket{\phi_2}\sim \mu(\caH_{B_1B_2})}{\bbE}  \|(\ket{\psi}, \ket{\phi_2})-(\ket{\phi}, \ket{\phi_2}) \|_2 
		\le 2\sqrt{2}\, \|\ket{\psi}-\ket{\phi}\|_2,
	\end{align}
	then we get the Lipschitz constant of $P_{B_1}(\ket{\psi})$ is $2\sqrt{2}$
\end{proof}

Having obtained the expectation value and the Lipschitz constant, we can provide the concentration result below.
\begin{lemma}
	\label{lem:trace_two_pure_B2}
	Suppose $\ket{\psi}, \ket{\phi}\sim \mu(\caH_{B_1B_2})$. Then for $\alpha\ge0$, 
	\begin{gather}
		\Pr\left\{ D(\psi_{B_1},\phi_{B_1})< 2-\frac{4}{\sqrt{d_1}\left(d_1d_1\right)_{1/2}}\sum_{i=1}^{d_1}\binom{1/2}{i}\binom{1/2}{i-1}\frac{(d_1)_{3/2-i}}{(d_1+1)_{-i}}-\alpha \right\} \le 4 \exp\left(-\frac{d_{12}\alpha^2}{144\pi^2\ln2}\right), \\
		\Pr\left\{ D(\psi_{B_2},\phi_{B_2})< 2-\frac{4}{\sqrt{d_2}\left(d_1d_1\right)_{1/2}}\sum_{i=1}^{d_1}\binom{1/2}{i}\binom{1/2}{i-1}\frac{(d_1)_{3/2-i}}{(d_1+1)_{-i}}-\alpha \right\} \le 4 \exp\left(-\frac{d_{12}\alpha^2}{144\pi^2\ln2}\right).
	\end{gather}
\end{lemma}

\begin{proof}
	The lemma is proved by the same method used in the previous sections. In order to use \lref{lem:Levy's lemma}, we need to split the probability about the function $E_X(\ket{\psi},\ket{\phi})$ on the Cartesian product $\caH_{B_1B_2}\times \caH_{B_1B_2}$ into two probabilities, each of which is about a function on a single Hilbert space $\caH_{B_1B_2}$ as
	\begin{align}
		\Pr\left\{ E_X(\ket{\psi},\ket{\phi}) < \bbE[E_X]-\alpha \right\}
		&= \Pr\left\{ E_X(\ket{\psi},\ket{\phi})-P_X(\ket{\psi}) + P_X(\ket{\psi})-\bbE[E_X] < -\alpha \right\} \notag\\
		&\le  \Pr\left( \left\{ E_X(\ket{\psi},\ket{\phi})-P_X(\ket{\psi})< -\frac{\alpha}{2} \right\} \bigvee \left\{P_X(\ket{\psi})-\bbE[E_X] <-\frac{\alpha}{2} \right\}\right) \notag\\
		&\le  \Pr\left\{ E_X(\ket{\psi},\ket{\phi})-P_X(\ket{\psi}) <-\frac{\alpha}{2} \right\} +
		\Pr\left\{ P_X(\ket{\psi})-\bbE[E_X] <-\frac{\alpha}{2} \right\} \notag\\
		&= 4\exp\left(-\frac{d_{12}\alpha^2}{144\pi^2\ln2 }\right)
	\end{align}
	where $\bbE[E_X]$ denotes ${\bbE}_{\ket{\psi},\ket{\phi}\sim \mu(\caH_{B_1B_2})}
	\left[ E_X(\ket{\psi}, \ket{\phi}) \right]$.

	In the last line, we have used two concentration results, which are directly obtained from the Levy's lemma. The first is about the function $P_X(\ket{\psi})$ of the random pure state $\ket{\psi}$, 
	\begin{align}
		\Pr\left\{P_X(\ket{\psi}) < \bbE[E_X]-\alpha \right\}
		\le 2 \exp\left(-\frac{d_{12}\alpha^2}{36\pi^2\ln2} \right),
	\end{align}
	and the second is about $E_X(\ket{\psi},\ket{\phi})$ as a function of Haar random pure state $\ket{\phi}$,
	\begin{align}
		\Pr\left\{ E_X(\ket{\psi},\ket{\phi}) < P_X(\ket{\psi})-\alpha \right\}
		\le 2\exp\left(-\frac{d_{12}\alpha^2}{36\pi^2\ln2}\right),    
	\end{align}
	where we have assumed the $\ket{\psi}$ is fixed. Finally, it is not difficult to find that $\bbE[E_X]$ is simply equal to $\subvarA_X(\caH_{B_1B_2})$, so according to \pref{prop:equivalence_aqim_inaccuracy_variation_maximal}, $\bbE[E_X]=\subvarA_X(\caH_{B_1B_2}) \geq \Lambda_X(\caH_{B_1B_2}) = \subinaAID_X(\caH_{B_1B_2})$. Thus, we can complete the proof using \pref{prop:lowerbound_B1B2}.
	
\end{proof}

Note that we can obtain another concentration result for $D(\psi_{B_2},\phi_{B_2})$ as
\begin{align}
	\Pr\left\{ D(\psi_{B_2},\phi_{B_2})< 2-2\sqrt{\frac{d_1}{d_2}}-\alpha \right\} \le 4 \exp\left(-\frac{d_{12}\alpha^2}{144\pi^2\ln2}\right),
\end{align}
since we can compute another upper bound for the expectation value of $D(\psi_{B_2},\phi_{B_2})$, as stated in the following.
\begin{proposition}
	\label{prop:expec_trace_distance_two_pure}
	Suppose $\ket{\psi}, \ket{\phi}$ are two distinct random pure states following $\mu(\caH_{B_1B_2})$, then 
	\begin{align}
		\subvarA_{B_2}(\caH_{B_1B_2})=\underset{\ket{\psi},\ket{\phi}\sim \mu(\caH_{B_1B_2})}{\bbE}\left[ E_{B_2}(\ket{\psi}, \ket{\phi}) \right]
		\ge 2-2\sqrt{\frac{d_1}{d_2}}.
		\label{eq:proof_expectation_psi_phi_Hb1b2_lower_bound2}
	\end{align}
\end{proposition}
\begin{proof}
	The computation is directly as
	\begin{align}
		\underset{\ket{\psi},\ket{\phi}\sim \mu(\caH_{B_1B_2})}{\bbE}
		\left[ E_{B_2}(\ket{\psi}, \ket{\phi}) \right]
		&\ge 2- 2\underset{\ket{\psi},\ket{\phi}\sim \mu(\caH_{B_1B_2})}{\bbE} F(\psi_{B_2}, \phi_{B_2}) \notag \\
		&= 2- 2\underset{\ket{\psi},\ket{\phi}\sim \mu(\caH_{B_1B_2})}{\bbE} \left\| \sqrt{\psi_{B_2}} \sqrt{\phi_{B_2}} \right\|_1 \notag\\
		&\ge 2- 2\sqrt{d_1} \sqrt{\underset{\ket{\psi},\ket{\phi}\sim \mu(\caH_{B_1B_2})}{\bbE} \Tr(\psi_{B_2} \phi_{B_2} )} \notag\\
		&= 2-2\sqrt{\frac{d_1}{d_2}},
	\end{align}
	where in the second line we used the fact the ranks  of $\psi_{B_2},\phi_{B_2}$ are smaller than $d_1$ and 
	the following result
	\begin{align}
		\underset{\ket{\psi},\ket{\phi}\sim \mu(\caH_{B_1B_2})}{\bbE} \Tr(\psi_{B_2} \phi_{B_2})
		= \Tr\left( \bbF_{B_2}\left( \underset{\ket{\psi},\ket{\phi}\sim \mu(\caH_{B_1B_2})}{\bbE} \psi_{B_2} \otimes\phi_{B_2}  \right)\right)
		= \frac{1}{d_2^2}\Tr\left( \bbF_{B_2} \mathbb{I}_{B_2}\right)
		= \frac{1}{d_2},
	\end{align}
	where in the last equation, $\mathbb{F}_{B_2}$ denotes the flip operator acting on the tensor product of two copies of subsystem $B_2$ and $\mathbb{I}_{B_2}$ denotes the identity operator acting on the tensor product of two copies of $B_2$.
\end{proof}

\section{Concentration results on bipartite systems}
\label{appendix:concentration_bipartite_identity_projector}

\subsection{Concentration result on typical subspaces in the case of a subsystem identity operator}
\label{app:identity_operator}

\subsubsection{Preparatory work: $\epsilon$-net and Levy's lemma}
To prove the results of this paper, we need to introduce two basic tools. The first tool is the existence of “small” fine nets, which are used to discretize continuous spaces.

\begin{lemma}[$\epsilon$-net]
	\label{lem:epsilon_net}
	For $0\le \epsilon \le 1$ and a $d$-dimensional Hilbert space $\caH$, there exists a subset $\caL$ of $\caH$ with $|\caL|\leq (5/\epsilon)^{2d}$, such that for every state $\ket{\psi}\in \caH$, there exists a state $\ket{\widetilde{\psi}}\in \caL$ satisfying $||\ket{\psi}-\ket{\widetilde{\psi}}||_2\leq \epsilon/2$.  
	Then, for a $\kappa$-Lipschitz function $f$ on $\caH$, there exists $|f(\ket{\psi})-f(\ket{\widetilde{\psi}})|\le \kappa\epsilon/2$. 
\end{lemma}
The first part of the above lemma has been proved in \cite{Hayden_2006, Hayden2004_randomizing}, which means there exists a discrete $\epsilon$-net for a Hilbert space $\caH$, and the second part directly follows from the definitions of $\epsilon$-net $\caL$ and $\kappa$-Lipschitz function.

The next tool is the well-known Levy's lemma, which demonstrates the concentration property of a $\kappa$-Lipschitz function. In the below, the Levy's lemma follows the conventions of \cite{Hayden_2006}.

\begin{lemma}[Levy's lemma]
	\label{lem:Levy's lemma}
	Let $f :\caH \to R$ be a function with Lipschitz constant $\kappa$ (with respect to the Euclidean norm of vectors) and a random state $\ket{\psi}\in \caH$ be sampled according to the Haar measure. Then for $\alpha\ge 0$
	\begin{align}
		\Pr\left\{f(X)-\bbE[f] \gtrless  \pm\alpha \right\}
		\le
		2\exp\left(-\frac{2d\alpha^2}{9\pi^3(\ln{2})\,\kappa^2} \right),
	\end{align}
	where $d$ is the dimension of the Hilbert space $\caH$ and $\bbE[f]$ is the mean value of $f$.
\end{lemma}

\subsubsection{Concentration results for the trace distance with respect to a subsystem identity operator}
\label{app:concentration_psi_identity}

Now we formally define our desired function and present some properties that will be used later. Given a bipartite system $\caH_{B_1B_2}$ and a pure state $\ket{\psi}\in \caH_{B_1B_2}$, let
\begin{align}
	R(\ket{\psi}):=D(\psi_{B_1}, \widetilde{\md{1}}_{B_1})=
	\left\| \psi_{B_1}-\widetilde{\md{1}}_{B_1} \right\|_1,
	\label{eq:def_function_f}
\end{align}
and we have the following lemma.
\begin{lemma}
	\label{lem:concentrate_of_identity_function}
	Suppose $\ket{\psi}\sim \mu(\caH_{B_1B_2})$. Then for $\alpha>0$,
	\begin{align}
		\Pr\left\{ R(\ket{\psi}) \ge \sqrt{\frac{d_1^2-1}{d_{12}+1}} + \alpha \right\}
		\le 2\exp\left( -\frac{d_{12}}{18\pi^3\ln{2}}\alpha^2 \right)
	\end{align}
\end{lemma}
\begin{proof}
	The above result can be derived directly by applying the Levy's lemma to this function after computing its Lipschitz constant and mean value. Firstly,
	we have
	\begin{align}
		\label{Eq:Lipschitz_function_identity}
		|R(\ket{\psi})-R(\ket{\phi})|=\left|\left\|\psi_{B_1}-\widetilde{\md{1}}_{B_1}\right\|_1-\left\|\phi_{B_1}-\widetilde{\md{1}}_{B_1}\right\|_1\right|\leq\|\psi_{B_1}-\phi_{B_1}\|_1\leq\|\psi-\phi\|_1\leq2\||\psi\rangle-|\phi\rangle\|_2,
	\end{align}
	where we have used the triangle inequality, the contractive property of trace distance $\|\mathrm{Tr}_B(\rho)\|_1\leq\|\rho\|_1$ and $\|\psi-\phi\|_1\leq2\||\psi\rangle-|\phi\rangle\|_2$, so the Lipschitz constant of $R$ is $\kappa=2$. 
	Secondly, the expectation value of $R$ is
	\begin{align}
		\subinaAID_{B_1}(\caH_{B_1B_2})
		&=\underset{\ket{\psi}\sim \mu(\caH_{B_1B_2})}{\bbE} [R(\ket{\psi})]
		\leq
		\sqrt{d_1}\:\underset{\ket{\psi}\sim \mu(\caH_{B_1B_2})}{\bbE} \left[\left\|\psi_{B_1}-\widetilde{\md{1}}_{B_1}\right\|_2 \right] \notag\\
		&\leq
		\sqrt{d_1}\:\sqrt{\underset{\ket{\psi}\sim \mu(\caH_{B_1B_2})}{\bbE} \left[\left\|\psi_{B_1}-\widetilde{\md{1}}_{B_1}\right\|_2^2 \right]}
		= \sqrt{\frac{d_1^2-1}{d_{12}+1}},\label{eq:proof_expectation_psi_identity_upper_bound}
	\end{align}
	where in the last equality, we have used \pref{Prop:expectation_2_norm_identity}. 
	Hence, by Levy's lemma \ref{lem:Levy's lemma}, we get the conclusion
	\begin{align}
		\Pr\left\{ R(\ket{\psi}) \ge \sqrt{\frac{d_1^2-1}{d_{12}+1}} + \alpha \right\}
		\leq \Pr\left\{ R(\ket{\psi}) \ge  \underset{\ket{\psi}\sim \mu(\caH_{B_1B_2})}{\bbE}\left[R(\ket{\psi}) \right]+\alpha \right\}
		\leq 2\exp\left( -\frac{d_{12}}{18\pi^3\ln{2}}\alpha^2 \right).
	\end{align}
\end{proof}

\begin{proposition}
	\label{Prop:expectation_2_norm_identity}
	Suppose $\ket{\psi}\sim \mu(\caH_{B_1B_2})$. Then
	\begin{align}
		\underset{\ket{\psi}\sim \mu(\caH_{B_1B_2})}{\bbE}
		\left[\left\|\psi_{B_1}-\widetilde{\md{1}}_{B_1}\right\|_2^2\right]
		=\frac{d_1^2-1}{d_1(d_{12}+1)}.
	\end{align}
\end{proposition}

\begin{proof}
	It follows from the definition of the 2-norm that
	\begin{align}
		\left \| \psi_{B_1} - \widetilde{\md{1}}_{B_1} \right \|_2^2
		&=\mathrm{Tr}\left[\left(\mathrm{Tr_{B_2}}\left(U\psi_0 U^\dag \right)-\widetilde{\md{1}}_{B_1}\right)^2 \right] \notag \\
		&=\mathrm{Tr}\left[\mathrm{Tr_{B_2}}\left(U\psi_0 U^\dag \right)\mathrm{Tr_{B_2}}\left(U\psi_0 U^\dag \right) \right]-\frac{1}        {d_1} \notag\\
		&=\mathrm{Tr}\left[\left(U\psi_0 U^\dag \right)^{\otimes 2} \mathbb{F}_{B_1} \otimes \mathbb{I}_{B_2}\right]-\frac{1}{d_1},
	\end{align}
	where in the last equation, $\mathbb{F}_{B_1}$ denotes the flip operator acting on the tensor product of two copies of subsystem $B_1$ and $\mathbb{I}_{B_2}$ denotes the identity operator acting on the tensor product of two copies of $B_2$. Using the fact that
	\begin{align}
		\underset{U\sim \mu(d_{12})}{\bbE}\left[\left(U\psi_0 U^\dag\right)^ {\otimes 2} \right]
		=\frac{\mathbb{I}_{B_1}\otimes\mathbb{I}_{B_2}+\mathbb{F}_{B_1}\otimes\mathbb{F}_{B_2}}{d_{12}(d_{12}+1)},
	\end{align}
	we can get the following result
	\begin{align}
		\underset{\ket{\psi}\sim \mu(\caH_{B_1B_2})}{\bbE}
		\left[\left\|\psi_{B_1}-\widetilde{\md{1}}_{B_1}\right\|_2^2\right]
		=\frac{d_1+d_2}{d_{12}+1}-\frac{1}{d_1}
		= \frac{d_1^2-1}{d_1(d_{12}+1)}.
		\label{eq:proof_expectation_psi_identity_2norm}
	\end{align}
\end{proof}

\subsubsection{Proof of \thref{theo:trace_identity_bipartite}}

\begin{proof}
	According to \lref{lem:epsilon_net}, for any subspace $\caH_C\in \mathrm{Gr}(\caH_{B_1B_2},d_C)$, there exists a $\epsilon$-net $\caL_C$ of $\caH_C$ with $|\caL_C|=(5/\epsilon)^{2d_C}$ and the $2$-Lipschitz function $R(\ket{\psi})$ defined in \eref{eq:def_function_f} satisfies
	\begin{align}
		\subinaMID_{B_1}(\caH_C)
		=\max_{\ket{\psi}\in \caH_C} R(\ket{\psi})
		\le \max_{\ket{\widetilde{\psi}}\in \caL_C} R(\ket{\widetilde{\psi}}) + \epsilon.
	\end{align}
	Hence, the probability of a random subspace $\caH_C$, such that there exists a state $\ket{\psi}\in \caH_C$ satisfying $R(\psi)\geq r+\alpha$ with $0<\alpha<1$ and $r$ being an undetermined constant, is 
	\begin{align}
		\mathrm{Pr}\left\{\max_{\ket{\psi}\in \caH_C} R(\ket{\psi}\geq r+\alpha\right\}
		&\leq \mathrm{Pr}\left\{\max_{\ket{\widetilde{\psi}}\in \caL_C} R(\ket{\widetilde{\psi}}) \ge r+\alpha-\epsilon\right\} 
		=\mathrm{Pr}\left\{\underset{\ket{\widetilde{\psi}}\in \caL_C}{\bigvee}
		R(\ket{\widetilde{\psi}})\geq r+\alpha-\epsilon\right\} \notag\\
		&\leq \sum_{\ket{\widetilde{\psi}}\in \caL_C}\mathrm{Pr}\left\{R(\ket{\widetilde{\psi}})\geq r+\alpha-\epsilon\right\},
	\end{align}
	where $\bigvee$ denotes the union of events and in the third inequality, we have used the union bound. Since $\caH_C$ is a random subspace and $\ket{\widetilde{\psi}}$ is a random state in $\caH_C$, then each $\ket{\widetilde{\psi}}$ is a random state in $\caH_{B_1B_2}$ sampled uniformly. Thus, we can substitute \lref{lem:concentrate_of_identity_function} into the last line of the above equation and let $r=\sqrt{(d_1^2-1)(d_{12}+1)}$ and $ \epsilon=\alpha/2$, which yields
	\begin{align}
		\mathrm{Pr}\left\{\max_{\ket{\psi}\in \caH_C} R(\ket{\psi}\geq \sqrt{\frac{d_1^2-1}{d_{12}+1}}+\alpha\right\}&\leq |\caL_C|\mathrm{Pr}\left\{R(\ket{\widetilde{\psi}})\geq \sqrt{\frac{d_1^2-1}{d_{12}+1}}+\frac{\alpha}{2}\right\} \notag\\
		&\leq 2\left(\frac{10}{\alpha}\right)^{2d_C}\exp\left( -\frac{d_{12}}{72\pi^3\ln{2}}\alpha^2 \right),
	\end{align}
	and this completes the proof.
\end{proof}

% \begin{lemma}
	% \label{lem:trace_identity_operator}
	% For a Haar random pure state $\ket{\psi}\in \caH_{B_1\cdots B_m}$, there exists
	%     \begin{align}
		%         \Pr\left( \left\| \psi_S-\frac{\mathds{1}_S}{d_S} \right\| \ge \alpha+\beta \right)
		%         \le 2\exp^{ -c_0D\alpha^2 }
		%     \end{align}
	%     where $\alpha\ge 0$, $\beta=\sqrt{d_S/d_{S^c}}$,
	%     $D=d_1\cdots d_{B_m}$ and $c_0$ can be $1/(18\pi^3\ln2)$.
	% \end{lemma}

\subsection{Concentration result on typical subspaces in the case of a subspace projector}
\label{app:projector}

\subsubsection{Concentration results on the trace distance with respect to a subspace projector}
\label{app:projector_trace_distance}

Before proving \thref{theo:trace_projector_bipartite}, we first define a function and prove some properties of it, which will be used in the proof.
Given a subspace $\caH_C$ of $\caH_{B_1B_2}$ with projector $\Pi_C$ and a pure state $\ket{\psi}\in \caH_C$, we define the function $g(\ket{\psi}, \Pi_C)$ as
\begin{align}
	\label{Eq:trace_projector_function}
	g(\ket{\psi}, \Pi_C)
	:= D(\psi_{B_1}, \Pi_C^{(B_1)})
	=\left\| \psi_{B_1}-\Pi_C^{(B_1)} \right\|_1.
\end{align}
It's obvious that the function $g(\ket{\psi}, \Pi_C)$ depends on two variables, the subspace $\caH_C$ and the state $\ket{\psi}$.
% Then the following two lemmas tells us the Lipschitz constants of $g(\ket{\psi}, \caH_C)$ with respect to each different variable. 

\begin{lemma}
	Given a fixed subspace $\caH_C$ of $\caH_{B_1B_2}$ with projector $\Pi_C$, then the Lipschitz constant of $g(\ket{\psi}, \Pi_C)$ with respect to $\ket{\psi}$ is $2$.
\end{lemma}
\begin{proof}
	According to the definition of Lipschitz constant, we consider the following computation for two states $\ket{\psi}, \ket{\phi}\in \caH_C$ 
	\begin{align}
		|g(\ket{\psi},\Pi_C)-g(\ket{\phi},\Pi_C)|
		=\left|\left\|\psi_{B_1}-\Pi_C^{(B_1)}\right\|_1-\left\|\phi_{B_1}-\Pi_C^{(B_1)}\right\|_1\right| 
		\le \|\psi_{B_1}-\phi_{B_1}\|_1
		\le\|\psi-\phi\|_1\leq2\||\psi\rangle-|\phi\rangle\|_2, \notag
	\end{align}  
	where the computation is the same as \eref{Eq:Lipschitz_function_identity}.
	Hence, the Lipschitz constant of the function $g(\ket{\psi},\Pi_C)$ with respect to the variable $\ket{\psi}$ is $2$.
\end{proof}

% In the definition of $g(\ket{\psi}, \caH_C)$, $\ket{\psi}$ changes with $\caH_C$. However, in order to know the effect of the variation of $\caH_C$ on $g(\ket{\psi}, \caH_C)$, we need to fix the state $\ket{\psi}$ in some sense.
% Hence, 
Given a fixed subspace $\caH_{C_0}$ in the Grassmannian $\mathrm{Gr}(\caH_{B_1B_2},d_C)$ and a fixed state $\ket{\psi_0}$ in $\caH_{C_0}$, any subspace $\caH_{C}\in\mathrm{Gr}(\caH_{B_1B_2},d_C)$ is equal to the fixed subspace $\caH_{C_0}$ up to a unitary $U$ as $\caH_C=U\caH_{C_0}$. At the same time, $\ket{\psi_0}$ is transformed into another state $\ket{\psi}=U\ket{\psi_0}\in \caH_C$. Given $\caH_{C_0}$ and $\ket{\psi_0}$, we can define a function of the unitary $U\in \mathrm{SU}(d_{12})$ as
\begin{align}
	\label{Eq:h_function}
	h(U):=g(\ket{\psi}, \Pi_{C})=g(U\ket{\psi_0}, U\Pi_{C_0}U^{\dagger}).
\end{align}

To derive the concentration inequality for $h(U)$, we need the general concentration result about a function defined on $\mathrm{SU}(d)$.
Such result about $1$-Lipschitz function has been proven in \cite{aubrun2017alice} by using the Log-Sobolev inequalities and Herbst’s argument, below we
show a more general result about the $\kappa$-Lipschitz function, which can be proved easily using the same method in \cite{aubrun2017alice}. Here, the metric we use is simply defined by
\begin{align}
	M_{\mathrm{SU}}(U,W)=\|U-W\|_2,
\end{align}
for $U,W\in \mathrm{SU}(d)$.

\begin{lemma}
	\label{lem:concentration_SU}
	Suppose $f$ is a $\kappa$-Lipschitz function defined on the special unitary group $\mathrm{SU}(d)$, then for $\alpha\ge 0$,
	\begin{align}
		\Pr\left\{ f(U) \ge \bbE[f]+\alpha \right\}
		\le \exp\left(-\frac{d\alpha^2}{4\kappa^2}\right),
	\end{align}
	where $\bbE[f]$ is the expectation value of $f$ with respect to the Haar measure on $\mathrm{SU}(d)$.
\end{lemma}
Based on the above theorem, we have the following concentration result on the function $h(U)$ on $\mathrm{SU}(d_{12})$.

\begin{lemma}
	\label{lem:concentrate_of_unitary_group}
	Given a subspace $\caH_{C_0}\in \mathrm{Gr}(\caH_{B_1B_2},d_C)$ and a state $\ket{\psi_0}$ in $\caH_{C_0}$, the following statements are valid.
	(1) The Lipschitz constant of $h(U)$ is $4$.
	(2) The function $h(U)$ satisfies the following concentration inequality, 
	\begin{align}
		\mathrm{Pr}\left\{ h(U) \ge \bbE[h] + \alpha \right\} \le \exp\left( -\frac{d_{12}}{64}\alpha^2\right),
	\end{align}
	where $\alpha\ge 0$  and 
	\begin{align}
		\label{Eq:expection_projector_upper_bound}
		\bbE[h]\le \sqrt{\frac{d_C-1}{d_C}\frac{d_{12}(d_1^2-1)}{d_{12}^2-1}}.
	\end{align}
\end{lemma}

\begin{proof}
	First, we compute the Lipschitz constant of $h(U)$. Let $\Delta\psi_0:=\psi_0-\Pi_{C_0}/d_C$ and $U,W\in \mathrm{SU}(d_{12})$, then
	\begin{align}
		\label{Eq:lipschitz_projector}
		|h(U)-h(W)| &= |\|\mathrm{Tr}_{B_1}(U\Delta\psi_0U^\dagger)\|_1-\|\mathrm{Tr}_{B_1}(W\Delta\psi_0 W^\dagger)\|_1| \notag\\
		&\le \|\mathrm{Tr}_{B_1}(U\Delta\psi_0U^\dagger) - \mathrm{Tr}_{B_1}(W\Delta\psi_0 W^\dagger)\|_1 \notag \\
		&\leq\|U\Delta\psi_0U^\dagger - W\Delta\psi_0 W^\dagger\|_1 \notag\\
		&= \|(U-W)\Delta\psi_0U^\dagger + W\Delta\psi_0(U-W)^\dagger\|_1 \notag\\
		&\leq \|U-W\|_2\|\Delta\psi_0\|_1\|U\|_\infty+\|U-W\|_2\|\Delta\psi_0\|_1\|W\|_\infty \notag\\
		&\le 4\|U-W\|_2, 
	\end{align}
	where we have used $\|ABC\|_1\le \|A\|_{\infty}\|B\|_1\|C\|_{\infty}$ for three operators $A,B,C$, and in the last inequality we used $\|\Delta\psi_0\|_1=2-2/d_C<2$ and $\|U\|_\infty=\|W\|_\infty=1$, then we conclude that the Lipschitz constant $h(U)$ is 4.

	Next, we need to compute the expectation value of $h(U)$ as
	\begin{align}
		\underset{U\sim \mu(d_{12})}{\bbE}h(U)
		&=\underset{U\sim \mu(d_{12})}{\bbE}\left\|\mathrm{Tr}_{B_2}(U\psi_0U^\dag)-\mathrm{Tr}_{B_2}\left(U\Pi_CU^\dag\right)\right\|_1 \notag\\
		&\leq \sqrt{d_1} \sqrt{ \underset{U\sim \mu(d_{12})}{\bbE}\left\|\mathrm{Tr}_{B_2}(U\psi_0U^\dag)-\mathrm{Tr}_{B_2}\left(U\Pi_CU^\dag\right)\right\|_2^2} \notag\\
		&=\sqrt{\frac{d_C-1}{d_C}\frac{d_{12}(d_1^2-1)}{d_{12}^2-1}},
	\end{align}
	where in the second line, we have used the following \pref{Prop:expectation_2_norm_projector}. Finally, we directly apply \lref{lem:concentration_SU} to $h(U)$ and obtain the desired concentration inequality.
\end{proof}

Additionally, it is worth emphasizing that the expectation value of $h(U)$ is indepentent with $\ket{\psi_0}$, and therefore $\bbE_{U\sim \mu(d_{12})} h(U)$ is exactly the same as the expectation value of $\bbE_{\caH_C,\ket{\psi}}[D(\psi_{B_1}, \Pi_C^{(B_1)})]$ with respect to a uniform random subspace $\caH_C\in  \mathrm{Gr}(\caH_{B_1B_2},d_C)$ and a Haar random pure state $\ket{\psi}\in \caH_C$, i.e.
\begin{align}
	\underset{\caH_C}{\bbE} [\subinaA_{B_1}(\caH_C)]
	&=\underset{\caH_C}{\bbE}\,\underset{\ket{\psi}\sim \mu(\caH_C)}{\bbE}\,[D(\psi_{B_1}, \Pi_C^{(B_1)})]
	=\underset{U\sim \mu(d_{12})}{\bbE}\underset{W\sim \mu(d_C)}{\bbE}g(UW\ket{\psi_0}, U\Pi_{C_0}U^{\dagger}) \notag\\
	&=\underset{U\sim \mu(d_{12})}{\bbE}g(U\ket{\psi_0}, U\Pi_{C_0}U^{\dagger}) 
	=\underset{U\sim \mu(d_{12})}{\bbE}h(U).
	\label{eq:proof_expectation_psi_pi_upper_bound}
\end{align}

\begin{proposition}
	\label{Prop:expectation_2_norm_projector}
	Suppose $\caH_{C_0}$ is a fixed subspace of $\caH_{B_1B_2}$ with projector $\Pi_{C_0}$ and $\ket{\psi_0}$ is a fixed state in $\caH_{C_0}$. Then
	\begin{align}
		\underset{U\sim \mu(d_{12})}{\bbE}\left\|\mathrm{Tr}_{B_2}(U\psi_0 U^\dag)-\mathrm{Tr}_{B_2}\left(U\Pi_{C_0}U^\dag\right)\right\|_2^2 
		= \frac{d_C-1}{d_C}\frac{d_2(d_1^2-1)}{d_{12}^2-1}.
	\end{align}
\end{proposition}

\begin{proof}
	Similar to the proof of \pref{Prop:expectation_2_norm_identity},
	\begin{align}
		\left\|\mathrm{Tr}_{B_2}(U\psi_0 U^\dag)-\mathrm{Tr}_{B_2}\left(U\Pi_{C_0}U^\dag\right)\right\|_2^2 
		&=\mathrm{Tr}\left[\left(U\psi_0 U^\dag\right)^{\otimes 2} \mathbb{F}_{B_1} \otimes \mathbb{I}_{B_2}\right]
		+\frac{1}{d_C^2}\mathrm{Tr}\left[\left(U\Pi_{C_0} U^\dag\right)^{\otimes 2} \mathbb{F}_{B_1} \otimes \mathbb{I}_{B_2}\right]\notag\\
		&\quad-\frac{2}{d_C}\mathrm{Tr}\left[\left(U\psi_0 U^\dag \otimes U\Pi_{C_0} U^\dag \right) \mathbb{F}_{B_1} \otimes \mathbb{I}_{B_2}\right].
	\end{align}
	Since we have the following equations,
	\begin{align}
		&\underset{U\sim \mu(d_{12})}{\bbE}\left[\left(U\psi_0 U^\dag\right)^ {\otimes 2} \right]
		=\frac{\mathbb{I}_{B_1}\otimes\mathbb{I}_{B_2}+\mathbb{F}_{B_1}\otimes\mathbb{F}_{B_2}}{d_{12}(d_{12}+1)}, \notag\\
		&\underset{U\sim \mu(d_{12})}{\bbE}\left[\left(U\Pi_{C_0} U^\dag\right)^ {\otimes 2} \right]=\frac{d_{12}d_C^2-d_C}{d_{12}(d_{12}^2-1)}\mathbb{I}_{B_1}\otimes\mathbb{I}_{B_2}+\frac{d_{12}d_C-d_C^2}{d_{12}(d_{12}^2-1)}\mathbb{F}_{B_1}\otimes\mathbb{F}_{B_2}, \\
		&\underset{U\sim \mu(d_{12})}{\bbE}\left[U^{\otimes2}\left(\psi_0\otimes\Pi_{C_0} \right){U^\dag}^{\otimes2} \right]=\frac{d_{12}d_C-1}{d_{12}(d_{12}^2-1)}\mathbb{I}_{B_1}\otimes\mathbb{I}_{B_2}+\frac{d_{12}-d_C}{d_{12}(d_{12}^2-1)}\mathbb{F}_{B_1}\otimes\mathbb{F}_{B_2}.\notag
	\end{align}
	then
	\begin{align}
		\label{Eq:some_expectation_values}
		&\underset{U\sim \mu(d_{12})}{\bbE}\mathrm{Tr}\left[\left(U\psi_0 U^\dag\right)^ {\otimes 2} \right]=\frac{d_1+d_2}{d_{12}+1},\notag \\
		&\underset{U\sim \mu(d_{12})}{\bbE}\mathrm{Tr}\left[\left(U\Pi_{C_0} U^\dag\right)^ {\otimes 2} \right]
		=\underset{U\sim \mu(d_{12})}{\bbE}\frac{1}{d_C} \mathrm{Tr}\left[U^{\otimes2}\left(\psi_0\otimes\Pi_{C_0} \right){U^\dag}^{\otimes2} \right], \\
		&\underset{U\sim \mu(d_{12})}{\bbE}\mathrm{Tr}\left[U^{\otimes2}\left(\psi_0\otimes\Pi_{C_0} \right){U^\dag}^{\otimes2} \right]
		= \frac{d_{12}^2d_C^2-d_2d_C}{d_{12}^2-1 } + \frac{d_1^2d_2d_C-d_1d_C^2}{d_{12}^2-1 }.\notag
	\end{align}
	Finally, we get
	\begin{align}
		\underset{U\sim \mu(d_{12})}{\bbE}\left\|\mathrm{Tr}_{B_2}(U\psi_0 U^\dag)-\mathrm{Tr}_{B_2}\left(U\Pi_{C_0}U^\dag\right)\right\|_2^2 
		= \frac{d_C-1}{d_C}\frac{d_2(d_1^2-1)}{d_{12}^2-1}.
	\end{align}
\end{proof}
Again, similar to \eref{eq:proof_expectation_psi_pi_upper_bound}, the above result also implies that
\begin{align}
	\underset{\caH_C}{\bbE}\,\underset{\ket{\psi}\sim\mu(\caH_C)}{\bbE}\,\left\|\psi_{B_1}, \Pi_C^{(B_1)}\right\|_2^2
	&=\underset{U\sim \mu(d_{12})}{\bbE}\underset{W\sim \mu(d_C)}{\bbE}\left\|\mathrm{Tr}_{B_2}(U W\psi_0 W^\dag U^\dag)-\mathrm{Tr}_{B_2}\left(U\Pi_CU^\dag\right)\right\|_2^2 \notag \\
	&=\underset{U\sim \mu(d_{12})}{\bbE}\left\|\mathrm{Tr}_{B_2}(U\psi_0 U^\dag)-\mathrm{Tr}_{B_2}\left(U\Pi_CU^\dag\right)\right\|_2^2  \notag \\
	&= \frac{d_C-1}{d_C}\frac{d_2(d_1^2-1)}{d_{12}^2-1}.
	\label{eq:proof_expectation_psi_pi_2norm}
\end{align}

\subsubsection{Proof of \thref{theo:trace_projector_bipartite}}

\begin{proof}
	Given a $d_C$-dimensional subspace $\caH_C\in \mathrm{Gr}(\caH_{B_1B_2},d_C)$, according to \lref{lem:epsilon_net}, we can construct a $\epsilon$-net $\caL_C$ of $\caH_C$ with $|\caL_C|=(5/\epsilon)^{2d_C}$ and the $2$-Lipschitz function $g(\ket{\psi}, \Pi_C)$ satisfies
	\begin{align}
		\subinaM_{B_1}(\caH_C)
		=\max_{\ket{\psi}\in \caH_C} g(\ket{\psi}, \Pi_C)
		\le \max_{\ket{\widetilde{\psi}}\in \caH_C} g(\ket{\widetilde{\psi}}, \Pi_C) + \epsilon.
	\end{align}
	Hence, the probability of a random subspace $\caH_C$, such that there exists a state $\ket{\psi}\in \caH_C$ satisfying $g(\psi,\Pi_C)\geq s+\alpha$ with $0<\alpha<1$ and $s$ being an undetermined constant, is 
	\begin{align}
		\mathrm{Pr}\left\{\max_{\ket{\psi}\in \caH_C} g(\ket{\psi}, \Pi_C)\geq s+\alpha\right\}
		&\leq \mathrm{Pr}\left\{\max_{\ket{\widetilde{\psi}}\in \caH_C} g(\ket{\widetilde{\psi}}, \Pi_C) \ge s+\alpha-\epsilon\right\} \notag\\
		&=\mathrm{Pr}\left\{\underset{\ket{\widetilde{\psi}}\in \caL_C}{\bigvee}
		g(\ket{\widetilde{\psi}},\Pi_C)\geq s+\alpha-\epsilon\right\} \notag\\
		&\leq \sum_{\ket{\widetilde{\psi}}\in \caL_C}\mathrm{Pr}\left\{g(\tilde{\psi},\Pi_C)\geq s+\alpha-\epsilon\right\}
	\end{align}
	where $\bigvee$ denotes the union of events and in the third inequality, we have used the union bound. Let $\epsilon=\alpha/2$, then we have
	\begin{align}
		\mathrm{Pr}\left\{\max_{\ket{\psi}\in \caH_C} g(\ket{\psi}, \Pi_C)\geq s+\alpha\right\}
		\leq \sum_{\ket{\widetilde{\psi}}\in \caL_C}\mathrm{Pr}\left\{g(\tilde{\psi},\Pi_C)\geq s+\frac{\alpha}{2}\right\}
		\label{eq:append_B_prob_ineq_1}
	\end{align}
	So the next step is to choose an appropriate parameter $s$ and calculate an upper bound of the probability appearing in the right-hand side.
	% $\mathrm{Pr}\left\{g(\tilde{\psi},\caH_C)\geq E+\alpha/2\right\}$ is just the probability of a random subspace $\caH_C$ with a particular state $\ket{\psi}\in \caH_C$ satisfying $g(\tilde{\psi},\caH_C)\geq E+\alpha/2$, so 
	
	Let's fix a $d_C$-dimensional subspace $\caH_{C_0}\in \mathrm{Gr}(\caH_{B_1B_2},d_C)$ and a state $\ket{\psi_0}\in \caH_{C_0}$ satisfying $\caH_{C}=U\caH_{C_0}$ and $\ket{\psi}=U\ket{\psi_0}$ with $U\in \mathrm{SU}(d_{12})$, then $h(U)=g(\tilde{\psi},\caH_C)$ according to \eref{Eq:h_function}. Hence, the probability becomes
	\begin{align}
		\mathrm{Pr}\left\{g(\tilde{\psi},\Pi_C)\geq s+\frac{\alpha}{2}\right\}
		=\mathrm{Pr}\left\{h(U)\geq s+\frac{\alpha}{2}\right\}.
	\end{align}
	The above procedure has changed the probability about a random subspace in the Grassmannian $\mathrm{Gr}(\caH_{B_1B_2},d_C)$ and a random state into another probability about a random unitary in $\mathrm{SU}(d_{12})$.
	According to \lref{lem:concentrate_of_unitary_group}, let $s$ be the upper bound of $\bbE(h)$ in \eref{Eq:expection_projector_upper_bound}, then we get
	\begin{align}
		\mathrm{Pr}\left\{g(\tilde{\psi},\Pi_C)\geq s+\frac{\alpha}{2}\right\}
		\le
		\mathrm{Pr}\left\{h(U)\geq \bbE(h)+\frac{\alpha}{2}\right\}
		\le \exp\left( -\frac{d_{12}}{256}\alpha^2\right).
		\label{eq:append_B_prob_ineq_2}
	\end{align}
	Finally, combining \eref{eq:append_B_prob_ineq_1} and \eref{eq:append_B_prob_ineq_2}, we get
	\begin{align}
		\mathrm{Pr}\left\{\max_{\ket{\psi}\in \caH_C} g(\ket{\psi}, \Pi_C)\geq \sqrt{\frac{d_C-1}{d_C}\frac{d_{12}(d_1^2-1)}{d_{12}^2-1}}+\alpha\right\}
		&\leq\left(\frac{10}{\alpha}\right)^{2d_C}\exp\left(-\frac{d_{12}\alpha^2}{256}\right).
	\end{align}
	where we used $|\caL_C|=(5/\epsilon)^{2d_C}=(10/\alpha)^{2d_C}$.
\end{proof}

\subsection{Proof of \pref{prop:trace_projector_identity_bipartite}}
\label{appendix:trace_identity_projector}

\begin{proof}
	As in the proof of \thref{theo:trace_projector_bipartite}, we also fix a subspace $\caH_{C_0}$ in the Grassmannian $\mathrm{Gr}(\caH_{B_1B_2},d_C)$ and transform the trace distance $D(\Pi_C^{(B_1)},\widetilde{\md{1}}_{B_1})$ into a function of $U\in \text{SU}(d_{12})$.
	Then we compute the Lipschitz constant of $D(\Pi_C^{(B_1)},\widetilde{\md{1}}_{B_1})=D(\Tr_{B_2}(U \Pi_{C_0} U^{\dagger}),\widetilde{\md{1}}_{B_1})$ as a function of unitary $U\in \mathrm{SU}(d_{12})$. Suppose $U,W\in \text{SU}(d_{12})$, then
	\begin{align}
		&\left|D(\Tr_{B_2}(U \Pi_{C_0} U^{\dagger}),\widetilde{\md{1}}_{B_1}) - D(\Tr_{B_2}(W\Pi_{C_0} W^{\dagger}),\widetilde{\md{1}}_{B_1})  \right|\notag\\
		&\le \left \| \Tr_{B_2}(U \Pi_C U^{\dagger}) -  \Tr_{B_2}(W\Pi_{C_0} W^{\dagger})    \right\|_1  \notag\\
		&\le \left \| U \Pi_{C_0} U^{\dagger} -  W\Pi_{C_0} W^{\dagger}  \right\|_1  \notag\\
		&= \left \| (U-W) \Pi_{C_0} U^{\dagger} -  W\Pi_{C_0} (U^{\dagger}-W^{\dagger})  \right\|_1  \notag\\
		&\le \| U-W\|_2 \|\Pi_{C_0}\|_1 \|U^{\dagger}\|_{\infty} + \| W\|_{\infty} \|\Pi_{C_0}\|_1 \|(U^{\dagger}-W^{\dagger})  \|_2  \notag\\
		&\le 2\|U-W\|_2,
	\end{align}
	where we have used the $\|\Pi_C\|_1=1$ and $\|U^{\dagger}\|_{\infty}=1$. The above computation is similar to that in \eref{Eq:lipschitz_projector}. Next, we consider the upper bound of the expectation value of $D(\Pi_C^{(B_1)},\widetilde{\md{1}}_{B_1})$. According to \eref{Eq:some_expectation_values}, for $\caH_C\in \mathrm{Gr}(d_{12},d_C)$
	\begin{align}
		\underset{\caH_C}{\bbE} \left[\left\| \Pi_C^{(B_1)}-\widetilde{\md{1}}_{B_1}\right\|_2^2\right] 
		&=\frac{1}{d_C^2} \underset{\caH_C}{\bbE}\left[\Tr\left(\Pi_C^{(B_1)}\Pi_C^{(B_1)} \right) \right]  -\frac{1}{d_1} \notag \\
		&=\frac{1}{d_C^2}\underset{U\sim \mu(d_{12})}{\bbE}\mathrm{Tr}\left[U^{\otimes2}\left(\psi_0\otimes\Pi_{C_0} \right){U^\dag}^{\otimes2} \right]-\frac{1}{d_1}
		\notag \\
		&=\frac{d_{12}-d_C}{d_1d_C}\frac{d_1^2-1}{d_{12}^2-1}. \label{eq:proof_expectation_pi_identity_2norm}
	\end{align}
	Then we can get
	\begin{align}
		\underset{\caH_C}{\bbE}\left[ D(\Pi_C^{(B_1)},\widetilde{\md{1}}_{B_1})
		\right]
		\le \sqrt{d_1} \sqrt{\underset{\caH_C}{\bbE} \left\| \Pi_C^{(B_1)}-\widetilde{\md{1}}_{B_1} \right\|_2^2  } 
		= \sqrt{\frac{d_{12}-d_C}{d_C}\frac{d_1^2-1}{d_{12}^2-1}}. \label{eq:proof_expectation_pi_identity_upper_bound}
	\end{align}
	Hence, we can directly obtain the following concentration inequality by applying \lref{lem:concentration_SU}
	\begin{align}
		\Pr\left\{ D( \Pi_{C}^{(B_1)}, \widetilde{\md{1}}_{B_1}) > t + \alpha\right\} 
		\le
		\Pr\left\{ D( \Pi_{C}^{(B_1)}, \widetilde{\md{1}}_{B_1}) > \underset{\caH_C}{\bbE}\left[ D(\Pi_C^{(B_1)},\widetilde{\md{1}}_{B_1})
		\right] + \alpha\right\}
		\le
		\exp\left(-\frac{d_{12}\alpha^2}{16}\right),
	\end{align}
	with $t=\sqrt{\frac{d_{12}-d_C}{d_C}\frac{d_1^2-1}{d_{12}^2-1}}$.
\end{proof}

\section{Consequences of \thref{theo:subpsace_identity_multipartite}}
\label{app:consequences_multipartite}

Here we primarily address this question: Given a fixed $d_C$ and a fixed inaccuracy $\delta$, what conditions must $k$ and $m$ satisfy?
This question corresponds to the following scenario: Given a Hilbert space $\caH_A$ with a known dimension $d_A = d_C$, we aim to mask it into a multipartite Hilbert space $\caH_{B_1 \dots B_m}$ up to a chosen inaccuracy $\delta$, then we ask what is the minimum number of parties $m$ required to achieve such masking, and what level of $k$-uniform masking can be realized?

According to \thref{theo:subpsace_identity_multipartite}, the inaccuracy $\delta$ is given by $d^{k-m/2} + \alpha$. Since $\alpha$ can be freely chosen, its value relative to $d^{k-m/2}$ leads to three possible cases:
(1) $\alpha$ is a fixed small constant greater than $d^{k-m/2}$.
(2) $\alpha$ is larger than $d^{k-m/2}$ but depends on $k$ and $m$.
(3) $\alpha$ is smaller than $d^{k-m/2}$.
In the following, we discuss each case in detail.

\begin{figure}
	\centering
	\includegraphics[width=1.0\linewidth]{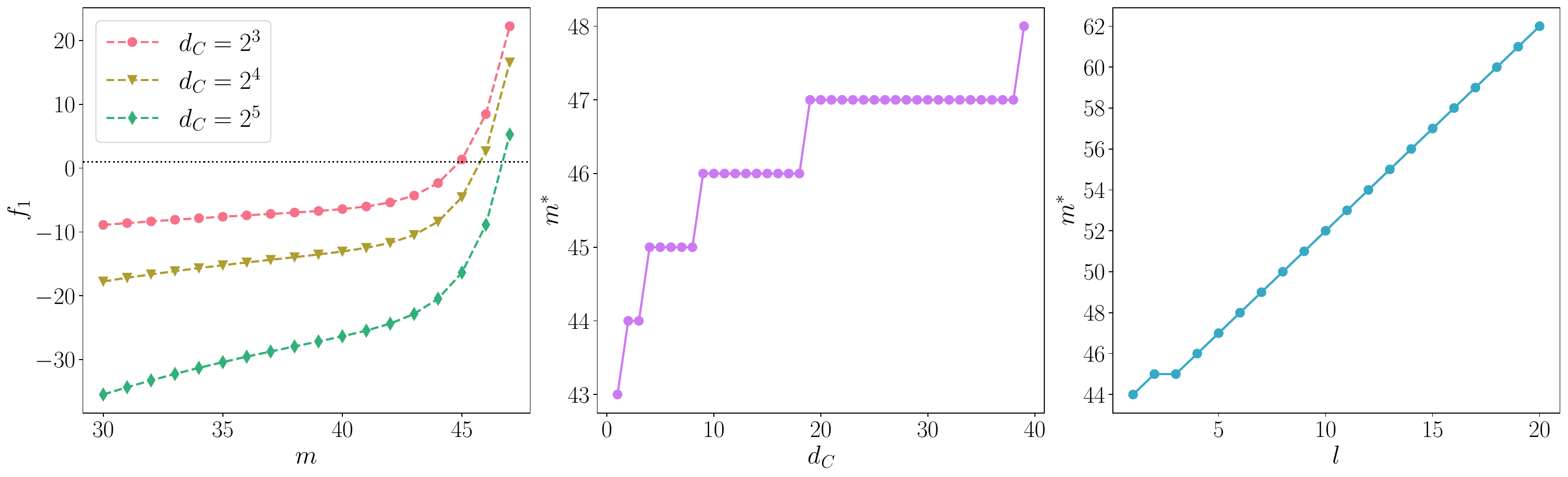}
	\caption{Illustration of Case $1$. For all three figures, we set $\alpha=0.0001$, $d=2$. The left figure represents the variation of $f_1$ with respect to $m$, given $d_C=2^l$ with $l=3,4,5$, and the black dashed line represents $f_1=1$. Because of the exponential term $d^m$ in $f_1$, the value of $f_1$ increases rapidly. 
		The middle figure represents the variation of $m^*$ with respect to $d_C$. Letting $d_C=2^l$, the right figure shows the threshold $m^*$ changes linearly with respect to $l$. }
	\label{fig:q1_case1}
\end{figure}

First, we consider the case where $\alpha$ is a fixed constant greater than $d^{k-m/2}$. Since $\alpha \ge d^{k-m/2}$, the behavior of the inaccuracy $\delta$ is dominated by $\alpha$, and we get a simple constraint about $k$ and $m$,
\begin{align}
	\label{Eq:case1_inaccuracy_condition}
	k \le \frac{m}{2} +\frac{\ln\alpha}{\ln d}.
\end{align}
Note that the definition of $k$-uniform masking already requires $k\le m/2$, then the above inequality provides a tighter bound for $k$.
For example, if $d=2$ and $\alpha=0.001$, then $k\lessapprox m/2 - 10$, which means we can only realize a $k$-uniform masking in the multipartite Hilbert space $\caH_{B_1 \dots B_m}$ with $k\lessapprox m/2 - 10$.
Next, we require the probability in \eref{Eq:prob_multipartite_idenity} to be positive, i.e.,
there exists a $2\alpha$-maskable subspaces $\caH_C$, then we get 
\begin{align}
	\ln2 + \ln\left[\binom{m}{k} \right]+2d_C\ln\left( \frac{10}{\alpha}\right) -\frac{d^m\alpha^2}{72\pi^3\ln2} \le 0.
\end{align}
Using the inequality $\binom{m}{k}\le 2^{mH(k/m)}$, we get a tighter bound as
\begin{align}
	\label{Eq:case1_H_upper_bound}
	H\left( \frac{k}{m} \right)
	\le f_1(m, d_C):= \frac{1}{(\ln2)m} \left[ \frac{d^m \alpha^2}{72\pi^3\ln2} - 2d_C \ln\left( \frac{10}{\alpha}\right) -\ln 2 \right].
\end{align}
Note that the right-hand side of the above inequality, denoted by $f_1(m, d_C)$, is independent of $k$.                                     
Since $0 \leq H(k/m) \leq 1$, if $f_1$ is negative, a $2\alpha$-maskable subspace cannot exist. Therefore, we must determine the values of $m$ for which the right-hand side is positive. Due to the presence of an exponential term $d^m$, $f_1$ rapidly increases from 0 to values greater than 1, as shown in the left figure of \fref{fig:q1_case1}.
Thus, it suffices to determine the smallest integer $m$ for which $f_1$ exceeds 1, i.e., $m^*(d_C)=\min\{m\in\bbN^+:f_1(m,d_C)\ge 1\}$, which is the threshold of $m$ for the existence of $2\alpha$-maskable subspaces.  
If we consider the qubit system, i.e., $d=2$ and let the subspace dimension be $d_C=2^l$, then the threshold number $m^*$ depends on $l$ linearly as shown in \fref{fig:q1_case1}, because the factor $1/[(\ln2)m]$ is negligible compared with $2^m$ when $m$ is large.
To get the slope of the curve of $m^*$ over $l$, we need to analyze the asymptotic behavior of $m^*$. When $d_C$ becomes very large, the inequality $f_1(m,d_C)\ge 1$ becomes
\begin{align}
	d^{m^*} \ge \frac{72 \pi^3 \ln2}{\alpha^2} \times 2d_C \ln\left( \frac{10}{\alpha}\right),
\end{align}
taking logarithm of both sides, then we get
\begin{align}
	m^* \propto \frac{\ln d_C}{\ln d} = l,
\end{align}
where we have chosen $d_C=d^l$. This result is consistent with numerical results shown in \fref{fig:q1_case1}.

\begin{figure}
	\centering
	\includegraphics[width=1.0\linewidth]{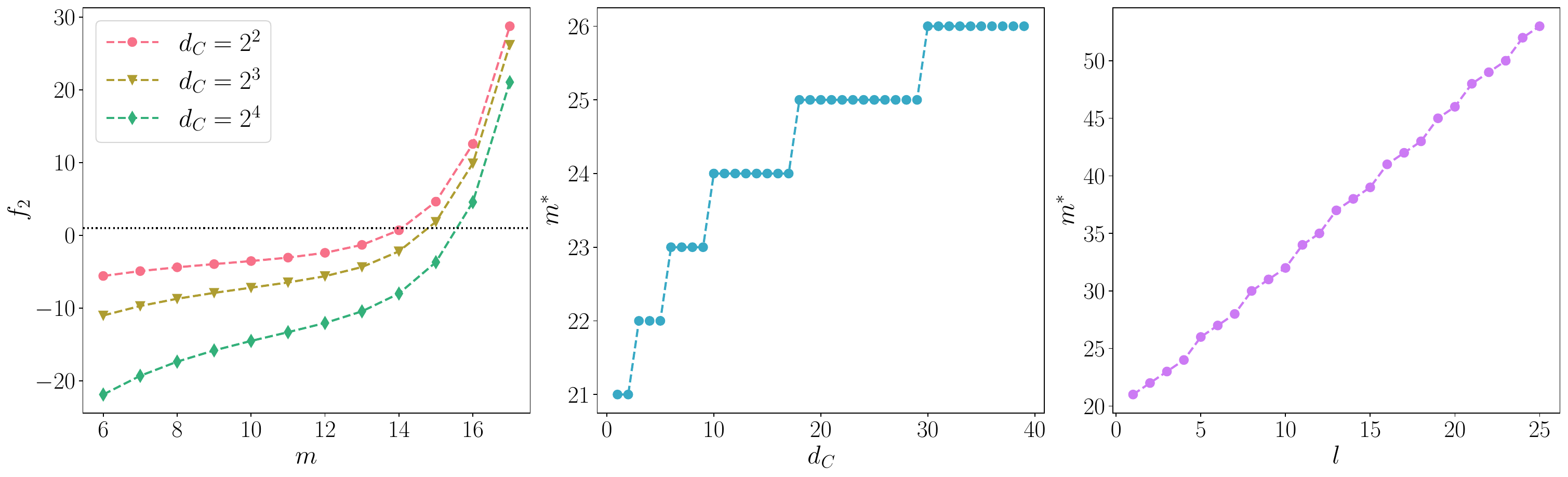}
	\caption{Illustration of Case $2$. For all three figures, we set $\zeta=1/4$, $d=2$. The left figure represents the variation of $f_2$ with respect to $m$ for $d_C=2^l$ with $l=3,4,5$, and the black dashed line represents $f_2=1$. Because of the exponential term $d^m$ in $f_2$, the value of $f_2$ exceeds $1$ rapidly. 
		The middle figure represents the variation of $m^*$ with respect to $d_C$. Letting $d_C=2^l$, the right figure shows the threshold $m^*$ is approximately linear with respect to $l$. }
	\label{fig:q1_case2_m_star}
\end{figure}

Secondly, we consider the case where $\alpha \ge d^{k-m/2}$ and depends on $k, m$. As an example, we let $\alpha=\sqrt{d^{k-m/2}}$, then $\delta=\sqrt{d^{k-m/2}} + d^{k-m/2}\le 2\sqrt{d^{k-m/2}} $ and its behavior is determined by $\sqrt{d^{k-m/2}}$. In order to avoid the situation that the inaccuracy $\delta$ becomes larger with $k$ increasing, 
we set $k=\zeta m$ with a fixed ratio $0<\zeta< 1/2$.\footnote{Otherwise, we will encounter the phenomenon that some $k'$-uniform masking can be realized but $k$-uniform masking can't, even $k'>k$.} 
By requiring the probability in \eref{Eq:prob_multipartite_idenity} to be positive and using the inequality $\binom{m}{k}\le 2^{mH(k/m)}$, we get a tight bound as
\begin{align}
	\label{Eq:case2_H_upper_bound}
	H\left( \frac{k}{m} \right)
	\le \frac{1}{(\ln2)m} \left[ \frac{d^{m(\zeta+1/2)}}{72\pi^3\ln2} - 2d_C \ln10 +m\left(\zeta-\frac{1}{2} \right)d_C\ln d -\ln 2 \right]:=f_2(m,d_C;\zeta) .
\end{align}
This case is similar with the previous one, except the previous condition \eref{Eq:case1_inaccuracy_condition} is replaced by $k=\zeta m$. 
When $\zeta$ is fixed, the behaviors of the function $f_2$ with respect to different $d_C$ are illustrated in \fref{fig:q1_case2_m_star}.
Following the same logic as the previous case, it suffices to consider the smallest integer $m$ for which $f_2$ exceeds 1, i.e., $m^*(d_C;\zeta)=\min\{m\in\bbN^+:f_2(m,d_C; \zeta)\ge 1\}$. 
When $\zeta$ is fixed, as a function of $d_C$, the behavior of $m^*$ is also similar with the previous case, as shown in the middle and right figures in \fref{fig:q1_case2_m_star}. When $d_C$ becomes very large, taking logarithm of both sides of the inequality $f_2(m,d_C;\zeta)\ge 1$ and ignoring smaller terms, we get
\begin{align}
	m^* \propto \frac{\ln d_C}{(\zeta +1/2)\ln d} = \frac{\ln 2}{(\zeta +1/2)\ln d} l,
\end{align}
where we have chosen $d_C=2^l$. When $d=2$, then $m^* \propto l/(\zeta +1/2)$, which is consistent with numerical results shown in \fref{fig:q1_case2_m_star}.

Thirdly, we consider the case where $\alpha\le d^{k-m/2}$ and assume $\alpha = d^{k-m/2}$, then the masking inaccuracy $\delta\le 2d^{k-m/2}$ and its behavior is determined by $d^{k-m/2}$. Similarly, we also require $k=\zeta m$ with $\zeta < 1/2$, then we get an inequality as
\begin{align}
	\label{Eq:case3_H_upper_bound}
	H\left( \frac{k}{m} \right)
	\le \frac{1}{(\ln2)m} \left[ \frac{d^{2\zeta m}}{72\pi^3\ln2} - 2d_C \ln10 +2 m\left(\zeta-\frac{1}{2} \right)d_C\ln d -\ln 2 \right]:=f_3(m,d_C;\zeta) .
\end{align}
The third case is almost the same as the second case and the main difference is the exponential term $d^{m(\zeta+1/2)}$ being replaced by $d^{2m\zeta}$. Hence, the behaviors of $f_3$ and the threshold $m^*(d_C;\zeta)=\min\{m\in\bbN^+:f_3(m,d_C; \zeta)\ge 1\}$ are also similar with previous case.
When $d_C$ becomes very large, taking logarithm of both sides of the inequality $f_2(m,d_C;\zeta)\ge 1$ and ignoring smaller terms, we get
\begin{align}
	m^* \propto \frac{\ln d_C}{2\zeta \ln d} = \frac{\ln 2}{2\zeta \ln d} l,
\end{align}
where we have chosen $d_C=2^l$. When $d=2$, then $m^* \propto l/(2\zeta)$, which is consistent with numerical results shown in \fref{fig:q1_case3_m_star}.

\begin{figure}
	\centering
	\includegraphics[width=1.0\linewidth]{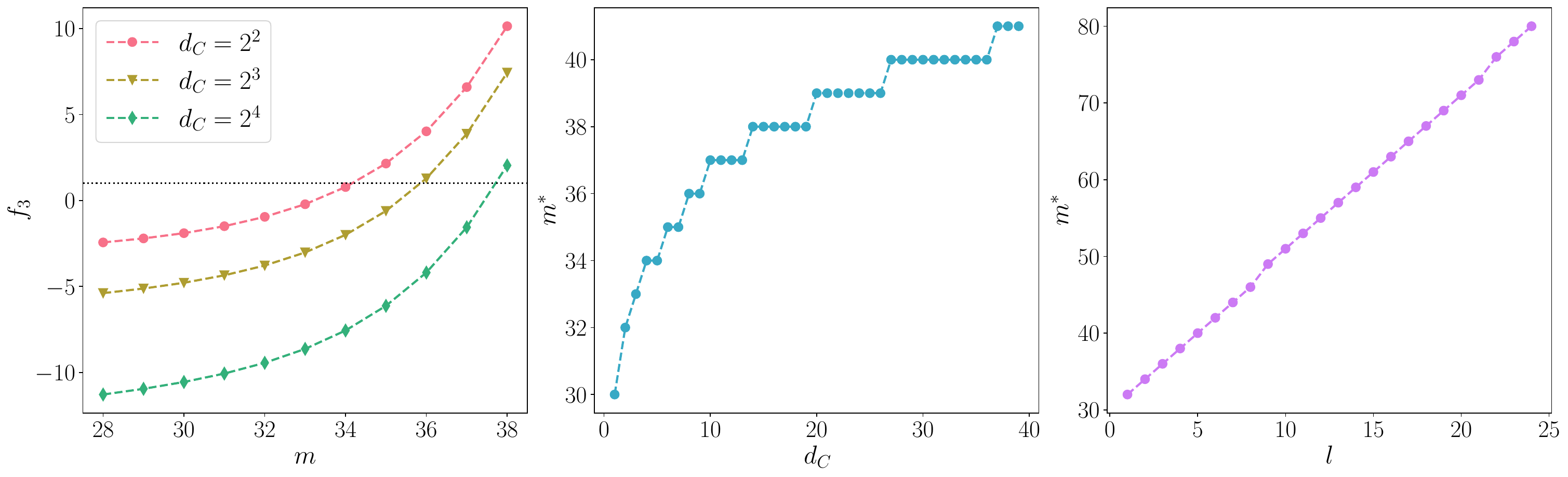}
	\caption{Illustration of Case $3$. For all three figures, we set $\zeta=1/4$, $d=2$. The left figure represents the variation of $f_3$ with respect to $m$ for $d_C=2^l$ with $l=3,4,5$, and the black dashed line represents $f_3=1$. Because of the exponential term $d^m$ in $f_2$, the value of $f_3$ exceeds $1$ rapidly. 
		The middle figure represents the variation of $m^*$ with respect to $d_C$. Letting $d_C=2^l$, the right figure shows the threshold $m^*$ is approximately linear with respect to $l$. }
	\label{fig:q1_case3_m_star}
\end{figure}

\section{Consequences of \thref{theo:random_codes}}
\label{appendix:Consequences_of_random_codes}

Here, we analyze the consequences of \thref{theo:random_codes} following the same logic as in previous sections and consider two key questions: Given the dimension $d_C$ and a fixed QEC inaccuracy $\widetilde{\eta}(\caE, \caR_S)$, what conditions must $k$ and $m$ satisfy?

\subsection{Determining the code distance and the number of physical qudits with inaccuracy being fixed}
First, let's consider the case where the QEC inaccuracy $\widetilde{\eta}(\caE, \caR_S)$ is a fixed small constant $\widetilde{\eta}_0$. Since $\widetilde{\eta}(\caE, \caR_S)=\sqrt{d_C(u+\alpha)}$, 
let $\gamma:=k/m$ satisfy $\gamma<1/2$, then the parameter $\alpha$ becomes 
\begin{align}
	\alpha=\frac{\tilde{\eta}_0^2}{d_C}-u \approx \frac{\tilde{\eta}_0^2}{d_C} -d^{(\gamma-1/2)m},
\end{align}
where we have used the approximation $u\approx d^{k-m/2}$. Since $\alpha>0$, then we get the upper bound of $\gamma$ as
\begin{align}
	\label{Eq:gamma_upper_bound}
	\gamma < \frac{1}{2}+\frac{1}{m}\left( 2\log_d\widetilde{\eta}_0 - \log_dd_C \right)
	=\frac{1}{2}+\frac{1}{m}\left( 2\log_d\widetilde{\eta}_0 - l\right)
\end{align}
where we have used $d_C=d^l$ with $l$ being the number of logical qudits.  
It is obvious that \eref{Eq:gamma_upper_bound} provides a tighter upper bound compared with $\gamma<1/2$.
\eref{Eq:gamma_upper_bound} also implies that $m$ has a lower bound:
\begin{align}
	\label{Eq:aqecc_m_lower1}
	m \ge \frac{2\left( \log d_C-2\log_d\widetilde{\eta}_0 \right)}{1-2\gamma}
	= \frac{2}{1-2\gamma} l - \frac{4}{1-2\gamma}\log_d\widetilde{\eta}_0.
\end{align}
For example, if $d=2$, $l=5$ and $\widetilde{\eta}_0=10^{-3}$, then $\gamma \lessapprox 1/2-25/m$ and $m \gtrapprox 50/(1-2\gamma)$.
Requiring the probability in \eref{Eq:random_code_probability} being positive, we get the inequality
\begin{align}
	\frac{d^m\alpha^2}{256} > \ln\left[ \binom{m}{k} \right] + 2d_C \ln\left( \frac{10}{\alpha} \right).
\end{align} 
Using the inequality $\binom{m}{k}\le 2^{mH(k/m)}$ and $H(k/m)\le 1$ again, we get a tighter bound as
\begin{align}
	\frac{d^m\alpha^2}{256} \ge (\ln2)m + 2d_C \ln\left( \frac{10}{\alpha} \right). 
\end{align}
Then we consider the asymptotic behavior, $d_C=d^l\gg m\gg 1$, $d^{(\gamma-1/2)m}\ll 1$, then $\alpha\approx \tilde{\eta}_0^2/d_C$. Substituting $\alpha\approx \tilde{\eta}_0^2/d_C$ into the above inequality and taking the logarithm, we get 
\begin{align}
	\label{Eq:aqecc_m_lower2}
	m > 2\text{log}_d(d_C) + \text{log}_d\left( \frac{256}{\tilde{\eta}_0^4} \right) + \text{log}_d \left((\ln2)m + 2d_C \ln\left( \frac{10}{\alpha} \right) \right)
	= 3 \text{log}_d (d_C) + \caO(\log\log d_C)
\end{align} 
For simplicity, we define two new functions as
\begin{align}
	T_1^{(1)}(m,d_C) &:=\frac{2}{1-2\gamma} \text{log}_d(d_C) -\frac{4}{1-2\gamma}\log_d\widetilde{\eta}_0 , \\
	T_2^{(1)}(m,d_C) &:=2\text{log}_d(d_C) + \text{log}_d\left( \frac{256}{\tilde{\eta}_0^4} \right) + \text{log}_d \left[ (\ln2)m + 2d_C \ln\left( \frac{10}{\alpha} \right) \right] .
\end{align}
Hence, given the dimension $d_C$, the threshold of $m$ for the existence of a AQECC $\caH_C$ with inaccuracy $\widetilde{\eta}_0$ is given by 
\begin{align}
	m^*:=\min\{m\in \bbN^+: m-T_1^{(1)}(m,d_C)> 0, m-T_2^{(1)}(m,d_C)>0\}.
\end{align}
In other words, $m^*$ is determined by $\max\{T_1(m,d_C), T_2(m,d_C) \}$.
Hence, when $d_C\gg 1$ or $l$ is large, the relative size of $T_1$ and $T_2$ are determined solely by the linear coefficients of the term $\log d_C$, then we get the conclusion: (1) when $1/2>\gamma>1/6$, $m^* \propto 2/(1-2\gamma)l$; (2)  when $\gamma<1/6$, $m^* \propto 3 l$, just as illustrated in \fref{fig:aqecc_case1}. 
So we find that asymptotically, the number of physical qudits $m^*$ only linearly depend on the logical qudits $l$ and
the code rate $l/m^*$ can be a fixed constant once the $\gamma$ is given.

\begin{figure}
	\centering
	\includegraphics[width=1.0\linewidth]{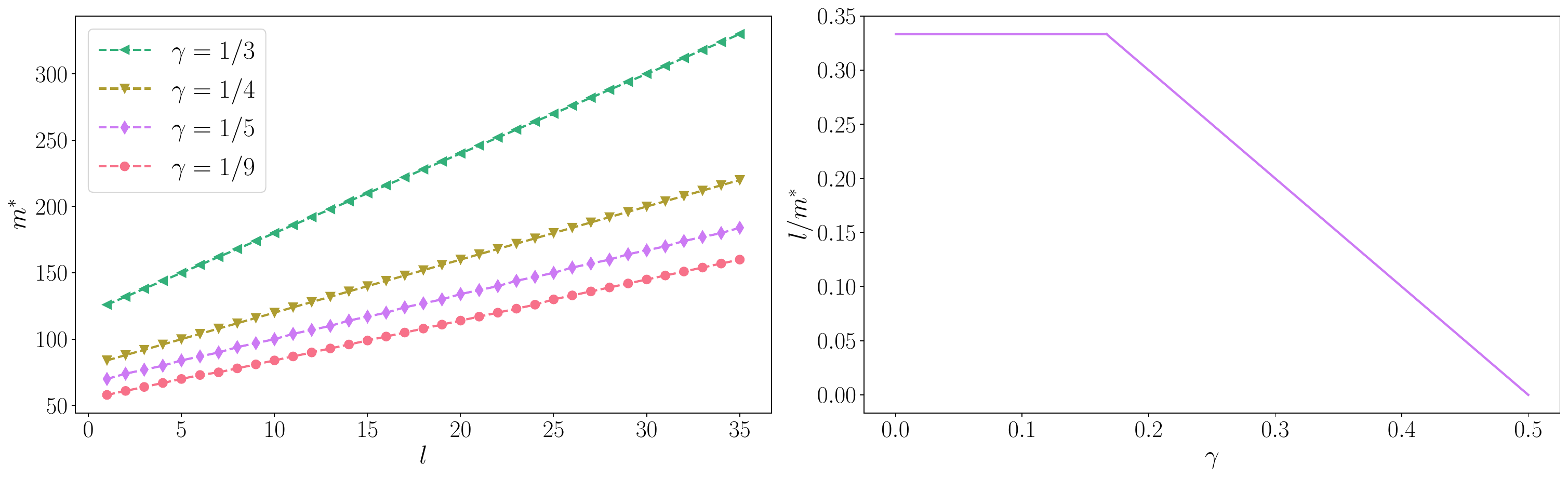}
	\caption{The left plot represents the variations of $m^*$ as a function of $l$ for different $\gamma$, where $\widetilde{\eta}_0=0.001$. The right plot shows the variation of the code rate $l/m^*$ with respect to $\gamma$. }
	\label{fig:aqecc_case1}
\end{figure}

\subsection{Determining the code distance and the number of physical qudits with inaccuracy approaching zero}

Secondly, we consider the case where $\widetilde{\eta}(\caE, \caR_S)$ depends on $k, m$. Since $\widetilde{\eta}(\caE, \caR_S)=\sqrt{d_C(u+\alpha)}$ with $u\approx d^{k-m/2}$, we assume the inaccuracy equals to $d^{a(k-m/2)}$ with $a>0$. In order to avoid the inaccuracy increasing with $k$, we let $k=\gamma m$ with $\gamma<1/2$.
The parameter $\alpha$ becomes 
\begin{align}
	\alpha =  \frac{d^{2a(k-m/2)}}{d_C} - u \approx \frac{d^{2am(\gamma-1/2)}}{d_C} - d^{m(\gamma-1/2)},
\end{align}
because $\alpha>0$, then we can see $a<1/2$ and obtain
\begin{align}
	\gamma < \frac{1}{2} - \frac{1}{m(1-2a)}\log_d d_C = \frac{1}{2} - \frac{l}{m(1-2a)}
\end{align}
where $d_C=d^l$ with $l$ being the number of logical qudits. The above inequality provides an upper bound for $\gamma$, at the same time, it also gives a lower bound for $m$ as
\begin{align}
	\label{Eq:aqeecc_case2_m_lower1}
	m > \frac{1}{(1-2a)(1/2-\gamma)}\log_d d_C = \frac{1}{(1-2a)(1/2-\gamma)} l.
\end{align}
For example, if $a=1/4$, $\gamma < 1/2 - 2l/m$ or $m>2l/(1/2-\gamma)$, and the inaccuracy decays exponentially as $1/d^{m(1/2-\gamma)/2}$ when $m$ becomes larger.
Again we require the probability in \eref{Eq:random_code_probability} being positive and use the inequality $\binom{m}{k}\le 2^{mH(k/m)}$ and $H(k/m)\le 1$, we get the inequality
\begin{align}
	\frac{d^m\alpha^2}{256} \ge (\ln2)m + 2d_C \ln\left( \frac{10}{\alpha} \right). 
\end{align}
Then we consider the asymptotic case $d_C=d^l\gg m\gg 1$, then $d^{(\gamma-1/2)m}\ll 1$ and $\alpha\approx d^{2a(\gamma-1/2)m}/d_C$. Substituting $\alpha\approx d^{2a(\gamma-1/2)m}/d_C$ into the above inequality and taking the logarithm, we get 
\begin{align}
	\label{Eq:aqecc_m_lower3}
	m > 2\text{log}_d(d_C) + \text{log}_d\left( \frac{256}{d^{4a(\gamma-1/2)m}} \right) + \text{log}_d \left[ (\ln2)m + 2d_C \ln\left( \frac{10}{\alpha} \right) \right]
	\approx 3 \text{log}_d (d_C)
\end{align} 
For simplicity, we also define two functions as
\begin{align}
	T_1^{(2)}(m,d_C) &:= \frac{1}{(1-2a)(1/2-\gamma)} l , \\
	T_2^{(2)}(m,d_C) &:= 2\text{log}_d(d_C) + \text{log}_d\left( \frac{256}{d^{4a(\gamma-1/2)m}} \right) + \text{log}_d \left[ (\ln2)m + 2d_C \ln\left( \frac{10}{\alpha} \right) \right].
\end{align}
Given $d_C$, the threshold of $m$ for the existence of a AQECC $\caH_C$ with inaccuracy $\widetilde{\eta}_0$ is given by 
\begin{align}
	m^*:=\min\{m\in \bbN^+: m-T_1^{(2)}(m,d_C)> 0, m-T_2^{(2)}(m,d_C)>0\}.
\end{align}
When $d_C\gg 1$ or $l$ is large, following the same logic, the relative size of $T_1^{(2)}$ and $T_2^{(2)}$ are determined solely by the linear coefficients of the term $\log d_C$, then the linear coefficient of $m^*$ is
\begin{align}
	c = \max\left\{\frac{1}{(1-2a)(1/2-\gamma)}, 3 \right\}.
\end{align}
Because of the parameter $a$, the case is a little more complicated. When $1/2>a\ge 1/6$, $1/((1-2a)(1/2-\gamma))$ is always larger than $3$, then the coefficient $c$ is always $3$. But when $a<1/6$, the line $c=1/((1-2a)(1/2-\gamma))$ will intersect with $c=3$, hence $c$ will become a piecewise function, just as illustrated in \fref{fig:aqecc_case2}.
\begin{figure}
	\centering
	\includegraphics[width=1.0\linewidth]{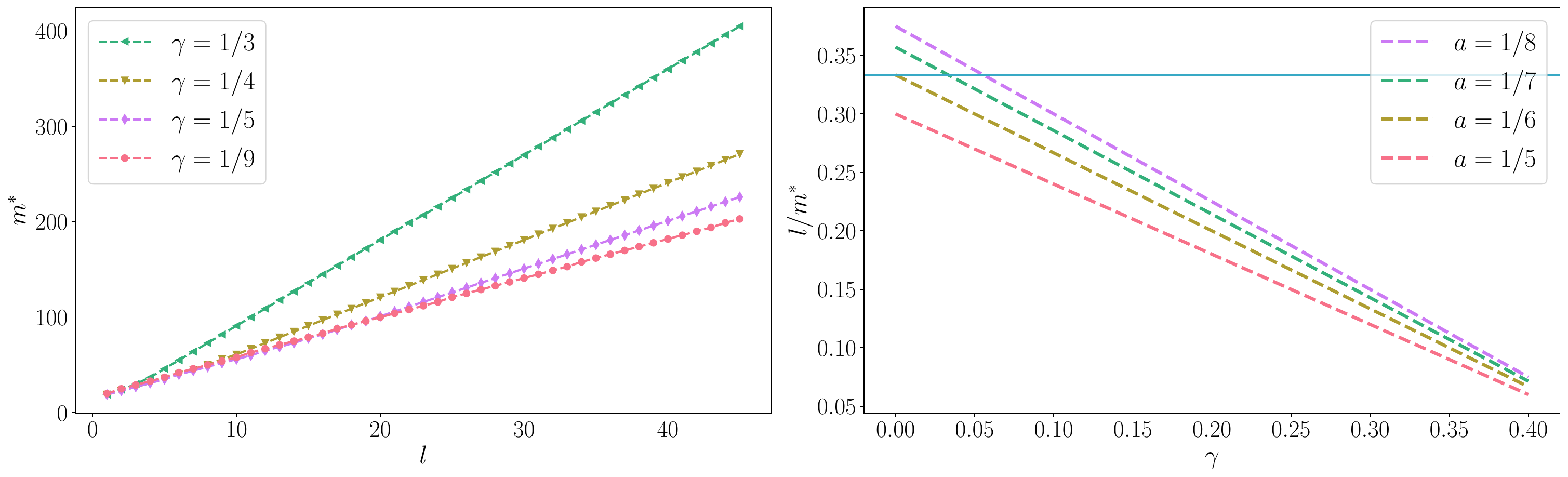}
	\caption{The left plot represents the variations of $m^*$ as a function of $l$ for different $\gamma$, where $a =1/6$. The right plot illustrates the variation of the code rate $l/m^*$ with respect to $\gamma$, given different values of $a$. The horizontal purple line represents $l/m^*=1/3$. }
	\label{fig:aqecc_case2}
\end{figure}

\end{document}